\documentclass[sigplan,screen]{acmart}

\setcopyright{none}
\acmPrice{}
\acmDOI{10.1145/3453483.3454077}
\acmYear{2021}
\copyrightyear{2021}
\acmSubmissionID{pldi21main-p332-p}
\acmISBN{978-1-4503-8391-2/21/06}
\acmConference[PLDI '21]{Proceedings of the 42nd ACM SIGPLAN International Conference on Programming Language Design and Implementation}{June 20--25, 2021}{Virtual, Canada}
\acmBooktitle{Proceedings of the 42nd ACM SIGPLAN International Conference on Programming Language Design and Implementation (PLDI '21), June 20--25, 2021, Virtual, Canada}

\usepackage{xspace}
\usepackage{enumerate}
\usepackage[shortlabels]{enumitem}
\usepackage{subcaption}
\usepackage{mathpartir}
\usepackage{turnstile}
\usepackage{mathtools}
\mathtoolsset{showmanualtags}
\usepackage{stmaryrd}
\usepackage{relsize}
\usepackage{multirow}
\usepackage{multicol}
\usepackage{algorithm}
\usepackage[noend]{algpseudocode}
\usepackage{pseudo}
\pseudoset{indent-length=0.4em}
\usepackage{tikz}
\usetikzlibrary{arrows,automata,chains,shapes,trees}
\usepackage{xfrac}
\usepackage{hhline}
\usepackage{rotating}
\usepackage{wrapfig}
\usepackage{cleveref}
\crefname{theorem}{Thm.}{Thms.}
\crefname{lemma}{Lem.}{Lemmas}
\crefname{corollary}{Cor.}{Cors.}
\crefname{figure}{Fig.}{Figs.}
\crefname{definition}{Defn.}{Defns.}
\crefname{table}{Tab.}{Tabs.}
\crefformat{section}{\S#2#1#3}
\crefmultiformat{section}{\S#2#1#3}{ and~\S#2#1#3}{, \S#2#1#3}{ and~\S#2#1#3}
\crefname{example}{Ex.}{Exs.}
\crefname{item}{item}{items}
\crefname{footnote}{footnote}{footnotes}
\crefname{observation}{Obs.}{Obs.}
\crefname{remark}{Remark}{Remarks}
\crefname{proposition}{Prop.}{Props.}
\crefname{fact}{Fact}{Facts}
\crefname{challenge}{Challenge}{Challenges}
\crefname{principle}{Principle}{Principles}
\usepackage{shortcuts}
\usepackage{balance}

\makeatletter 
\AtEndDocument{%
\def\cb@checkPdfxy#1#2#3#4#5{%
\cb@@findpdfpoint{#1}{#2}%
\ifdim#3sp=\cb@pdfx\relax      
\ifdim#4sp=\cb@pdfy\relax      
\ifdim#5=\cb@pdfz\relax
\else
\cb@error
\fi
\else
\cb@error
\fi
\else
\cb@error
\fi
}}
\makeatother 

\usepackage{color}
\usepackage[color,rightbars]{changebar}
\cbcolor{ACMRed}
\nochangebars

\reversemarginpar

\newif\iflong
\longtrue

\usepackage[cal=cm]{mathalfa}
\newcommand{\mathbbl}[1]{\vvmathbb{#1}}

\usepackage{pifont}
\newcommand{\cmark}{\ding{51}}
\newcommand{\xmark}{\ding{55}}

\newtheoremstyle{acmremark}%
  {.5\baselineskip\@plus.2\baselineskip\@minus.2\baselineskip}
  {.5\baselineskip\@plus.2\baselineskip\@minus.2\baselineskip}
  {\itshape}
  {\parindent}
  {\itshape}
  {.}
  {.5em}
  {\thmname{#1}\thmnumber{ #2}\thmnote{ {(#3)}}}
\AtEndPreamble{%
  \theoremstyle{acmplain}%
  \newtheorem*{theorem*}{Theorem}%
  \newtheorem*{lemma*}{Lemma}%
  \newtheorem*{proposition*}{Proposition}%
  \newtheorem*{corollary*}{Corollary}%
  \theoremstyle{acmdefinition}%
  \theoremstyle{acmremark}%
  \newtheorem{remark}[theorem]{Remark}%
  \theoremstyle{acmplain}%
}

\AtBeginDocument{%
  \setlength\abovedisplayskip{0.5\abovedisplayskip}%
  \setlength\belowdisplayskip{0.5\belowdisplayskip}%
  \setlength\abovedisplayshortskip{0.5\abovedisplayshortskip}%
  \setlength\belowdisplayshortskip{0.5\belowdisplayshortskip}%
  \setlength\floatsep{0.5\floatsep}%
  \setlength\abovecaptionskip{0.5\abovecaptionskip}%
}
\setlist{nosep,leftmargin=\parindent}

\makeatletter
\renewcommand{\paragraph}{%
  \@startsection{paragraph}{4}%
  {\z@}{-.1\baselineskip \@plus -2\p@ \@minus -.2\p@}{-3.5\p@}%
  {\@parfont\@adddotafter}%
}
\makeatother

\newcommand{\dotsim}{%
  \mathbin{\textnormal{%
    \mathsurround=0pt
    \ooalign{%
      \hidewidth\vphantom{$\div$}\raisebox{.95ex}{\scalebox{1.2}{.}}\hidewidth\cr
      $\sim$\cr
      \hidewidth\vphantom{$\div$}\raisebox{-.05ex}{\scalebox{1.2}{.}}\hidewidth\cr
    }%
  }}%
}

\newcommand{\bind}{\mathbin{\gg\!=}}
\newcommand{\many}[1]{\overrightarrow{#1}}
\newcommand{\Gm}{\Gamma}
\newcommand{\Sg}{\Sigma}
\newcommand{\m}[1]{\mathsf{#1}}
\newcommand{\dplus}{\mathbin{{+}\!\!{+}}}

\newcommand{\kif}{\kw{if}\xspace}
\newcommand{\kthen}{\kw{then}\xspace}
\newcommand{\kelse}{\kw{else}\xspace}

\newcommand{\ksample}[1]{\ensuremath{\kw{sample}(#1)}}

\newcommand{\kgets}{\ensuremath{\gets}\xspace}

\newcommand{\kproc}{\kw{proc}\xspace}
\newcommand{\kreturn}[1]{\ensuremath{\kw{return}(#1)}}
\newcommand{\kcall}{\kw{call}\xspace}
\newcommand{\ksamplei}[2]{\ensuremath{\kw{sample}_\m{rv}\{\id{#2}\}(#1)}}
\newcommand{\ksampleo}[2]{\ensuremath{\kw{sample}_\m{sd}\{\id{#2}\}(#1)}}

\newcommand{\kifi}[1]{\ensuremath{\kw{if}_\m{rv}\{\id{#1}\}}\xspace}
\newcommand{\kifo}[1]{\ensuremath{\kw{if}_\m{sd}\{\id{#1}\}}\xspace}

\newcommand{\klsample}[2]{\ensuremath{\kw{sample}(@{#2},#1)}}

\newcommand{\eabs}[3]{\ensuremath{\lambda(#1.#3)}}
\newcommand{\eapp}[2]{\ensuremath{\m{app}(#1;#2)}}
\newcommand{\etriv}{\ensuremath{\m{triv}}}

\newcommand{\etrue}{\ensuremath{\m{true}}}
\newcommand{\efalse}{\ensuremath{\m{false}}}
\newcommand{\econd}[3]{\ensuremath{\m{if}(#1;#2;#3)}}

\newcommand{\elet}[3]{\ensuremath{\m{let}(#1; #2.#3)}}
\newcommand{\ebinop}[3]{\ensuremath{\m{op}_{#1}(#2;#3)}}
\newcommand{\eber}[1]{\ensuremath{\cn{Ber}(#1)}}
\newcommand{\eunif}{\ensuremath{\cn{Unif}}}
\newcommand{\ebeta}[2]{\ensuremath{\cn{Beta}(#1;#2)}}
\newcommand{\egamma}[2]{\ensuremath{\cn{Gamma}(#1;#2)}}
\newcommand{\enormal}[2]{\ensuremath{\cn{Normal}(#1;#2)}}
\newcommand{\ecat}[1]{\ensuremath{\cn{Cat}(#1)}}
\newcommand{\egeo}[1]{\ensuremath{\cn{Geo}(#1)}}
\newcommand{\epois}[1]{\ensuremath{\cn{Pois}(#1)}}

\newcommand{\mret}[1]{\ensuremath{\m{ret}(#1)}}
\newcommand{\mbnd}[3]{\ensuremath{\m{bnd}(#1;#2.#3)}}

\newcommand{\msamplei}[2]{\ensuremath{\m{sample}_\m{rv}\{\id{#2}\}(#1)}}
\newcommand{\msampleo}[2]{\ensuremath{\m{sample}_\m{sd}\{\id{#2}\}(#1)}}

\newcommand{\mbranchi}[3]{\ensuremath{\m{cond}_\m{rv}\{\id{#3}\}(#1;#2)}}
\newcommand{\mbrancho}[4]{\ensuremath{\m{cond}_\m{sd}\{\id{#4}\}(#1;#2;#3)}}

\newcommand{\mcall}[2]{\ensuremath{\m{call}(#1;#2)}}

\newcommand{\fundec}[7]{\ensuremath{\m{fix}\{#6 ; #7 \}(#1.#4.#5)}}

\newcommand{\tbool}{\ensuremath{\mathbbl{2}}}
\newcommand{\tnat}{\ensuremath{\mathbbl{N}}}
\newcommand{\treal}{\ensuremath{\mathbbl{R}}}
\newcommand{\tureal}{\ensuremath{\mathbbl{R}_{(0,1)}}}
\newcommand{\tpreal}{\ensuremath{\mathbbl{R}_+}}
\newcommand{\tunit}{\ensuremath{\mathbbl{1}}}
\newcommand{\tdist}[1]{\ensuremath{\m{dist}(#1)}}

\newcommand{\one}{\ensuremath{\pmb{1}}}
\newcommand{\ichoice}{\ensuremath{\mathbin{\varoplus}}}
\newcommand{\echoice}{\ensuremath{\mathbin{\binampersand}}}

\newcommand{\msg}[2]{\ensuremath{\mathfrak{msg}(#1.#2)}}

\newcommand{\chtype}[2]{\ensuremath{(\id{#1}\!: #2)}}
\newcommand{\evalp}[1]{\ensuremath{\mathrel{\Downarrow^{#1}}}}
\newcommand{\evalto}{\ensuremath{\mathrel{\Downarrow}}}
\newcommand{\concat}{\ensuremath{\mathbin{\dplus}}}

\newcommand{\msgobjl}[1]{\ensuremath{\textbf{\textit{val}}^{\m{P}}(#1)}}
\newcommand{\msgobjr}[1]{\ensuremath{\textbf{\textit{val}}^{\m{C}}(#1)}}

\newcommand{\dirobjl}[1]{\ensuremath{\textbf{\textit{dir}}^{\m{P}}(#1)}}
\newcommand{\dirobjr}[1]{\ensuremath{\textbf{\textit{dir}}^{\m{C}}(#1)}}
\newcommand{\foldobj}{\ensuremath{\textbf{\textit{fold}}}}
\newcommand{\vclo}[2]{\ensuremath{\m{clo}(#1, #2)}}

\newcommand{\red}{\mathrel{\vdash_\m{red}}}

\begin{document}

\title{Sound Probabilistic Inference via Guide Types}
\iflong
\subtitle{Technical Report}
\fi


\author{Di Wang}
\affiliation{
  \institution{Carnegie Mellon University}
  \country{USA}
}

\author{Jan Hoffmann}
\affiliation{
  \institution{Carnegie Mellon University}
  \country{USA}
}

\author{Thomas Reps}
\affiliation{
  \institution{University of Wisconsin}
  \country{USA}
}


\begin{abstract}
Probabilistic programming languages aim to describe and automate
Bayesian modeling and inference.
Modern languages support \emph{programmable inference}, which allows
users to customize inference algorithms by incorporating \emph{guide}
programs to improve inference performance.
For Bayesian inference to be sound, guide programs must be compatible
with model programs.
%
%
One pervasive but challenging condition for model-guide compatibility
is \emph{absolute continuity}, which requires that the model and guide
programs define probability distributions with the same support.

This paper presents a new probabilistic programming language that
\emph{guarantees} absolute continuity, and features general programming
constructs, such as branching and recursion.
Model and guide programs are implemented as \emph{coroutines} that
communicate with each other to synchronize the set of random
variables they sample during their execution.
Novel \emph{guide types} describe and enforce communication protocols between
coroutines. If the model and guide are well-typed using the same protocol,
then they are guaranteed to enjoy absolute continuity.
An efficient algorithm infers guide types from code so that users do
not have to specify the types.
%
The new programming language is evaluated with an implementation that
includes the type-inference algorithm and a prototype compiler that
targets Pyro.
Experiments show that our language is capable of expressing a variety
of probabilistic models with nontrivial control flow and recursion,
and that the coroutine-based computation does not introduce
significant overhead in actual Bayesian inference.
\end{abstract}



\maketitle


%
%
%

\section{Introduction}
\label{Se:Introduction}

%
\emph{Probabilistic programming languages} (PPLs)~%
\cite{misc:dippl,JSS:CGH17,ICLR:THS17,MAPL:AAD19,JRSS:GTS97,DSC:Plummer03,UAI:GMR08,AISTATS:WMM14,HASKELL:SGG15}
provide a flexible way of describing statistical models and
automatically performing Bayesian inference: a method for inferring the
posterior of a statistical model from observed data.
Bayesian inference accounts for uncertainty in latent variables
that produce the observed data.
It has applications in many fields, including artificial
intelligence~\cite{NATURE:Ghahramani15}, cognitive science~\cite{kn:GKT08}, and
applied statistics~\cite{book:GCS13}.

%
Because there is not a single known inference algorithm
that works well for all models~\cite{PLDI:MSH18},
several PPLs have recently added support for \emph{programmable inference}%
~\cite{PLDI:MSH18,PLDI:CSL19,JMLR:BCJ18,AISTATS:GXG18,UAI:ZS17,JSS:Murray15}.
This capability allows users to customize inference algorithms based on the characteristics
of a particular model or dataset.
Researchers have shown that programmable inference enables improved inference performance on
a variety of modeling problems~\cite{PLDI:MSH18,PLDI:CSL19,JMLR:BCJ18,NIPS:FJB19}.

%
Two important families of inference algorithms can be customized by incorporating
\emph{guide} programs, which are implemented by the user.
The first family is \emph{Monte-Carlo} methods, such as importance sampling
and Markov-Chain Monte Carlo,
where a guide program serves as a \emph{proposal}, which generates random samples
for latent variables.
The second family is \emph{variational inference}, where a guide program is
a parameterized program that specifies a collection of \emph{approximating}
distributions on latent variables.

%
To ensure soundness of programmable inference, the guide programs have to
be \emph{compatible} with the implemented model program;
incompatible guide programs could crash the inference process or lead to
incorrect inference results~\cite{POPL:LCS20,POPL:LYR19}.
Recently, \citet{POPL:LYR19} developed a static analysis for finding bugs
in model-guide pairs for variational inference in Pyro~\cite{JMLR:BCJ18}.
\citet{POPL:LCS20} proposed a type system that proves model-guide compatibility
for multiple inference algorithms.
\begin{changebar}
However, neither approach handles general conditional statements
that can influence the set of latent variables sampled by the model, and it
is unclear how to extend them to analyze recursive programs precisely.
\end{changebar}

%
In this paper, we develop a new PPL that supports recursion and conditional statements,
as well as guarantees \emph{absolute continuity},
one of the most pervasive conditions for ensuring model-guide compatibility.
Our PPL uses a new paradigm for writing inference code: users implement the model
and guide programs as \emph{coroutines}, which can communicate with each other during their execution.
We develop a new type system, which we dub \emph{guide types}, to describe the
communication protocols between coroutines.
These guide types can be automatically inferred and are proof certificates
of absolute continuity for model-guide pairs.
They apply to multiple kinds of Bayesian-inference algorithms.

%
In our development, we follow a common scheme of \emph{trace-based} programmable inference
that underlies Pyro~\cite{JMLR:BCJ18}, Venture~\cite{PLDI:MSH18}, Gen~\cite{PLDI:CSL19}, etc.
These PPLs define the meaning of a probabilistic program by a probability distribution on
\emph{sample traces} that record all the random samples that the program draws during its
execution.
A program $p$ is \emph{absolutely continuous} with respect to a program $q$,
if any set of sample traces with non-zero probability under the program $p$ must also have non-zero
probability under the program $q$.
In this paper, we reduce the problem of checking absolute continuity to the following verification
task:

\emph{
  Given a model program $p$ and a guide program $q$, verify that they define probability
  distributions with the same \emph{support}, i.e., they have the same set of possible sample
  traces.
}

%
The major challenge in our development is to reason about the sets of possible sample
traces for the model and guide programs, when the two programs can diverge in their
execution, as always with relational reasoning. 
Control-flow constructs make it difficult to keep track of sample sites precisely;
for example, a conditional statement can sample different sets of random variables
in its two branches.
It is intractable to enumerate all possible execution paths in the two programs and compare
the sample sites path-to-path, especially when the programs are recursive.

%
The first part of our solution is to think of the model and guide programs as \emph{coroutines}
that can exchange messages.
Conceptually, we use coroutine-style communication to \emph{synchronize} each pair of sample sites
that represent the same random variable, as well as each branch selection that influences
control flow.
%
%
%
The communication between the two coroutines should then be conducted according to a protocol
so that messages always occur in \emph{guidance} pairs:
when one partner sends, the other receives; and
when one partner offers a selection, the other branches.


The second part of our solution is to develop \emph{guide types} as
guidance protocols between the model and guide coroutines.
%
%
In our formalization, we \emph{structure} the sequence of messages between two coroutines,
rather than describe it as a collection of unrelated messages.
%
%
To handle general recursion, we parameterize the guide type for each coroutine by a
\emph{continuation} type that describes the guidance protocol for the computation that
continues after a recursive invocation.
We also develop an efficient algorithm that infers guide types automatically from the code.

There have been several type systems for coroutines~\cite{book:PFPL16,TFP:AT10,APLAS:AT10},
but all of them require that all messages from a coroutine to another have the same type;
thus, they are not sufficient to handle \emph{sample passing} and \emph{branch selection}
in our coroutine-based paradigm.
In our development of guide types, we took inspiration from
type systems for communication protocols in concurrent systems, such as
\emph{session types}~\cite{CONCUR:Honda93,ESOP:HVK98}.
Guide types have different semantics from and are simpler than session types,
and use a parametrization technique to model recursive computation.

%
We then establish formal guarantees of our new PPL.
First, we prove that guide types ensure \emph{safety} of communication between coroutines,
i.e., the coroutines send and receive messages in a consistent manner.
Second, we prove that guide types serve as proof certificates of absolute
continuity between the model and guide programs; consequently, we use guide
types to justify soundness of importance sampling, Markov-Chain Monte Carlo,
and variational inference.
Note that for variational inference, the soundness guarantee is \emph{partial},
because sound inference requires some additional conditions (e.g., differentiability), whereas
this paper focuses just on absolute continuity.

%
We implemented a type-inference algorithm for guide types and a prototype
compiler from our PPL to Pyro.
We evaluated our PPL on a broad suite of probabilistic models, and our
experimental results show that
(i) our PPL is more expressive than a state-of-the-art PPL that ensures
soundness of programmable inference~\cite{POPL:LCS20}, and
(ii) type inference completes in several milliseconds, and the performance
of Bayesian inference on the compiled code is similar to handwritten Pyro
code, i.e., coroutine communication does not introduce significant overhead.

%
\paragraph{Contributions}
We make four main contributions.
\begin{itemize}
  \item
  We develop a new PPL with a coroutine-based paradigm for implementing model and guide
  programs.
  
  \item
  We propose guide types, which
  prescribe guidance protocols between the model and guide coroutines,
  and develop an efficient inference algorithm for guide types.
  
  \item
  We prove type safety of guide types, and show that guide types ensure key soundness
  conditions of model-guide pairs for multiple kinds of Bayesian-inference algorithms. 
  
  \item
  We implemented our PPL and evaluated its effectiveness on a variety of
  probabilistic models.
\end{itemize}


\section{Overview}
\label{Se:Overview}

In this section, we first review Bayesian inference and trace-based programmable inference
(\cref{Se:Overview:Inference}).
We then demonstrate the coroutine-based paradigm for implementing inference code
and the use of guide types to enforce guidance protocols between coroutines.
(\cref{Se:Overview:Coroutines}).

\subsection{Bayesian Inference}
\label{Se:Overview:Inference}

Probabilistic programs specify generative models that sample random variables.
The semantics of a probabilistic program can be defined as a probability distribution on the
\emph{sample traces} that record all the random values that a program draws during its
execution~\cite{ICFP:BLG16,JCSS:Kozen81}.
Consider the program $\id{Model}$ in \cref{Fig:ToyModel};
it specifies a probabilistic model on random variables introduced by commands
\klsample{d}{\ell}, where $\ell$ is a \emph{label} that identifies a sample
site in a program; and $d$ is a \emph{primitive distribution}, such as \cn{Gamma}
distributions whose support is the positive real line $\tpreal$, \cn{Normal} distributions
whose support is the real line $\treal$, and \cn{Beta} distributions whose support is
the unit interval $\tureal$.
Two possible sample traces in the program $\id{Model}$ are $[ @x=1; @z={-0.5} ]$ and $[ @x = 3; @y = 0.9; @z = 0.7 ]$.
More generally, the program specifies a distribution on sample traces whose support is
\begin{equation}\label{Eq:ToyModelSupport}
  \begin{split}
  & \{ [ @x = a; @z = c ] \mid 0 < a < 2 \} \\
  {}\cup{} & \{ [ @x = a; @y = b; @z = c ] \mid a \ge 2, 0 < b < 1 \}.
  \end{split}
\end{equation}

\begin{figure}
\centering
\begin{small}
\begin{pseudo}
  \kproc $\id{Model}$() = \\+
    $v$ \kgets \klsample{\egamma{2}{1}}{x}; \\
    \kif $v < 2$ \kthen \\+
      \_ \kgets \klsample{\enormal{-1}{1}}{z}; \\
      \kreturn{v} \\-
    \kelse \\+
      $m$ \kgets \klsample{\ebeta{3}{1}}{y}; \\
      \_ \kgets \klsample{\enormal{m}{1}}{z}; \\
      \kreturn{v}
\end{pseudo}
\end{small}
\caption{A program $\id{Model}$ with a conditional statement.}
\label{Fig:ToyModel}
\end{figure}

\begingroup
\setlength{\columnsep}{7pt}%
Bayesian Inference amounts to conditioning a probabilistic model on \emph{observations}
and computing a posterior distribution on \emph{latent variables}.
For the program $\id{Model}$, we consider that $@z$ is the single $\treal$-valued observation,
while both $@x$ and $@y$ are latent variables.
Intuitively, latent variables encode knowledge about the ``ground truth'' that we cannot
observe directly, and the model program specifies a \emph{prior} distribution on the
``ground truth.''
\begin{wrapfigure}{r}{3.5cm}
\centering
\includegraphics[width=3.5cm]{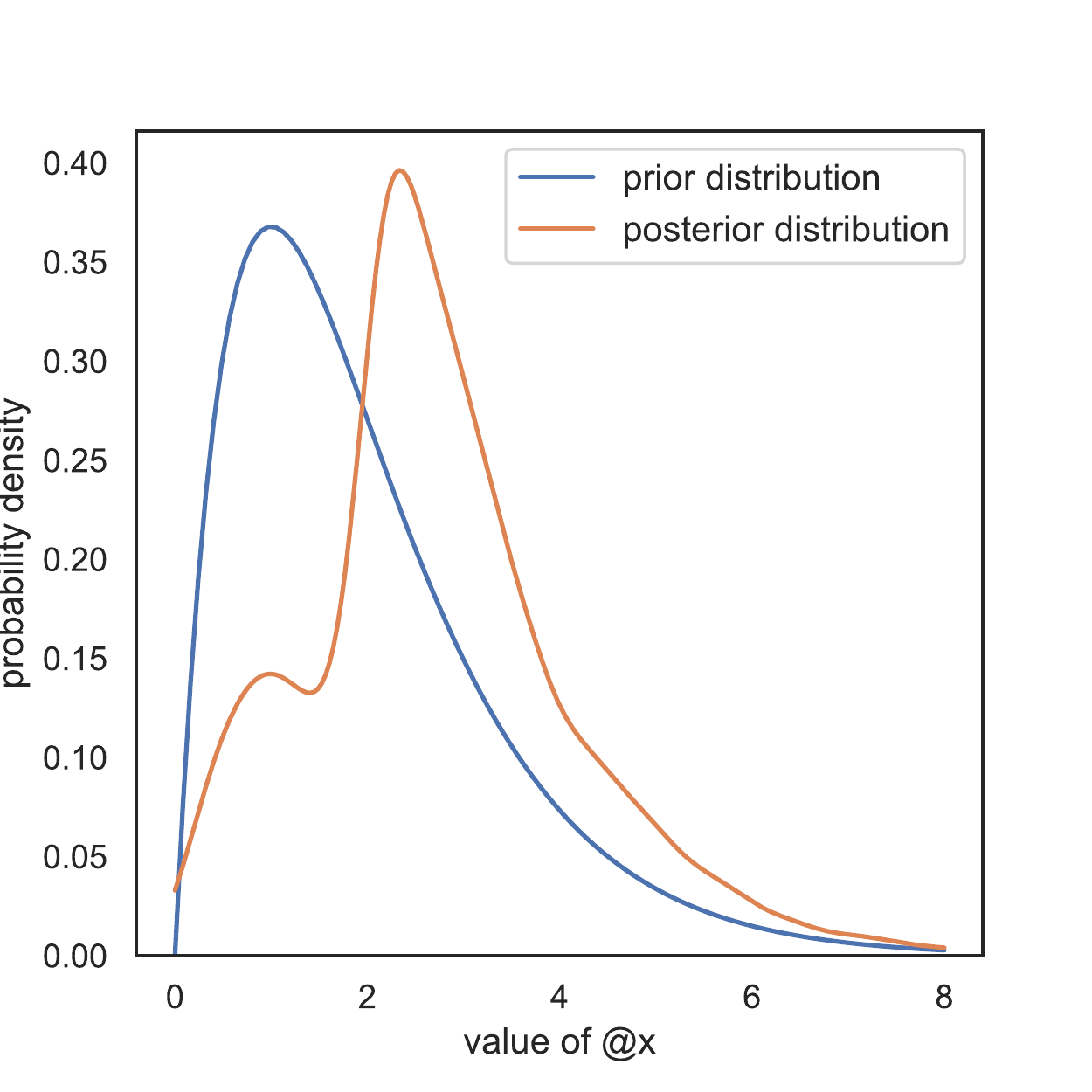}
\caption{\small Probability densities of the prior and posterior distribution of the random variable $@x$.}
\label{Fig:Curves}
\end{wrapfigure}
Given a concrete value of the observation (e.g., $@z = 0.8$), the objective of
Bayesian inference is to approximate the \emph{posterior} distribution of the latent variables
(e.g., likely values of $@x$ and $@y$ under the condition that $@z = 0.8$).
\cref{Fig:Curves} plots the prior distribution of the random variable $@x$,
and its posterior distribution under the observation $@z = 0.8$.

It is usually intractable to sample directly from or even derive posterior distributions.
There have been two popular families of inference algorithms:
\emph{Monte-Carlo methods} and \emph{variational inference}.
These inference algorithms usually require some \emph{guide} programs, which can have
a substantial influence on the performance of the inference.
Although many PPLs provide mechanisms for automatically generating those guide programs, the ability
to allow users to \emph{customize} them, has been shown to be helpful, and sometimes crucial, for effective
inference~\cite{PLDI:MSH18,PLDI:CSL19,JMLR:BCJ18,NIPS:FJB19}.
However, customizability introduces non-trivial challenges to ensuring \emph{soundness} of
Bayesian inference.
We now illustrate some mistakes when programming guide programs for Monte-Carlo methods and
for variational inference.

\endgroup

\paragraph{Monte-Carlo methods}
A Monte-Carlo method generates iteratively random samples
such that empirical distribution of the samples approximates the posterior distribution.
Two popular Monte-Carlo methods are \emph{importance sampling} (IS) and
\emph{Markov-Chain Monte Carlo} (MCMC).
IS generates independent and identically distributed samples from a \emph{proposal} distribution, and reweights
the samples by their \emph{importance}, which corrects the discrepancy
between the posterior and proposal distributions.
MCMC generates iteratively a new random sample from an old one; that is,
it constructs a Markov chain whose stationary distribution is the posterior
distribution.

We now illustrate a mistake when programming guide programs for IS. 
For IS to converge asymptotically to the posterior distribution, the posterior distribution must be
\emph{absolutely continuous} with respect to the proposal distribution, i.e., any set of samples
with non-zero probability under the posterior distribution must also have non-zero probability
under the proposal distribution.
In \cref{Se:Inference}, we will show that it suffices to verify if the model program \emph{conditioned}
with respect to a concrete observation and the guide program have the same set of possible
sample traces.
For example, for the program $\id{Model}$ shown in \cref{Fig:ToyModel}, the support of a sound guide
program could be
\begin{equation}\label{Eq:ToyProposalSupport}
  \begin{split}
    & \{ [ @x = a ] \mid 0 < a < 2 \} \\
    {}\cup{} & \{ [@x = a; @y = b] \mid a \ge 2, 0 < b < 1 \},
  \end{split}
\end{equation}
which is obtained by factoring out the observation $@z$ from the support of the unconditioned
model shown in \eqref{Eq:ToyModelSupport}.

\cref{Fig:ToyIS} presents two guide programs for performing IS from
the program $\id{Model}$ shown in \cref{Fig:ToyModel}, where the supports of the \cn{Pois} and
\cn{Unif} distributions are natural numbers $\tnat$ and the unit interval $\tureal$, respectively.
The support of the program $\id{Guide}_1$ is exactly the one shown in \eqref{Eq:ToyProposalSupport};
thus, $\id{Guide}_1$ is a sound guide program;
that is, $\id{Guide}_1$ samples the latent random variables $@x,@y$ from the same space as $\id{Model}$ does.
On the other hand, the support of the program $\id{Guide}_1'$ does not match \eqref{Eq:ToyProposalSupport},
and it is actually an \emph{unsound} guide program for two reasons:
\begin{itemize}
  \item
  In the model program, the latent variable $@x$ can be any positive real number,
  whereas the program $\id{Guide}_1'$ only samples natural numbers for $@x$.
  
  \item
  In the model program, when the value of $v$ (i.e., the latent variable $@x$) is greater than $2$,
  the other latent variable $@y$ should be present in the sample trace.
  However, when the value of $v$ is greater than $10$, the program $\id{Guide}_1'$ will not produce a sample
  for $@y$.
\end{itemize}

\begin{figure}
\begin{subfigure}{0.57\columnwidth}
\centering
\caption*{A \emph{sound} guide}
\begin{small}
\begin{pseudo}
  \kproc $\id{Guide}_1$() = \\+
    $v$ \kgets \klsample{\egamma{1}{1}}{x}; \\
    \kif $v < 2$ \kthen \\+
      \kreturn{} \\-
    \kelse \\+
      \_ \kgets \klsample{\eunif}{y}; \\
      \kreturn{}
\end{pseudo}
\end{small}
\end{subfigure}
\begin{subfigure}{0.42\columnwidth}
\centering
\caption*{An \emph{unsound} guide}
\begin{small}
\begin{pseudo}
  \kproc $\id{Guide}_1'$() = \\+
    $v$ \kgets \klsample{\epois{4}}{x}; \\
    \kif $v > 10$ \kthen \\+
      \kreturn{} \\-
    \kelse \\+
      \_ \kgets \klsample{\eunif}{y}; \\
      \kreturn{}
\end{pseudo}
\end{small}
\end{subfigure}
\caption{Sound and unsound guide programs for IS.}
\label{Fig:ToyIS}
\end{figure}

\paragraph{Variational inference (VI)}
In contrast to Monte-Carlo methods, VI uses \emph{optimization} (e.g., stochastic gradient descent)
to find a candidate from an approximating family of distributions that minimizes the distance
between the posterior distribution and the approximating distributions.
In PPLs such as Pyro, users specify the approximating family by a \emph{parameterized} probabilistic
program called a \emph{guide}; instantiating the parameters with a concrete valuation that produces a member
of the approximating family.
A widely used distance is the Kullback-Leibler (KL) divergence from the posterior
distribution to the guide distribution.
For the KL divergence to be well-defined, the guide distribution must be \emph{absolutely continuous}
with respect to the posterior distribution.
In \cref{Se:Inference}, we again reduce the verification of absolute continuity to checking a
sufficient condition, namely, that the model conditioned with respect to a concrete observation
and the guide have the same support.
Note that VI requires several more conditions (such as differentiability) for
inference to be sound~\cite{POPL:LYR19}.
In this paper, we focus on verification of absolute continuity.

\cref{Fig:ToyVI} presents two guide programs for performing VI on the program $\id{Model}$ shown
in \cref{Fig:ToyModel}.
The real-valued parameters of the guide programs are $\theta_1,\ldots,\theta_4$.
The support of the program $\id{Guide}_2$ (instantiated with concrete parameters) is exactly the one
shown in \cref{Eq:ToyProposalSupport}.
On the other hand, the program $\id{Guide}_2'$ defines an \emph{unsound} guide, because it samples $@x$
from a normal distribution, whose support is the whole real line, whereas the program $\id{Model}$
always samples a positive value for $@x$.

\begin{figure}
\begin{subfigure}{0.49\columnwidth}
\centering
\caption*{A \emph{sound} guide}
\begin{small}
\begin{pseudo}
  \kproc $\id{Guide}_2$($\theta_1$,$\theta_2$,$\theta_3$,$\theta_4$) = \\+
    $v$ \kgets $\kw{sample}(@x,$ \\+++++++
       $\egamma{\theta_1}{\theta_2})$; \\-------
    \kif $v < 2$ \kthen \\+
      \kreturn{} \\-
    \kelse \\+
      \_ \kgets $\kw{sample}(@y,$ \\+++++++
        $\ebeta{\theta_3}{\theta_4})$; \\-------
      \kreturn{}
\end{pseudo}
\end{small}
\end{subfigure}
\begin{subfigure}{0.49\columnwidth}
\centering
\caption*{An \emph{unsound} guide}
\begin{small}
\begin{pseudo}
  \kproc $\id{Guide}_2'$($\theta_1$,$\theta_2$) = \\+
    $v$ \kgets $\kw{sample}(@x,$ \\+++++++
      $\enormal{\theta_1}{\theta_2})$; \\-------
    \kif $v < 2$ \kthen \\+
      \kreturn{} \\-
    \kelse \\+
      \_ \kgets $\kw{sample}(@y,$ \\+++++++
        $\eunif)$; \\-------
      \kreturn{}
\end{pseudo}
\end{small}
\end{subfigure}
\caption{Sound and unsound guide programs for VI.}
\label{Fig:ToyVI}
\end{figure}

\subsection{Sound Bayesian Inference via Guide Types}
\label{Se:Overview:Coroutines}

\paragraph{Programs as coroutines}
Our first contribution is a \emph{coroutine}-based paradigm for implementing
the model and guide programs for Bayesian inference.
In an inference algorithm, the model program and its guide program have many connections.
The two most significant patterns we can observe in common inference algorithms are as follows:
\begin{itemize}
  \item
  The guide program is used to generate sample traces, and then the model program
  is simulated with these traces to compute likelihoods.
  
  \item
  The guide program needs to have similar control-flow structure to that of the model program.
  For example, if the model program has a conditional command whose two branches sample different
  sets of latent variables, the guide program should also have a conditional command with an
  equivalent branch condition.
\end{itemize}
The first pattern illustrates a form of \emph{sample passing} from the guide program to
the model program, and the second pattern indicates that the model program should provide
\emph{branch selection} to the guide program.
Such \emph{bidirectional guidance} inspired us to treat the model and guide programs as
\emph{coroutines} that communicate with each other during their execution, rather than as
totally independent programs.
On the other hand, we do \emph{not} want the coroutines to be tightly coupled: Bayesian practitioners
usually maintain a separation between the model and the guide so that they can refine the guide
iteratively to improve inference performance.

Therefore, we use \emph{message-passing communication} to implement the coroutines;
this formalism allows us to separate the model and the guide as individual programs,
but connect them via \emph{channels} over which coroutines exchange messages.
\cref{Fig:ProgAsCoro} reimplements the model and guide programs in \cref{Fig:ToyModel} and
\cref{Fig:ToyIS}, respectively, by making the guidance communication explicit.
The \ksample{\cdot} commands and conditional commands are annotated with $\m{rv}$ (i.e., ``receive'')
or $\m{sd}$ (i.e., ``send'') to indicate the direction of communication, and associated with
a name of the channel on which the communication is carried out.
In this example, we use two channels:
$\id{latent}$ for communication between the guide and the model,
and $\id{obs}$ for identifying observations in the model.
Every channel has a unique provider and a unique consumer.
Note that in this way we do \emph{not} need to use labeled samples---as Pyro and some other PPLs
do---because the sampling sites are synchronized through guidance communication.

Operationally, when a coroutine is executing a command associated with a channel $c$, it resumes
the other coroutine that accesses channel $c$, until the other coroutine encounters a command
that also communicates on channel $c$.
Then they perform \emph{synchronization}; for example,
\begin{itemize}
  \item
  When $\id{Model}$ executes $\ksamplei{\egamma{2}{1}}{latent}$,
  it resumes the other end of the \id{latent} channel, i.e., the coroutine $\id{Guide}_1$,
  until $\id{Guide}_1$ reaches the command $\ksampleo{\egamma{1}{1}}{latent}$.
  Recall that the guide program is used in importance sampling;
  thus, the coroutine $\id{Guide}_1$ draws a sample from the distribution $\egamma{1}{1}$,
  and then sends it to the coroutine $\id{Model}$, which uses the sample and the prior
  distribution $\egamma{2}{1}$ to calculate the importance weight.

  \item
  When $\id{Guide}_1$ executes the conditional command on line 3 (where the $\star$ symbol
  indicates that the branch selection is received from the other coroutine),
  it resumes the other end of the \id{latent} channel, i.e., the coroutine $\id{Model}$,
  until $\id{Model}$ reaches the conditional command on line 3.
  The coroutine $\id{Model}$ is the sender of the branch selection;
  thus, it evaluates the branch predicate $v<2$, and sends the result back to $\id{Guide}_1$.  
\end{itemize}
When the synchronization is completed, either coroutine can continue to execute.

\begin{figure}
\centering
\begin{subfigure}{\columnwidth}
\centering
\caption*{Model}
\begin{small}
\begin{pseudo}
  \kproc $\id{Model}$() \kw{consume} \id{latent} \kw{provide} \id{obs} = \\+
    $v$ \kgets \ksamplei{\egamma{2}{1}}{latent}; \\
    \kifo{latent} $v < 2$ \kthen \\+
      \_ \kgets \ksampleo{\enormal{-1}{1}}{obs}; \\
      \kreturn{v} \\-
    \kelse \\+
      $m$ \kgets \ksamplei{\ebeta{3}{1}}{latent}; \\
      \_ \kgets \ksampleo{\enormal{m}{1}}{obs}; \\
      \kreturn{v}
\end{pseudo}  
\end{small}
\end{subfigure}
\begin{subfigure}{\columnwidth}
\centering
\caption*{Guide}
\begin{small}
\begin{pseudo}
  \kproc $\id{Guide}_1$() \kw{consume} . \kw{provide} \id{latent} = \\+
    $v$ \kgets \ksampleo{\egamma{1}{1}}{latent}; \\
    \kifi{latent} $\star$ \kthen \\+
      \kreturn{} \\-
    \kelse \\+
      $\_$ \kgets \ksampleo{\eunif}{latent}; \\
      \kreturn{}
\end{pseudo}  
\end{small}
\end{subfigure}
\caption{Probabilistic programs as coroutines.}
\label{Fig:ProgAsCoro}  
\end{figure}

\paragraph{Guide types}
Our second contribution is \emph{guide types} that enforce guidance protocols between coroutines,
and an efficient algorithm that infers guide types from code.

We take inspiration from type systems for communication protocols in concurrent systems, such
as \emph{session types}~\cite{CONCUR:Honda93,ESOP:HVK98}.
The key idea is to \emph{structure} the sequence of guidance messages on a channel, rather than
describe it as a collection of unrelated messages.

We sketch some type constructors in our development of guide types.
The type $\one$ types an ended channel, where no messages can be exchanged.
The type $A \echoice B$ types a channel whose provider waits for a branch selection,
and continues with a protocol of type $A$ or a protocol of type $B$ based on the received selection.
The type $\tau \wedge A$ types a channel whose provider samples and sends a random value
of type $\tau$, and then continues with a type $A$ protocol.
The guide type for a channel is the same for the provider and the consumer of the channel,
but the two ends of a channel interpret the guide type for the channel \emph{dually}
(e.g., sends as receives).

With these three type constructors, we can express the protocols for the \id{latent} and \id{obs}
channels shown in \cref{Fig:ProgAsCoro} as
\begin{align}
  \id{latent} & \!: \tpreal \wedge (\one \echoice (\tureal \wedge \one)), \label{Eq:LatentType} \\
  \id{obs} & \!: \treal \wedge \one. \label{Eq:ObsType}
\end{align}

The provider and the consumer of the channel \id{latent} are the coroutines $\id{Guide}_1$ and
$\id{Model}$, respectively.
From the provider $\id{Guide}_1$'s perspective, the protocol shown as type \eqref{Eq:LatentType}
guides $\id{Guide}_1$ to draw a $\tpreal$-valued sample and send it on \id{latent}, then wait for
a branch selection, and finally end the communication on \id{latent} if the received branch
selection is then-branch, otherwise draw an $\tureal$-valued sample before ending the communication.
The coroutine $\id{Guide}_1$ implements this guidance protocol exactly.
Meanwhile, from the consumer $\id{Model}$'s perspective, the type constructors have
\emph{dual} semantics, i.e.,
\textsl{send} becomes \textsl{receive} and vice versa;
thus, the protocol for \id{latent} guides $\id{Model}$ to receive an $\tpreal$-valued sample, and then send out a
branch selection on channel \id{latent}; if $\id{Model}$ selects the else-branch,
then it further receives an $\tureal$-valued sample on channel \id{latent}.

The channel \id{obs}, whose provider is the coroutine $\id{Model}$, is used to identify observations in the
probabilistic model.
The coroutine $\id{Model}$ accesses \id{obs} on lines 4 and 8, each of which lies in a branch of the conditional
command on line 3.
Because the conditional command is associated with \id{latent}, it should \emph{not} bother with the
communication on channel \id{obs};
thus, we require that the two branches of the conditional command have the same guidance protocol for
\id{obs}.
The protocol shown as type \eqref{Eq:ObsType} specifies that the coroutine $\id{Model}$ produces a
single $\treal$-valued observation, and $\id{Model}$ implements this protocol exactly.

\paragraph{Recursion}
Probabilistic programs can use recursion to express complex generative
models, such as a \emph{probabilistic context-free grammar} (PCFG), which is a popular model for
constructing languages~\cite{SRU:JLM92}.
\cref{Fig:RecursiveModel} shows a recursive model that generates a random expression tree with two
constructors: $\m{Const}(\cdot)$ for leaf nodes and $\m{Add}(\cdot;\cdot)$ for internal nodes.

\begin{figure}
\centering
\begin{small}
\begin{pseudo}
  \kproc \id{Pcfg}() \kw{consume} \id{latent} \kw{provide} . = \\+
    $k$ \kgets \ksamplei{\ebeta{3}{1}}{latent}; \\
    \kcall \id{PcfgGen}($k$) \\-
  \\
  \kproc \id{PcfgGen}($k$) \kw{consume} \id{latent} \kw{provide} . = \\+
    $u$ \kgets \ksamplei{\eunif}{latent}; \\
    \kifo{latent} $u < k$ \kthen \\+
      $v$ \kgets \ksamplei{\enormal{0}{1}}{latent}; \\
      \kreturn{\m{Const}(v)} \\-
    \kelse \\+
      $\id{lhs}$ \kgets \kcall \id{PcfgGen}($k$); \\
      $\id{rhs}$ \kgets \kcall \id{PcfgGen}($k$); \\
      \kreturn{\m{Add}(\id{lhs}; \id{rhs})}
\end{pseudo}
\end{small}
\caption{A recursive probabilistic model.}
\label{Fig:RecursiveModel}
\end{figure}

To support recursion in probabilistic programs, we add a standard recursive-type
constructor to guide types.
However, \emph{composition} of the guide types from multiple procedure calls in a
non-tail-recursive program remains a challenge.
One straightforward approach is to add a sequencing type $A \fatsemi B$ that
types a channel whose provider starts with a type $A$ protocol and  then continues
with a type $B$ protocol, but such sequencing types will complicate the type
system, because they allow a guidance protocol to be described by \emph{different}
types.
For example, both $(\treal \wedge \treal \wedge \one)$ and
$(( \treal \wedge \one) \fatsemi (\treal \wedge \one) )$
describe a channel whose provider sends two $\treal$-valued random samples.

To sidestep the need for a nontrivial equivalence check in the type system, we adapt the idea of
\emph{type-level polymorphism}, and parameterize the guide type for a recursive
coroutine by a \emph{continuation} type that describes the communication after
a procedure call to this coroutine returns.
For example, consider the following parametric type $\m{R}[\cdot]$.
\[
\m{R}[X] \defeq \tureal \wedge ((\treal \wedge X) \echoice \m{R}[\m{R}[X]]),
\]
It specifies a guidance protocol by \emph{prepending} messages to the
continuation protocol defined by the type parameter $X$.
The type $\m{R}[X]$ precisely describes the behavior of the $\id{PcfgGen}$ coroutine shown in
\cref{Fig:RecursiveModel}:
the coroutine first receives an $\tureal$-valued random sample (line 6);
evaluates and sends out a branch selection (line 7); and then based on the branch
selection, the coroutine either receives an $\treal$-valued random sample (line 8)
and then returns (i.e., continues with the continuation protocol $X$), or
makes two recursive procedure calls (lines 11 and 12).
The guide type of the else-branch can be justified by backward reasoning:
at line 13, the coroutine returns (i.e., continues with the continuation protocol $X$);
at line 12, because the guide type \emph{after} the procedure call is $X$,
we obtain the guide type \emph{before} the procedure call by instantiating $\m{R}$ with $X$;
and at line 11, because the guide type \emph{after} the procedure call is $\m{R}[X]$,
we again instantiate $\m{R}$, but with $\m{R}[X]$, to derive the guide type of the else-branch.
Finally, for the coroutine $\id{Pcfg}$ shown in \cref{Fig:RecursiveModel},
we derive $\tureal \wedge \m{R}[\one]$ as the guidance protocol for channel $\id{latent}$.

\paragraph{Control-flow divergence.}
In \cref{Fig:ProgAsCoro}, the model program $\id{Model}$ and the guide program $\id{Guide}_1$
have very similar control flow.
In general, our type system permits the guide's control-flow structure to diverge from the model's,
as long as the two programs communicate with each other in a consistent way, i.e., the two
programs follow the same guidance protocol for the channel over which they communicate.
For example, the program below implements a part of a Bayesian linear-regression model with outliers~\cite{PLDI:CSL19},
where the latent variable $\id{prob\_outlier}$ describes how likely a data point does not
conform to the linear relationship, and $\id{is\_outlier}$ is a Boolean-valued latent variable that
indicates if a data point is an outlier.
\begin{small}
\begin{pseudo}
  $\id{prob\_outlier}$ \kgets \ksamplei{\eunif}{latent}; \\
  $\id{is\_outlier}$ \kgets \ksamplei{\eber{\id{prob\_outlier}}}{latent}; \\
  \kreturn{}
\end{pseudo}  
\end{small}
For MCMC algorithms, the guide program generates a new random sample
from an old one; thus,
for better inference performance,
an MCMC guide usually behaves differently for
different old samples.
The following program implements a part of a guide that branches on $\id{is\_outlier}$
from the old sample~\cite{PLDI:CSL19}.
Intuitively, this guide proposes the negation (with a small amount of noise) of the old $\id{is\_outlier}$,
which is bound to a program variable $\id{old\_is\_outlier}$;
i.e., if the old $\id{is\_outlier}$ is true (resp., false), then the guide is likely to propose false (resp., true).
\begin{small}
\begin{pseudo}
  $\id{prob\_outlier}$ \kgets \ksampleo{\ebeta{2}{5}}{latent}; \\
  \kif $\id{old\_is\_outlier}$ \kthen \\+
    $\id{is\_outlier}$ \kgets \ksampleo{\eber{0.1}}{latent}; \\
    \kreturn{} \\-
  \kelse \\+
    $\id{is\_outlier}$ \kgets \ksampleo{\eber{0.9}}{latent}; \\
    \kreturn{}
\end{pseudo}  
\end{small}
Although the model and the guide have divergent control-flow structures, 
in our type system, we can express the guidance protocol for channel $\id{latent}$ as
$\tureal \wedge \tbool \wedge \one$; that is, both programs sample an
$\tureal$-valued random variable and then sample a Boolean-valued one.

\paragraph{Type inference}
Guide types can be automatically inferred from code; in practice, they can still be used
as specifications of the programs for better understanding.
Our implementation can infer guide types for the examples mentioned so far, including
the recursive one shown in \cref{Fig:RecursiveModel}.


\section{A Coroutine-Based PPL}
\label{Se:Technical}

In this section, we formulate a core monadic calculus for coroutine-based probabilistic programming.


\paragraph{Syntax}
\cref{Fig:Syntax} presents the grammar of basic types $\tau$, expressions $e$, values $v$,
commands $m$, and programs $\calD$ in the core calculus via abstract binding trees~\cite{book:PFPL16}.
There is a modal distinction in the core language:
expressions describe \emph{purely deterministic} computations, while
commands describe \emph{probabilistic} computations.
Intuitively, we treat randomness as a kind of monadic effect~\cite{LICS:Moggi89}.

\begin{figure}
\centering
\begin{small}
\begin{align*}
  \tau & \Coloneqq \tunit \mid \tbool \mid \tureal \mid \tpreal \mid \treal \mid \tnat_n \mid \tnat \mid \tau_1 \to \tau_2 \mid \tdist{\tau}   \\
  e & \Coloneqq x \mid \etriv \mid \etrue \mid \efalse \mid \econd{e}{e_1}{e_2} \mid \bar{r} \mid \bar{n} \mid \ebinop{\Diamond}{e_1}{e_2} \\
  & \mid \eabs{x}{\tau}{e} \mid \eapp{e_1}{e_2} \mid \elet{e_1}{x}{e_2} \\
  & \mid \eber{e} \mid \eunif \mid \ebeta{e_1}{e_2} \mid \egamma{e_1}{e_2} \\
  & \mid \enormal{e_1}{e_2} \mid \ecat{e_1,\cdots,e_n} \mid \egeo{e} \mid \epois{e} \\
  v & \Coloneqq \etriv \mid \etrue \mid \efalse \mid \bar{r} \mid \bar{n} \mid \vclo{V}{\eabs{x}{\tau}{e}} \\
  & \mid \eber{v} \mid \eunif \mid \ebeta{v_1}{v_2} \mid \egamma{v_1}{v_2} \\
  & \mid \enormal{v_1}{v_2} \mid \ecat{v_1,\cdots,v_n} \mid \egeo{v} \mid \epois{v} \\
  m & \Coloneqq \mret{e} \mid \mbnd{m_1}{x}{m_2} \mid \mcall{f}{e} \\
  & \mid \msamplei{e}{a} \mid \msampleo{e}{a} \\
  & \mid \mbranchi{m_1}{m_2}{a} \mid \mbrancho{e}{m_1}{m_2}{a} \\
  \calD & \Coloneqq \many{\fundec{f}{\tau_1}{\tau_2}{x}{m}{a}{b}}
\end{align*}
\end{small}
\caption{Syntax of the core calculus.}
\label{Fig:Syntax}  
\end{figure}

The purely deterministic fragment is a simply-typed lambda calculus augmented with
\emph{scalar} types (i.e., nullary products $\tunit$, Booleans $\tbool$, unit interval $\tureal$,
positive real numbers $\tpreal$, real numbers $\treal$, integer rings $\tnat_n$,
and natural numbers $\tnat$),
as well as a \emph{distribution} type $\tdist{\tau}$.
The syntactic form $\ebinop{\Diamond}{e_1}{e_2}$ represents expressions that perform
built-in binary operations $\Diamond$ on scalar values.
Inhabitants of $\tdist{\tau}$ are the primitive distributions from which probabilistic
programs can draw a random value of type $\tau$;
for example, Bernoulli distributions $\eber{\cdot}$ have type $\tdist{\tbool}$,
the uniform distribution on unit interval $\eunif$ has type $\tdist{\tureal}$,
and geometric distributions $\egeo{\cdot}$ have type $\tdist{\tnat}$.
For each primitive distribution $d$, we assume that it admits two fields:
$d.\mathrm{support}$ and $d.\mathrm{density}$ are the support and the density function
of the distribution, respectively.
In the core calculus, the type of a primitive distribution characterizes the support of the
distribution \emph{precisely}: for a distribution $d$ of type $\tdist{\tau}$ and a value $v$,
it holds that $v \in d.\mathrm{support}$ if and only if $v$ is an inhabitant of type $\tau$.
Primitive distributions can be generalized to density-carrying expressions~\cite{POPL:BAV12,TACAS:BBG13}
to further improve language expressibility.

The probabilistic fragment is a monadic calculus augmented with probabilistic constructs
and communication primitives for coroutine-based programming.
The \emph{sampling} commands $\msampleo{e}{a}$ and $\msamplei{e}{a}$ first evaluate
the expression $e$ to a primitive distribution $d$.
Then the \textsl{send} version $\msampleo{d}{a}$ draws a value from $d$ and sends it on
channel $a$, whereas the \textsl{receive} version $\msamplei{d}{a}$ receives a value from
channel $a$ and treats it as a sample from $d$.
The random samples can influence the \emph{likelihoods} of computations; thus, randomness
can be seen as a source of side effects.
The \emph{branching} commands also have a \textsl{send} version $\mbrancho{e}{m_1}{m_2}{a}$,
which evaluates $e$ to a Boolean value and sends it as the branch selection on channel $a$;
and a \textsl{receive} version $\mbranchi{m_1}{m_2}{a}$, which receives a branch selection from
channel $a$.
The syntactic form $\mcall{f}{e}$ represents a procedure call, where $f$ is a procedure
name and $e$ is the argument.

A probabilistic program $\calD$ is a collection of (mutually recursive) procedures,
each of which has the form $\fundec{f}{\tau_1}{\tau_2}{x}{m}{a}{b}$, where $f$ is the procedure name,
$x$ is the parameter, $m$ is a command that represents the procedure body,
$a$ is the name of the channel consumed by $f$, and $b$ is the name of the channel provided by $f$.
Note that both $a$ and $b$ are optional; that is, the procedure $f$ might not consume any channel,
and it might not provide any channel.

\paragraph{Semantics}
We develop a big-step operational semantics for the core calculus.
Details of the semantics are included in
\iflong
\cref{Se:FullSpec}.
\else
the technical report~\cite{Techreport}.
\fi
The evaluation judgments for expressions have the form $V \vdash e \evalto v$, where $V$ is an
\emph{environment} that maps program variables to values.
The evaluation rules for expressions are skipped here because they are standard.

We adopt a trace-based approach~\cite{ICFP:BLG16,JCSS:Kozen81} in our semantics of
probabilistic computations.
A \emph{guidance trace} $\sigma$ is a finite sequence of guidance messages exchanged on a channel;
each \emph{guidance message} has the form
$\msgobjl{v}$ (resp., $\dirobjl{v}$) for a sample value $v$ (resp., a branch selection $v$)
from the \underline{p}rovider to the consumer, the form
$\msgobjr{v}$ (resp., $\dirobjr{v}$) for a sample value $v$ (resp., a branch selection $v$)
from the \underline{c}onsumer to the provider,
or a procedure-call indicator $\foldobj$.\footnote{
The $\foldobj$ message is only useful in the theoretical development;
it can be seen as the introduction form for guidance traces whose type is a type-operator instantiation (see \cref{Se:GuideTypes}).
}
The evaluation judgments for commands have the form $V \mid \chtype{a}{\sigma_a} ; \chtype{b}{\sigma_b}
\vdash m \evalp{w} v$, where $V$ is an environment,
$m$ is a command that consumes channel $a$ and provides channel $b$,
$\sigma_a$ and $\sigma_b$ are guidance traces on the channels,
$v$ is the evaluation result,
and $w \ge 0$ is a \emph{weight} that expresses how likely the guidance traces are.
Intuitively, a probabilistic program specifies a probability distribution on guidance traces,
and the weights represent \emph{probability densities} with respect to the distribution. 

\begin{figure}
\centering 
\begin{mathpar}\small
  \Rule{EM:Ret}
  { V \vdash e \evalto v
  }
  { V \mid \chtype{a}{[]} ; \chtype{b}{[]} \vdash \mret{e} \evalp{1} v }
  \and
  \Rule{EM:Bnd}
  { V  \mid \chtype{a}{\sigma_a} ; \chtype{b}{\sigma_b} \vdash m_1 \evalp{w_1} v_1 \\
    V[x \mapsto v_1] \mid \chtype{a}{\sigma_a'} ; \chtype{b}{\sigma_b'} \vdash m_2 \evalp{w_2} v_2
  }
  { V \mid \chtype{a}{\sigma_a \concat \sigma_a'} ; \chtype{b}{\sigma_b \concat \sigma_b'} \vdash \mbnd{m_1}{x}{m_2} \evalp{w_1 \cdot w_2} v_2 }
  \and
  \Rule{EM:Sample:Recv:L}
  { V \vdash e \evalto d \\
    v \in d.\mathrm{support} \\
    w = d.\mathrm{density}(v)
  }
  { V \mid \chtype{a}{[\msgobjl{v}]} ; \chtype{b}{[]} \vdash \msamplei{e}{a} \evalp{w} v }
  \and
  \Rule{EM:Sample:Send:R}
  { V \vdash e \evalto d \\
    v \in d.\mathrm{support} \\
    w = d.\mathrm{density}(v)
  }
  { V \mid \chtype{a}{[]} ; \chtype{b}{[\msgobjl{v}]} \vdash \msampleo{e}{b} \evalp{w} v }
  \and
  \Rule{EM:Cond:Send:L}
  { V \vdash e \evalto v_e \\
    i = \m{ite}(v_a, 1, 2) \\
    V \mid \chtype{a}{\sigma_a} ; \chtype{b}{\sigma_b} \vdash m_i \evalp{w} v
  }
  { V \mid \chtype{a}{[\dirobjr{v_a}] \concat \sigma_a} ; \chtype{b}{\sigma_b} \vdash \mbrancho{e}{m_1}{m_2}{a} \evalp{w \cdot [v_a = v_e]} v }
  \and
  \Rule{EM:Cond:Recv:R}
  { i = \m{ite}(v_b, 1, 2) \\
    V \mid \chtype{a}{\sigma_a} ; \chtype{b}{\sigma_b} \vdash m_i \evalp{w} v
  }
  { V \mid \chtype{a}{\sigma_a} ; \chtype{b}{[\dirobjr{v_b}] \concat \sigma_b} \vdash \mbranchi{m_1}{m_2}{b} \evalp{w} v }
  \and
  \Rule{EM:Call}
  { \calD(f) = \fundec{f}{\tau_1}{\tau_2}{x_f}{m_f}{a}{b} \\
    V \vdash e \evalto v_1 \\
    \emptyset[x_f \mapsto v_1] \mid \chtype{a}{\sigma_a} ; \chtype{b}{\sigma_b} \vdash m_f \evalp{w} v_2
  }
  { V \mid \chtype{a}{[\foldobj] \concat \sigma_a} ; \chtype{b}{[\foldobj] \concat \sigma_b} \vdash \mcall{f}{e} \evalp{w} v_2 }  
\end{mathpar}
\caption{Selected evaluation rules for commands.}
\label{Fig:EvalElab}
\end{figure}

\cref{Fig:EvalElab} shows the evaluation rules for selected commands.
We use the following notational conventions.
We denote the empty environment by $\emptyset$, and
updating a binding of $x$ in an environment $V$ to $v$ by $V[x \mapsto v]$.
We use the $\dplus$ operator to concatenate two traces.
We write $\m{ite}$ as a shorthand for $\m{if{\text{-}}then{\text{-}}else}$.
The \emph{Iverson brackets} $[\cdot]$ are defined by $[\varphi]=1$ if $\varphi$ is true
and otherwise $[\varphi]=0$.

The (\textsc{EM:Sample:*}) rules take a value from the guidance traces as the result of the sampling,
and use the density functions of primitive distributions to calculate the weight for the guidance traces.
The (\textsc{EM:Cond:Send:L}) rule evaluates the branch predicate to obtain a Boolean value,
and enforce that the branch selection from the guidance trace of the consumed channel must be
the same as the predicate's value;
if the guidance trace sets the branch selection to a different value, we simply set the weight
of this trace to zero.
The (\textsc{EM:Call}) rule requires the guidance traces start with a $\foldobj$ message, and proceeds
by evaluating the body of the callee.

\begin{example}\label{Ex:ExampleCommand}
  Consider the command
  \begin{align*}
    m_1 \defeq{} & \enskip \m{bnd}( ~ \msamplei{\enormal{0}{1}}{a}; ~ x. \\[-3pt]
    & \enskip \m{bnd}( ~ \msampleo{\enormal{x}{1}}{b}; ~ y. \\[-3pt]
    & \enskip \mret{\ebinop{+}{x}{y}} \enskip ) \enskip ),
  \end{align*}
  which consumes a channel $a$ and provides a channel $b$.
  Let $\varphi \defeq \lambda x. \frac{1}{\sqrt{2\pi}} e^{-\frac{1}{2} x^2}$ be the probability
  density function of the standard normal distribution $\enormal{0}{1}$.
  Let $\sigma_a \defeq [ \msgobjl{\bar{1}} ]$ and $\sigma_b \defeq [ \msgobjl{\bar{2}} ]$.
  Then we can derive the evaluation judgment 
  \[
  \emptyset \mid \chtype{a}{\sigma_a} ; \chtype{b}{\sigma_b} \vdash m_1 \evalp{\varphi(1) \cdot \varphi(1)} \bar{3},
  \]
  for the command $m_1$ and the guidance traces $\sigma_a,\sigma_b$.
\end{example}

\paragraph{Communication}
There are a lot of formalisms for communication in (concurrent) programming systems,
such as CCS~\cite{book:Milner89}, Theoretical CSP~\cite{CACM:Hoare78},
and $\pi$-calculus~\cite{JIC:MPW92a,JIC:MPW92b}.
In this paper, we use a lightweight approach to handling communication;
that is, in the semantics, we assume we have all the messages exchanged on all the
communication channels.
We use this formalism because
(i) our focus is to reason about soundness of Bayesian inference, rather than concurrency-related
properties (e.g., deadlock freedom); and
(ii) the inference algorithms we study in \cref{Se:Inference} involve only two
coroutines---one for the model and the other for the guide---so the communication
in our system is much simpler than that in general concurrent systems.

\begin{example}\label{Ex:AnotherExampleCommand}
  Consider the command
  \[
    m_2 \defeq{}  \enskip \m{bnd}( \enskip \msampleo{\enormal{3}{1}}{a}; \enskip \_. \enskip \mret{\etriv} \enskip ),
  \]
  which provides a channel $a$ that is consumed by the command $m_1$ in \cref{Ex:ExampleCommand}.
  To model the communication between $m_2$ and $m_1$, we simply use the guidance trace
  $\sigma_a = [\msgobjl{\bar{1}}]$ as the sequence of messages exchanged on
  channel $a$ in the semantics, and derive evaluation judgments for $m_2$ and $m_1$
  separately.
  We showed the judgment for $m_1$ in \cref{Ex:ExampleCommand}; here, we
  can derive the judgment
  \[
  \emptyset \mid \varnothing ; \chtype{a}{\sigma_a} \vdash m' \evalp{\varphi(-2)} \etriv,
  \]
  for command $m_2$ and guidance trace $\sigma_a$.
  We use the $\varnothing$ symbol to indicate that $m_2$ does not consume
  any channel.
\end{example}

\section{Guide Types}
\label{Se:GuideTypes}

\paragraph{Type formation}
We take inspiration from a \emph{structuring} principle in session types~\cite{CONCUR:Honda93,ESOP:HVK98},
and develop \emph{guide types} to enforce protocols for guidance traces.
The grammar shown below formulates the syntax of guide types.
We write $A,B$ for guide types, $X$ for type variables, $T$ for unary type operators, and
$F$ for procedure signatures.
{\begin{align*}
  A,B & \Coloneqq X \mid \one \mid T[A] \mid \tau \wedge A \mid \tau \supset A \mid A \ichoice B \mid A \echoice B \\
  F & \Coloneqq \tau_1 \leadsto \tau_2 \mid \chtype{a}{T_a} ; \chtype{b}{T_b} \\
  \calT & \Coloneqq \many{\m{typedef}(T.X. A)}
\end{align*}}
%
%
The type $\one$ indicates an ended channel, where the guidance trace is empty.
The type $T[A]$ instantiates a unary type operator $T$ with a guide type $A$.
For sample passing and branch selection, each type constructor has a dual version that reverses
the role of the provider and the consumer.
The type $\tau \wedge A$ types a channel whose \emph{provider} samples a random value
, sends it on the channel, and then continues with a type $A$ guidance protocol;
dually, the type $\tau \supset A$ types a channel whose \emph{consumer} samples and sends
a random value.
Similarly, the type $A \ichoice B$ types a channel whose \emph{provider} evaluates a branch
predicate, sends a branch selection on the channel, and then continues with a type $A$ guidance
protocol or a type $B$ protocol based on the branch selection;
dually, the type $A \echoice B$ types a channel whose \emph{consumer} evaluates and sends a
branch selection.

\begin{remark}
  In the rest of this paper, we will \emph{not} use the dual types $\tau \supset A$ and $A \ichoice B$.
  We introduce these types here for theoretical completeness, and they may be used in some future development.
\end{remark}

Type operators prescribe guidance protocols for procedures by parameterizing with a continuation
type that describes the guidance protocol after a procedure call.
A procedure signature $\tau_1 \leadsto \tau_2 \mid \chtype{a}{T_a} ; \chtype{b}{T_b}$ types a
procedure that takes a parameter of type $\tau_1$, returns a result of type $\tau_2$,
consumes a channel $a$, and provides a channel $b$,
such that
if the guidance protocols for $a$ and $b$ \emph{after} a procedure call are $A$ and $B$, respectively,
then the guidance protocols for $a$ and $b$ \emph{before} the procedure call are $T_a[A]$ and $T_b[B]$,
respectively.

A \emph{type definition} $\m{typedef}(T.X.A)$ declares a unary type operator $T$ that takes a type
parameter $X$ and produces a guide type $A$, which can reference $X$.
Because type operators are used to prescribe procedure signatures, we assume that a probabilistic
program is always accompanied by a collection $\calT$ of (mutually recursive) type definitions.

\begin{example}\label{Ex:RecurType}
  We can formally declare the type operator $\m{Recur}$ for the \id{PcfgGen} procedure shown
  in \cref{Fig:RecursiveModel} as
  $
  \m{typedef}( ~ \m{R}.~X.~  \tureal \wedge ((\treal \wedge X) \echoice \m{R}[\m{R}[X]]) ~).
  $
\end{example}

\paragraph{Typing rules}
The typing judgments for expressions have the form $\Gm \vdash e : \tau$,
where $\Gm$ is a \emph{typing context} that maps program variables to basic types (defined in \cref{Fig:Syntax}).
A full list of typing rules is included in
\iflong
\cref{Se:FullSpec}.
\else
the technical report~\cite{Techreport}.
\fi
The typing rules for expressions are skipped here because they are standard.

The typing judgments for commands have the form
\[
\Gm \mid \chtype{a}{A}; \chtype{b}{B} \vdash_{\Sg} m \dotsim \tau \mid \chtype{a}{A'} ; \chtype{b}{B'},
\]
where $\Sg$ maps procedure identifiers to procedure signatures.
The intuitive meaning of the typing judgment is that
if the channels $a$ and $b$ are of the guidance protocols $A$ and $B$, respectively,
then we can evaluate the command $m$ to a value of type $\tau$,
and after the evaluation, the channels $a$ and $b$ are of the guidance protocols $A'$
and $B'$, respectively.

\begin{figure}
\centering
\begin{mathpar}\small
  \Rule{TM:Ret}
  { \Gm \vdash e : \tau
  }
  { \Gm \mid \chtype{a}{A} ; \chtype{b}{B} \vdash \mret{e} \dotsim \tau \mid \chtype{a}{A} ; \chtype{b}{B}  }
  \and
  \Rule{TM:Bnd}
  { \Gm \mid \chtype{a}{A} ; \chtype{b}{B} \vdash m_1 \dotsim \tau_1 \mid \chtype{a}{A'} ; \chtype{b}{B'} \\\\
    \Gm, x : \tau_1 \mid \chtype{a}{A'} ; \chtype{b}{B'} \vdash m_2 \dotsim \tau_2 \mid \chtype{a}{A''} ; \chtype{b}{B''}
  }
  { \Gm \mid \chtype{a}{A} ; \chtype{b}{B} \vdash \mbnd{m_1}{x}{m_2} \dotsim \tau_2 \mid \chtype{a}{A''} ; \chtype{b}{B''} }
  \and
  \Rule{TM:Sample:Recv:L}
  { \Gm \vdash e : \tdist{\tau}
  }
  { \Gm \mid \chtype{a}{\tau \wedge A} ; \chtype{b}{B} \vdash \msamplei{e}{a} \dotsim \tau \mid \chtype{a}{A} ; \chtype{b}{B} }
  \and
  \Rule{TM:Sample:Send:R}
  { \Gm \vdash e : \tdist{\tau}
  }
  { \Gm \mid \chtype{a}{A} ; \chtype{b}{ \tau \wedge B} \vdash \msampleo{e}{b} \dotsim \tau \mid \chtype{a}{A} ; \chtype{b}{B} }
  \and
  \Rule{TM:Cond:Send:L}
  { \Gm \vdash e : \tbool \\
    \Gm \mid \chtype{a}{A_1} ; \chtype{b}{B} \vdash m_1 \dotsim \tau \mid \chtype{a}{A'} ; \chtype{b}{B'} \\\\
    \Gm \mid \chtype{a}{A_2} ; \chtype{b}{B} \vdash m_2 \dotsim \tau \mid \chtype{a}{A'} ; \chtype{b}{B'}
  }
  { \Gm \mid \chtype{a}{A_1 \echoice A_2} ; \chtype{b}{B} \vdash \mbrancho{e}{m_1}{m_2}{a} \dotsim \tau \mid \chtype{a}{A'} ; \chtype{b}{B'} }
  \and
  \Rule{TM:Cond:Recv:R}
  { \Gm \mid \chtype{a}{A} ; \chtype{b}{B_1} \vdash m_1 \dotsim \tau \mid \chtype{a}{A'} ; \chtype{b}{B'} \\\\
    \Gm \mid \chtype{a}{A} ; \chtype{b}{B_2} \vdash m_2 \dotsim \tau \mid \chtype{a}{A'} ; \chtype{b}{B'}
  }
  { \Gm \mid \chtype{a}{A} ; \chtype{b}{B_1 \echoice B_2} \vdash \mbranchi{m_1}{m_2}{b} \dotsim \tau \mid \chtype{a}{A'} ; \chtype{b}{B'} }
  \and
  \Rule{TM:Call}
  { \Sg(f) = \tau_1 \leadsto \tau_2 \mid \chtype{a}{T_a} ; \chtype{b}{T_b} \\
    \Gm \vdash e : \tau_1
  }
  { \Gm \mid \chtype{a}{T_a[A]} ; \chtype{b}{T_b[B]} \vdash \mcall{f}{e} \dotsim \tau_2 \mid \chtype{a}{A} ; \chtype{b}{B} }
\end{mathpar}
\caption{Selected typing rules for commands.}
\label{Fig:TypingElab}
\end{figure}

\cref{Fig:TypingElab} presents the typing rules for commands.
We assume a fixed global $\Sg$ that we omit from the rules.
Intuitively, the rules formulate a \emph{backward}-reasoning system:
we start with continuation types $A'$ and $B'$ for the channels $a$ and $b$, respectively,
and then prepend the guidance messages sent or received by the command $m$ to $A'$ and $B'$,
to obtain the guide types $A$ and $B$ for the channels $a$ and $b$ before the evaluation of $m$,
respectively.
For sample passing and branch selection, each guide type has two derivation rules:
one for the consumed channel $a$, and the other for the provided channel $b$.
For example, the type $\tau \wedge A$ represents a channel whose provider sends a
sample of type $\tau$; thus, if the consumed channel $a$ has such a type,
the rule (\textsc{TM:Sample:Recv:L}) \emph{receives} a sample from the provider of $a$,
and if the provided channel $b$ has such a type, the rule (\textsc{TM:Sample:Send:R})
\emph{sends} a sample to the consumer of $b$.

The rule (\textsc{TM:Call}) handles procedure calls.
For a procedure call $\mcall{f}{e}$, the rule fetches from $\Sg$ the procedure $f$'s signature
$\tau_1 \leadsto \tau_2 \mid \chtype{a}{T_a} ; \chtype{b}{T_b}$,
and then instantiates the type operators $T_a,T_b$ with continuation types $A,B$, respectively,
to obtain the guide types $T_a[A]$ and $T_b[B]$ for the channels $a$ and $b$ before the procedure
call, respectively.

\begin{example}\label{Ex:TypeDerivation}
  Consider the command
  \begin{align*}
    m_3 \defeq{} & \enskip \m{bnd}( \enskip \mcall{f}{k}; \enskip \_. \\[-3pt]
      & \enskip \m{bnd}( \enskip \msamplei{a}{\enormal{0}{1}}; \enskip \_. \\[-3pt]
      & \enskip \m{bnd}( \enskip \mcall{f}{k}; \enskip \_. \\[-3pt]
      & \enskip \mret{\etriv} \enskip ) \enskip ) \enskip ),
  \end{align*}
  where the variable $k$ has type $\tureal$ and the procedure $f$ has signature
  $
  \tureal \leadsto \tunit \mid \chtype{a}{T} ; \varnothing,
  $
  i.e., the procedure $f$ consumes channel $a$ but does not provide any channel,
  and channel $a$ is associated with a type operator $T$.
  Now we show that we can derive a typing judgment for $m_3$
  by backward reasoning.
  First, by (\textsc{TM:Ret}), we have
  \[
  k : \tureal \mid \chtype{a}{\one}; \varnothing \vdash_\Sg \mret{\etriv} \dotsim \tunit \mid \chtype{a}{\one}; \varnothing.
  \]
  Then by (\textsc{TM:Call}), we derive
  \[
  k : \tureal \mid \chtype{a}{T[\one]}; \varnothing \vdash_\Sg \mcall{f}{k} \dotsim \tunit \mid \chtype{a}{\one}; \varnothing.
  \]
  Define $m_4 \defeq \mbnd{\mcall{f}{k}}{\_}{\mret{\etriv}}$.
  Thus, by (\textsc{TM:Bnd}),
  \[
  k : \tureal \mid \chtype{a}{T[\one]}; \varnothing \vdash_\Sg m_4 \dotsim \tunit \mid \chtype{a}{\one}; \varnothing.
  \]
  Define $m_5 \defeq \mbnd{\msamplei{\enormal{0}{1}}{a}}{\_}{m_4}$.
  By (\textsc{TM:Sample:Recv:L}) and (\textsc{TM:Bnd}), we have
  \[
  k : \tureal \mid \chtype{a}{\treal \wedge T[\one]}; \varnothing \vdash_\Sg m_5 \dotsim \tunit \mid \chtype{a}{\one} ; \varnothing.
  \]
  Finally, we again apply (\textsc{TM:Call}) and (\textsc{TM:Bnd}) to derive
  \[
  k : \tureal \mid \chtype{a}{T[\treal \wedge T[\one]]}; \varnothing \vdash_\Sg m_3 \dotsim \tunit \mid \chtype{a}{\one} ; \varnothing.
  \]
\end{example}

\paragraph{Type safety}

We present some theoretical results about type safety of guide types.
Proofs are included in
\iflong
\cref{Se:FullSpec}.
\else
the technical report~\cite{Techreport}.
\fi

We first formulate two judgments for well-formedness of values and guidance traces.
The judgment $v : \tau$ means that value $v$ has type $\tau$.
The judgment $\sigma : A$ means that the guidance trace is a sequence of messages
that satisfies protocol $A$.
Rules for these judgments are straightforward; we omit them here but include them
in
\iflong
\cref{Se:FullSpec}.
\else
the technical report~\cite{Techreport}.
\fi

The theorem below states that if $m$ is a well-typed closed command,
and it evaluates to a value $v$ under guidance traces $\sigma_a,\sigma_b$,
then $v$ is a well-typed value, and $\sigma_a,\sigma_b$ are well-typed guidance traces.

\begin{theorem}\label{Cor:SoundnessW}
  If $\cdot \mid \chtype{a}{A} ; \chtype{b}{B} \vdash_\Sg m \dotsim \tau \mid \chtype{a}{\one} ; \chtype{b}{\one}$
  and $\emptyset \mid \chtype{a}{\sigma_a} ; \chtype{b}{\sigma_b} \vdash m \evalp{w} v$,
  then
  $\sigma_a : A$,
  $\sigma_b : B$,
  and $v : \tau$.
\end{theorem}

Furthermore, we can show some \emph{normalization} properties of guide types.
The theorem below states that if $m$ is a well-typed closed command,
and $\sigma_a,\sigma_b$ are well-typed guidance traces, then $m$ can evaluate to
some well-typed $v$ under $\sigma_a,\sigma_b$.

\begin{theorem}\label{Cor:SoundnessI}
  If
  $\cdot \mid \chtype{a}{A} ; \chtype{b}{B} \vdash_\Sg m \dotsim \tau \mid \chtype{a}{\one} ; \chtype{b}{\one}$,
  $\sigma_a : A$,
  and $\sigma_b : B$,
  then
  there exist $w,v$ such that
  $\emptyset \mid \chtype{a}{\sigma_a} ; \chtype{b}{\sigma_b} \vdash m \evalp{w} v$ and
  $v : \tau$.
\end{theorem}

We can strengthen the normalization property when a command will \emph{not} send out
any branch selections.
The theorem below states that
if a well-typed command $m$ consumes a channel $a$ with a type $A$ that does not contain $\echoice$ and provides a channel $b$ with a type $B$ that does not contain $\ichoice$,
and $\sigma_a,\sigma_b$ are well-typed guidance traces,
then $m$ can evaluate to some well-typed value $v$ under $\sigma_a,\sigma_b$
with a \emph{strictly positive} weight $w$.

\begin{theorem}\label{Cor:SoundnessII}
  If
  $\cdot \mid \chtype{a}{A} ; \chtype{b}{B} \vdash_\Sg m \dotsim \tau \mid \chtype{a}{\one} ; \chtype{b}{\one}$,
  $A$ is $\echoice$-free,
  $B$ is $\ichoice$-free,
  $\sigma_a : A$,
  and $\sigma_b : B$,
  then
  there exist $w,v$ such that
  $\emptyset \mid \chtype{a}{\sigma_a} ; \chtype{b}{\sigma_b} \vdash m \evalp{w} v$,
  $v : \tau$,
  and $w > 0$.
\end{theorem}

\paragraph{Type-inference algorithm}

We now sketch a type-inference algorithm that derives guide types automatically
from the implementation. 
In the algorithm, we assume we have information about basic types---such as the parameter
and result types for procedures and the typing contexts that map program variables to
basic types---because without guide types, our core language is a simply-typed lambda
calculus, for which type inference is decidable.

First, for each procedure $\fundec{f}{\tau_1}{\tau_2}{x}{m}{a}{b}$ in the program,
we create two fresh type operators $T_a$ and $T_b$ for the channels $a$ and $b$, respectively,
and obtains $\tau_1 \leadsto \tau_2 \mid \chtype{a}{T_a} ; \chtype{b}{T_b}$ as the
signature of this procedure.
Then we collect signatures of all the procedures in the program to obtain the map $\Sg$.

Now the task is to derive definitions of the type operators.
We observe that the rules in \cref{Fig:TypingElab} are syntax directed, and
they can be turned into an algorithmic system by interpreting
\[
\Gm \mid \chtype{a}{A} ; \chtype{b}{B} \vdash_\Sg m \dotsim \tau \mid \chtype{a}{A'} ; \chtype{b}{B'}
\]
as a function from $\Sg, \Gm, m, \tau, a, b, A', B'$ to $A,B$; i.e.,
we assume we know all the basic types, and we perform backward reasoning to infer
guide types.
Therefore, for each procedure $\fundec{f}{\tau_1}{\tau_2}{x}{m}{a}{b}$ with
signature $\tau_1 \leadsto \tau_2 \mid \chtype{a}{T_a} ; \chtype{b}{T_b}$,
we create two fresh type variables $X_a$ and $X_b$, derive two guide types
$A$ and $B$ through
\[
x:\tau_1 \mid \chtype{a}{A} ; \chtype{b}{B} \vdash_\Sg m \dotsim \tau_2 \mid \chtype{a}{X_a} ; \chtype{b}{X_b},
\]
and then add type definitions $\m{typedef}(T_a.X_a.A)$ and $\m{typedef}(T_b.X_b.B)$.


\section{Soundness of Bayesian Inference}
\label{Se:Inference}

In this section, we use guide types to reason about Bayesian inference.
We first present a measure-theoretic formulation of Bayesian inference in the
coroutine-based PPL, and prove that guide types are \emph{certificates}
of absolute continuity (\cref{Se:AbsoluteContinuity}).
We then sketch how guide types ensure key soundness conditions for multiple
Bayesian-inference algorithms (\cref{Se:SoundInference}).
\iflong
\Cref{Se:FullInference}
\else
The technical report~\cite{Techreport}
\fi
includes the details (e.g., formalizations and proofs) of this section.

\subsection{Verification of Absolute Continuity}
\label{Se:AbsoluteContinuity}

%
We use the following notions from measure theory:
$\sigma$-algebras, measurable spaces, measurable functions, measures, 
and Lebesgue integration.
\iflong
\Cref{Se:MeasureTheory} provides a review of these notions.
\fi

\paragraph{Semantic domains}
For each scalar type $\tau$, we equip it with a \emph{standard Borel space} $\interp{\tau}$
on the inhabitants of $\tau$, i.e., $\interp{\tau}$ is a measurable space isomorphic
to a countable set or the real line.
%
%
We then equip each type $\tau$ with a \emph{stock measure} $\lambda_{\interp{\tau}}$:
if $\interp{\tau}$ is a countable set, we define $\lambda_{\interp{\tau}}$ to be the counting measure; 
otherwise, $\interp{\tau}$ is a subset of the real line, so we define $\lambda_{\interp{\tau}}$ to be
the Lebesgue measure.

%
%
Because guidance traces are finite sequences of messages that
contain values of scalar types,
we can define $\interp{A}$ as a standard Borel space on guidance traces that satisfy protocol $A$.
We then construct the stock measure $\lambda_{\interp{A}}$ for $A$ by decomposing $A$ to products and/or sums of
scalar types, and then combining the stock measures for scalar types via product and/or coproduct measures.


\paragraph{Denotation of commands}
For a well-typed closed command $m$, i.e.,
$\cdot \mid \chtype{a}{A} ; \chtype{b}{B} \vdash_{\Sg} m \dotsim \tau \mid \chtype{a}{\one} ; \chtype{b}{\one}$,
we define the \emph{density function} of $m$ as
\[
  \mathbf{P}_m(\sigma_a,\sigma_b) \defeq \begin{dcases*}
   w & if $\emptyset \mid \chtype{a}{\sigma_a} ; \chtype{b}{\sigma_b} \vdash m \evalp{w} v$ \\
   0 & otherwise
 \end{dcases*}.
\]
We can prove that $\mathbf{P}_m$ is a measurable function from $\interp{A} \otimes \interp{B}$---the
product measurable space of $\interp{A}$ and $\interp{B}$---to nonnegative real numbers.
Thus, we construct a measure denotation $\interp{m}$ for $m$, by integrating $\mathbf{P}_m$ with
respect to the stock measure on the product space $\interp{A} \otimes \interp{B}$, i.e.,
\[
\interp{m}(S_{a,b}) \defeq \int_{S_{a,b}} \mathbf{P}_m(\sigma_a,\sigma_b) \lambda_{\interp{A} \otimes \interp{B}}(d (\sigma_a,\sigma_b)),
\]
where $S_{a,b}$ is a measurable set in $\interp{A} \otimes \interp{B}$.


\paragraph{Bayesian inference}
Let us fix a well-typed model program $m_{\m{m}}$ that consumes latent random variables on
a channel \id{latent} and provides observations on a channel \id{obs}, i.e.,
\[
\cdot \mid \chtype{latent}{A}; \chtype{obs}{B} \vdash_{\Sg} m_{\m{m}} \dotsim \tau_{\m{m}} \mid \chtype{latent}{\one}; \chtype{obs}{\one}.
\]
Usually, the program $m_{\m{m}}$ does \emph{not} receive any branch selections, i.e.,
$A$ is $\ichoice$-free and $B$ is $\echoice$-free.
Given a concrete observation $\sigma_o : B$ such that $\int \mathbf{P}_{m_\m{m}}(\sigma_\ell,\sigma_o) \lambda_{\interp{A}}(d \sigma_\ell) > 0$,
Bayesian inference is the problem of approximating the \emph{posterior} $\interp{m_{\m{m}}}_{\sigma_o}$,
a measure \emph{conditioned} with respect to $\sigma_o$, defined by
\begin{equation}\label{Eq:Posterior}
\interp{m_{\m{m}}}_{\sigma_o}(S_\ell) \defeq \frac{ \int_{S_\ell} \mathbf{P}_{m_{\m{m}}}(\sigma_\ell,\sigma_o) \lambda_{\interp{A}}(d \sigma_\ell) }{ \int \mathbf{P}_{m_{\m{m}}}(\sigma_\ell, \sigma_o) \lambda_{\interp{A}}(d \sigma_\ell) },
\end{equation}
where $S_\ell$ is a measurable set in $\interp{A}$, i.e., a set of guidance traces of type $A$.
Note that if we fix the observation $\sigma_o$, then the denominator of \cref{Eq:Posterior}
is a constant independent of $S_\ell$.
Thus, it is sufficient for an inference algorithm to ignore the denominator and approximate the measure
$S_\ell \mapsto \int_{S_\ell} \mathbf{P}_{m_{\m{m}}}(\sigma_\ell,\sigma_o) \lambda_{\interp{A}}(d \sigma_\ell)$.

\paragraph{Guide programs}
Bayesian-inference algorithms usually require some guide programs, such as
proposals for importance sampling and
approximating families for variational inference.
These guide programs specify measures on latent random
variables; in our system, we implement a guide program $m_{\m{g}}$ as
a coroutine that works with the model program $m_{\m{m}}$ and provides the
\id{latent} channel with guide type $A$ that $m_{\m{m}}$ consumes, i.e.,
\begin{gather*}
\cdot \mid \varnothing ; \chtype{latent}{A} \vdash_{\Sg} m_{\m{g}} \dotsim \tau_{\m{g}} \mid \varnothing ; \chtype{latent}{\one}, \\
\cdot \mid \chtype{latent}{A}; \chtype{obs}{B} \vdash_{\Sg} m_{\m{m}} \dotsim \tau_{\m{m}} \mid \chtype{latent}{\one}; \chtype{obs}{\one}.
\end{gather*}
The guide and model have the \emph{same} guide type $A$ on channel \id{latent}.
Because the guide \emph{provides} the channel and the model \emph{consumes} the channel,
the two programs interpret the guide type $A$ \emph{dually}; thus, their communication is compatible.

The coroutine-based paradigm folds the model and guide programs into a single entity;
thus, during the inference, both the model and guide coroutines execute.
To model \emph{possible} combinations of traces for a model-guide system, we introduce
a reduction relation $V \mid \chtype{a}{\sigma_a} ; \chtype{b}{\sigma_b} \red m \evalto v$,
where $V$ is an environment, $m$ is a command, $\sigma_a$ and $\sigma_b$ are guidance traces on channel $a$ and channel $b$, respectively, and $v$ is the reduction result.
The reduction relation is essentially the same as the evaluation relation for the operational semantics,
except that reduction does \emph{not} account for probabilities.
Below are two example rules.
\begin{mathpar}\small
  \Rule{RM:Sample:Send:R}
  { V \vdash e \evalto d \\
    v \in d.\mathrm{support} \\
  }
  { V \mid \chtype{a}{[]} ; \chtype{b}{[\msgobjl{v}]} \red \msampleo{e}{b} \evalto v }
  \and
  \Rule{RM:Cond:Send:L}
  { V \vdash e \evalto v_e \qquad
    i = \m{ite}(v_e, 1, 2) \qquad
    V \mid \chtype{a}{\sigma_a} ; \chtype{b}{\sigma_b} \red m_i \evalto v
  }
  { V \mid \chtype{a}{[\dirobjr{v_e}] \concat \sigma_a} ; \chtype{b}{\sigma_b} \red \mbrancho{e}{m_1}{m_2}{a} \evalto v }
\end{mathpar}
With the reduction relation, we say that a combination of traces $(\sigma_\ell,\sigma_o)$ is
\emph{possible} for the model program $m_\m{m}$ and the guide program $m_\m{g}$, if
$\emptyset \mid \chtype{latent}{\sigma_\ell} ; \chtype{obs}{\sigma_o} \red m_\m{m} \evalto v_\m{m}$ and $\emptyset \mid \varnothing; \chtype{latent}{\sigma_\ell} \red m_\m{g} \evalto v_\m{g}$ for some values $v_\m{m}$ and $v_\m{g}$.
We prove a lemma that connects the reduction relation with command denotations.

\begin{lemma}\label{Lem:JustificationOfIncludingModel}
  Suppose that $A$ is $\ichoice$-free, $B$ is $\echoice$-free, and
  \begin{align*}
    \cdot \mid \varnothing; \chtype{latent}{A} & \vdash_\Sg m_{\m{g}} \dotsim \tau_{\m{g}} \mid \varnothing; \chtype{latent}{\one}, \\
    \cdot \mid \chtype{latent}{A} ; \chtype{obs}{B} & \vdash_\Sg m_{\m{m}} \dotsim \tau_{\m{m}} \mid \chtype{latent}{\one} ; \chtype{obs}{\one}.
  \end{align*}
  Then a combination of traces $(\sigma_\ell, \sigma_o)$ is possible for the model $m_\m{m}$ and the guide $m_\m{g}$
  if and only if
  $\mathbf{P}_{m_\m{m}}(\sigma_\ell,\sigma_o) \neq 0$.
\end{lemma}

We can now define a denotation for the guide $m_{\m{g}}$,
accompanied by the model $m_{\m{m}}$ and conditioned
on a concrete observation $\sigma_o : B$, as a measure defined on possible traces:
\[
  \interp{m_{\m{g}} }^{m_\m{m}}_{\sigma_o}(S_\ell) \defeq \int_{S_\ell} [ \mathbf{P}_{m_{\m{m}}}(\sigma_\ell,\sigma_o) \neq 0 ] \cdot \mathbf{P}_{m_{\m{g}}}(\sigma_\ell) \lambda_{\interp{A}}(d\sigma_\ell),
\]
where $S_\ell$ is a measurable set in $\interp{A}$.


\paragraph{Absolute continuity}
A measure $\mu$ is said to be \emph{absolutely continuous} with respect to a measure $\nu$,
if $\mu$ and $\nu$ are defined on the same measurable space, and
$\nu(S) \neq 0$ for every measurable set $S$ for which $\mu(S) \neq 0$.

We prove that for a model-guide pair, guide types serve as certificates for
absolute continuity.

\begin{theorem}\label{The:AbsoluteContinuity}
  Suppose that $A$ is $\ichoice$-free, $B$ is $\echoice$-free,
  \begin{align*}
    \cdot \mid \varnothing; \chtype{latent}{A} & \vdash_\Sg m_{\m{g}} \dotsim \tau_{\m{g}} \mid \varnothing; \chtype{latent}{\one}, \\
    \cdot \mid \chtype{latent}{A} ; \chtype{obs}{B} & \vdash_\Sg m_{\m{m}} \dotsim \tau_{\m{m}} \mid \chtype{latent}{\one} ; \chtype{obs}{\one},
  \end{align*}
  and $\sigma_o : B$ such that
  $\int \mathbf{P}_{m_{\m{m}}}(\sigma_\ell, \sigma_o) \lambda_{\interp{A}}(d\sigma_\ell) > 0$.
  Then
  the measure $\interp{m_{\m{m}}}_{\sigma_o}$ is absolutely continuous with respect to
  the measure $\interp{m_{\m{g}}}^{m_{\m{m}}}_{\sigma_o}$, and vice versa.
\end{theorem}


\subsection{Soundness of Inference Algorithms}
\label{Se:SoundInference}

We now describe how guide types can help us reason about
inference algorithms. 

\paragraph{Importance sampling (IS)}
IS approximates the posterior distribution by drawing latent variables using the guide program,
and then reweights the samples by their \emph{importance}.
The operational rule below formulates a single step in the algorithm:
given a model program $m_{\m{m}}$, a guide program $m_{\m{g}}$, and a concrete observation $\sigma_o$,
IS performs joint execution of the two programs to draw a sample $\sigma_\ell$ with density $w_{\m{g}}$
and compute $\frac{w_{\m{m}}}{w_{\m{g}}}$ as the importance of $\sigma_\ell$.
\begin{mathpar}\small
  \inferrule
  { \emptyset \mid \varnothing ; \chtype{latent}{\sigma_\ell} \vdash m_{\m{g}} \evalp{w_{\m{g}}} \_ \\
    \emptyset \mid \chtype{latent}{\sigma_\ell} ; \chtype{obs}{\sigma_o} \vdash m_{\m{m}} \evalp{w_\m{m}} \_
  }
  { m_{\m{g}}; m_{\m{m}}; \sigma_o \vdash_{\textsc{is}}^{w_{\m{g}}}  \tuple{\sigma_\ell, \sfrac{w_{\m{m}}}{w_{\m{g}}} } }
\end{mathpar}
By \cref{The:AbsoluteContinuity}, if the model and guide programs are well-typed, then
the posterior $\interp{m_{\m{m}}}_{\sigma_o}$ is absolutely continuous with respect to
$\interp{m_{\m{g}}}^{m_{\m{m}}}_{\sigma_o}$;
thus, IS is able to sample any possible latent variables $\sigma_\ell$ in the posterior.
With the importance ratios, IS can be seen as generating $\sigma_\ell$ with density
$w_{\m{g}} \cdot \frac{w_{\m{m}}}{w_{\m{g}}} = w_{\m{m}}$. 
Thus, IS generates a measure proportional to $\interp{m_{\m{m}}}_{\sigma_o}$.

%

\paragraph{Markov-Chain Monte Carlo (MCMC)}
MCMC uses a transition kernel to generate iteratively a new random sample from an old one.
A popular MCMC algorithm is \emph{Metropolis-Hastings} (MH), which constructs the transition
kernel from a \emph{proposal} subroutine.
%
To implement proposal subroutines in our system, we extend the core calculus such that guidance
traces can be used as first-class data.
Then we implement the proposal subroutine as a procedure $g$ whose argument is a guidance trace
on the channel for latent random variables.
The operational rule below formulates a single step in the MH algorithm;
given a proposal procedure $g$, a model $m_\m{m}$, an observation $\sigma_o$,
and the current latent trace $\sigma_\ell$, MH first performs joint execution of $\mcall{g}{\sigma_\ell}$
and $m_\m{m}$ to generate a new latent trace $\sigma_\ell'$ with density $w_{\m{fwd}}$, 
and then uses the new $\sigma_\ell'$ and the old $\sigma_\ell$ to calculate a \emph{backward} density $w_{\m{bwd}}$.
MH then computes an \emph{acceptance ratio}
$\alpha \defeq \min(1, \frac{w'_\m{m} \cdot w_{\m{bwd}}}{w_\m{m} \cdot w_{\m{fwd}}})$, and accepts
the new sample $\sigma_\ell'$ with probability $\alpha$.
\begin{mathpar}\small
  \inferrule
  { \emptyset \mid \varnothing ; \chtype{latent}{\sigma_\ell'} \vdash \mcall{g}{\sigma_\ell} \evalp{w_{\m{fwd}}} \_ \\
    \emptyset \mid \chtype{latent}{\sigma_\ell'} ; \chtype{obs}{\sigma_o} \vdash m_\m{m} \evalp{w'_\m{m}} \_ \\
    \emptyset \mid \varnothing ; \chtype{latent}{\sigma_\ell} \vdash \mcall{g}{\sigma_\ell'} \evalp{w_{\m{bwd}}} \_ \\
    \emptyset \mid \chtype{latent}{\sigma_\ell} ; \chtype{obs}{\sigma_o} \vdash m_\m{m} \evalp{w_\m{m}} \_
  }
  { g ;m_\m{m}; \sigma_o \vdash_{\textsc{mh}} \sigma_\ell \xRightarrow{w_{\m{fwd}} \cdot \alpha}  \sigma_\ell' }
\end{mathpar}
Similar to IS, MH requires that the command $\mcall{g}{\sigma_\ell}$ be able to sample any
possible latent variables $\sigma_\ell'$ in the posterior.
We prove the soundness of MH by a variant of \cref{The:AbsoluteContinuity}, where the programs do
not need to be closed so that they can reference data in the environment (e.g., the old samples).

\paragraph{Variational inference (VI)}
VI uses optimization to find a candidate from an approximating family of guide programs that
minimizes the distance from the posterior distribution to the guide distribution.
We focus on verifying if the distance is well-defined, whereas VI requires extra conditions
for the optimization problem to be well-formed.
Here, we parameterize the guide $m_{\m{g}, \theta}$ by a vector $\theta \in \Theta$ of parameters,
and use KL divergence as the distance, which is defined by
\[
\mathrm{KL}(\mu \parallel \nu) \defeq \int p_{\mu}(\sigma_\ell) (\log p_{\mu}(\sigma_\ell) - \log p_{\nu}(\sigma_\ell)) \lambda_{\interp{A}}(d\sigma_\ell),
\] 
where $\mu$ and $\nu$ are measures on $\interp{A}$ with
densities $p_{\mu}$ and $p_{\nu}$, respectively,
and $\mu$ is absolutely continuous with respect to $\nu$.
The rule below formulates the computation of KL divergence for a specific $\theta$,
via joint execution of the two programs.
\begin{mathpar}\small
  \inferrule
  { \emptyset \mid \varnothing ; \chtype{latent}{\sigma_\ell} \vdash m_{\m{g},\theta} \evalp{w_{\m{g}}} \_ \\
    \emptyset \mid \chtype{latent}{\sigma_\ell} ; \chtype{obs}{\sigma_o} \vdash m_{\m{m}} \evalp{w_\m{m}} \_
  }
  { m_{\m{g},\theta}; m_{\m{m}}; \sigma_o \vdash_{\textsc{vi}}^{w_{\m{g}}}  \tuple{\sigma_\ell,  \log w_{\m{m}} - \log w_{\m{g}} } }
\end{mathpar}
The rule can be seen as defining a map $\sigma_\ell \mapsto w_{\m{g}} \cdot (\log w_{\m{m}} - \log w_{\m{g}})$,
which is the integrand of the divergence $\mathrm{KL}(\interp{m_{\m{g},\theta}}^{m_{\m{m}}}_{\sigma_o} \parallel \interp{m_{\m{m}}}_{\sigma_o})$.
By \cref{The:AbsoluteContinuity}, if the model and guide programs are well-typed,
then $\interp{m_{\m{g},\theta}}^{m_{\m{m}}}_{\sigma_o}$ is absolutely continuous
with respect to $\interp{m_{\m{m}}}_{\sigma_o}$; thus, the KL divergence used in VI
is well-defined.


\section{Experimental Evaluation}
\label{Se:Evaluation}

\paragraph{Implementation}
We implemented the coroutine-based PPL in OCaml.
Our implementation consists of about 2,000 LOC; it contains a parser,
a type checker with automatic inference of guide types, and a prototype
compiler from our PPL to Pyro~\cite{JMLR:BCJ18}.
Our implementation extends the core calculus with tensors (i.e.,
multi-dimensional matrices) and primitive iteration operators for them.
The prototype compiler supports code generation for importance sampling
and variational inference.
We use the Python package \textsf{greenlet}~\cite{misc:greenlet} to support
coroutines in the compiled code.

\paragraph{Evaluation setup}
We evaluated our implementation to answer the following two research
questions:
\begin{enumerate}
  \item How expressive is the coroutine-based PPL, compared to a state-of-the-art
  probabilistic programming language that ensures soundness of programmable inference~\cite{POPL:LCS20}?
  
  \item How efficient is our implementation, in terms of the time for type
  inference, and the performance of Bayesian inference on the compiled code?
\end{enumerate}
%
For the first question, we obtained 23 benchmarks from prior work~\cite{POPL:LCS20}
and collected 6 new benchmarks. The 29 benchmark programs consist of
(i) example models from Anglican~\cite{AISTATS:WMM14}, Turing~\cite{AISTATS:GXG18},
and Pyro~\cite{JMLR:BCJ18}, as well as
(ii) PCFG models, including a Gaussian-process domain-specific language
(DSL)~\cite{POPL:SCS19} and synthetic models (such as examples shown in this paper).
Compared to prior work~\cite{POPL:LCS20}, a larger subset of benchmark models
are expressible and type-checked in our PPL.
\begin{changebar}
Particularly, our PPL is capable of expressing models with recursion and general conditional
branches, whereas prior work~\cite{POPL:LCS20} is not.
\end{changebar}

For the second question, we ran Bayesian inference on the compiled code,
and compared the performance with non-coroutine-based, but equivalent, Pyro code.
We obtained guide programs from where we obtained the benchmark models,
and then reimplemented them in our PPL; for example, we implemented the encoder component
of a variational autoencoder as the guide program~\cite{JMLR:BCJ18}.
For those benchmark models without guides, we first invoked our PPL to type-check the model
program and infer a guide type for the model, and then implemented a guide program whose type
was the guide type.
The compiled model and guide use Pyro's primitives (such as \texttt{pyro.sample}) to sample
random data and condition on given data, as well as exchange messages and switch control with
each other using the concurrent-programming package \texttt{greenlet}.
We leveraged Pyro's inference engines to carry out importance sampling or variational inference.
Type inference is very fast in practice; our implementation completed the
type-inference phase in several milliseconds on all of the benchmarks.
Our experiments showed that coroutines (implemented via messaging passing) do not introduce significant overhead
in actual Bayesian inference.

The experiments were performed on a machine with an Intel Core i7 3.6GHz processor
and 16GB of RAM under macOS Catalina 10.15.7.

\paragraph{Results}
\cref{Tab:Benchmark} gives an overview of selected benchmark models.
Our benchmarks cover a wide range of Bayesian models, such as linear regression,
Gaussian mixtures, hidden Markov models, Bayesian networks, and variational
autoencoders.
Our benchmarks also include the classic Marsaglia algorithm (which generates a
normal distribution from a uniform distribution), a Poisson-trace algorithm
(shown in \cref{Fig:PoissonTrace}, which generates a Poisson distribution from a uniform distribution), and a
Gaussian-process DSL (which uses a PCFG to generate the kernel function of a
Gaussian process).

\begin{figure}
\centering
\begin{small}
\begin{pseudo}
  \kproc \id{Ptrace}($\lambda$) \kw{consume} \id{latent} \kw{provide} \id{obs} = \\+
    $k$ \kgets \kcall \id{PtraceHelper}($e^{-\lambda}$, $0$, $1$); \\
    \ksampleo{\enormal{k}{0.1}}{obs} \\-
  \\
  \kproc \id{PtraceHelper}($l$, $k$, $p$) \kw{consume} \id{latent} \kw{provide} . = \\+
    $u$ \kgets \ksamplei{\eunif}{latent}; \\
    \kifo{latent} $p \cdot u \le l$ \kthen \\+
      \kreturn{k} \\-
    \kelse \\+
      \kcall \id{PtraceHelper}($l$, $k+1$, $p \cdot u$)
\end{pseudo}  
\end{small}
\caption{An algorithm to generate Poisson-distributed numbers given by~\citet{book:KnuthArtVol2}.}
\label{Fig:PoissonTrace}
\end{figure}

\begin{table}
  \centering
  \small
  \caption{Selected benchmark descriptions.
  \textbf{T?} = is type-checked in our PPL;
  \textbf{LOC} = \#lines of code of the model in our PPL;
  \textbf{TP?} = is type-checked by prior work~\cite{POPL:LCS20}.}
  \label{Tab:Benchmark}
  \begin{tabular}{l l c c c}
    \hline
    \bf Program & \bf Description & \bf T? & \bf LOC & \bf TP? \\
    \hline
    \textsf{lr} & Bayesian Linear Regression & \cmark & 16 & \cmark \\
    \textsf{gmm} & Gaussian Mixture Model & \cmark & 44 & \cmark \\
    \textsf{kalman} & Kalman Smoother & \cmark & 32 & \cmark \\
    \textsf{sprinkler} & Bayesian Network & \cmark & 22 & \cmark \\
    \textsf{hmm} & Hidden Markov Model & \cmark & 31 & \cmark \\
    \textsf{branching} & Random Control Flow & \cmark & 19 & \xmark \\
    \textsf{marsaglia} & Marsaglia Algorithm & \cmark & 22 & \xmark \\
    \textsf{dp} & Dirichlet Process &  \xmark & N/A & \xmark \\
    \textsf{ptrace} & Poisson Trace & \cmark & 11 & \xmark \\
    \textsf{aircraft} & Aircraft Detection & \cmark & 32 & \cmark \\
    \textsf{weight} & Unreliable Weigh & \cmark & 8 & \cmark \\
    \textsf{vae} & Variational Autoencoder & \cmark & 26 & \cmark \\
    \textsf{ex-1} & \cref{Fig:ProgAsCoro} & \cmark & 13 & \xmark \\
    \textsf{ex-2} & \cref{Fig:RecursiveModel} & \cmark & 21 & \xmark \\
    \textsf{gp-dsl} & Gaussian Process DSL & \cmark & 58 & \xmark \\
    \hline
  \end{tabular}
\end{table}

As shown in \cref{Tab:Benchmark}, our coroutine-based PPL is capable of
expressing most of the benchmarks, except those involving stochastic
memoization~\cite{UAI:GMR08}, such as the program \textsf{dp}.
The programs \textsf{branching}, \textsf{marsaglia}, \textsf{ptrace}, and \textsf{ex-1}
have nontrivial branching, and the programs \textsf{marsaglia}, \textsf{ptrace},
\textsf{ex-2}, and \textsf{gp-dsl} define recursive models;
our implementation successfully inferred guide types for these programs,
whereas prior work~\cite{POPL:LCS20} could not express them.
Our implementation derived guide types for 25 of the 29 benchmarks,
whereas prior work was able to express only 18 of them.

\begin{changebar}
For all the benchmarks, we assume that each guide program
samples random variables in the \emph{same} order as its corresponding model program does.
However, this assumption can sometimes be too restrictive: it has been shown that
the ability to allow the model and the guide to sample random variables in \emph{different}
orders is desirable for inference amortization methods~\cite{NIPS:WGZ18}.
Prior work~\cite{POPL:LCS20} allows different sampling orders in the model and the guide,
whereas our system cannot handle such scenarios.
\end{changebar}

\begin{table}
  \centering
  \small
  \caption{Selected performance statistics.
  \textbf{BI} = Bayesian-inference algorithm (IS or VI);
  \textbf{CG (ms)} = time for type inference and code generation in milliseconds;
  \textbf{GLOC} = \#lines of code in compiled code (model + guide);
  \textbf{GI (s)} = time for Bayesian inference on compiled code in seconds;
  \textbf{HLOC} = \#lines of code in handwritten code (model + guide);
  \textbf{HI (s)} = time for Bayesian inference on handwritten code in seconds.}
  \label{Tab:Evaluation}
  \begin{tabular}{l | c c c c | c c}
    \hline
    \bf Program & \bf BI & \bf CG (ms) & \bf GLOC & \bf GI (s) & \bf HLOC & \bf HI (s) \\
    \hline
    \textsf{ex-1} & IS & 0.75 & 57 & 5.44 & 16 & 5.27 \\
    \textsf{branching} & IS & 1.74 & 58 & 8.49 & 16 & 7.48 \\
    \textsf{gmm} & IS & 8.03 & 185 & 64.13 & 38 & 56.00 \\
    \hline
    \textsf{weight} & VI & 0.66 & 35 & 2.76 & 7 & 2.66 \\
    \textsf{vae} & VI & 10.36 & 72 & 34.96 & 26 & 32.69 \\
    \hline
  \end{tabular}
\end{table}

\cref{Tab:Evaluation} presents performance statistics of selected benchmark programs.
We evaluated our PPL's performance under two criteria:
(i) the time for type inference and code generation, and
(ii) the time for Bayesian inference compared to handwritten inference code under the
same set of hyperparameters (e.g., iteration rounds, optimization algorithms, and initial
values of parameters).
Our experiments showed that our implementation usually completes type inference and
code generation in several milliseconds, and the compiled code, although using coroutines,
has similar performance to handwritten inference code.


\section{Related Work}
\label{Se:RelatedWork}

\paragraph{Sound Bayesian inference}
Most closely related to our work are techniques for reasoning about soundness
of trace-based programmable inference.
\citet{POPL:LYR19} developed a static analysis of stochastic variational inference
with guide programs, which describe custom approximating families in Pyro.
\begin{changebar}
Their analysis supports nontrivial features of Pyro, such as tensor manipulation and
\emph{plates}, i.e.,
\end{changebar}
vectors of conditionally independent samples.
Their approach aims at proving that the model and guide programs have
the same support and satisfy differentiability-related conditions.
Their static analysis does not handle the case when a conditional statement
determines the set of random samples.
\citet{POPL:LCS20} proposed \emph{trace types} as precise signatures for
sampling traces of probabilistic programs, and then used the type system
to prove absolute continuity in multiple kinds of inference algorithms.
%
%
Trace types can be seen as a type-and-effect system, where a trace type records
the precise set of samples drawn by a single program.
\begin{changebar}
Trace types support higher-order functions, \emph{stochastic} branches that can influence the set of random samples,
as well as three forms of loops, including stochastic while-loops with an unbounded number of iterations, but not general recursion.
Because the value
of a conditional predicate cannot be determined in general at static-analysis time, trace
types do not support general conditional statements that can influence the set of random samples.
Both \citet{POPL:LYR19}'s and \citet{POPL:LCS20}'s approach allow the model and the guide to sample
random variables in different orders.
\end{changebar}
In this paper, we propose a new PPL that guarantees absolute
continuity between a model-guide pair, and features general programming
constructs, including recursion and branching.
A key innovation of our work is the coroutine-based paradigm of writing
inference code;
this paradigm makes the relational reasoning of the support-match property
explicit, and in particular enables precise analysis of complex control flow.
\begin{changebar}
However, compared to prior work, our system only supports
scenarios where the model and the guide sample random variables in the same order.
\end{changebar}

There has been a line of work on validating Monte-Carlo inference algorithms.
\citet{POPL:SKV17} developed a semantic framework to verify the
soundness of Monte-Carlo inference algorithms with generic proposal distributions. 
\citet{misc:AYC18} presented a type system for verifying hand-coded Monte-Carlo
algorithms that explicitly manipulate densities, rather than use proposal distributions.
For MCMC methods, \citet{ICFP:BLG16} and \citet{LIPICS:HNR15} developed provably correct MH algorithms.
\citet{ESOP:CP19} proposed an \emph{intensional} semantics, which captures
execution traces of programs, to validate an incremental MH algorithm.
Several systems~\cite{misc:AYC18,POPL:LCS20,PLDI:HTM17,phd:Bonawitz08} studied
sound combinators for kernels used by MCMC.
In contrast to the aforementioned work, our PPL is based on trace-based programmable inference.
It would be interesting to
develop programmable versions of those sound inference algorithms
in our PPL.

\citet{FLOPS:NCR16} and \citet{UAI:ZS17} validated the soundness of program transformations
in Hakaru, which contains a programmable MH algorithm.
The development of Hakaru is not centered around sample traces, and it
uses \emph{symbolic disintegration}~\cite{POPL:SR17,NSSOR:CP97} to calculate
the marginal densities for computing the acceptance ratio in an MH step.
In this paper, we focus on a trace-based scheme for programmable inference.
Establishing the relationship among different schemes of programmable inference
is an interesting future research direction.

\paragraph{Session types}
Honda et al.~\cite{CONCUR:Honda93,ESOP:HVK98} introduced session types to
prescribe binary communication protocols for message-passing
processes.
Session types can be interpreted either classically~\cite{ICFP:Wadler12}, or
intuitionistically~\cite{CONCUR:CP10,JMSCS:CPT16}.
To enable non-binary communication, researchers proposed
multiparty session types~\cite{POPL:HYC08,SFM:CDM15,POPL:SY19}.
%
The tail-recursive structure of standard session types imposes communication
protocols that can be described by a \emph{regular} language.
Recently, several systems have been developed to go beyond tail-recursive protocols,
such as context-free~\cite{ICFP:TV16}, label-dependent~\cite{POPL:TV19}, and nested~\cite{misc:DDM20}
session types.

In our development of guide types, we took inspiration from the structuring principle
of session types.
Compared to session types, guide types have different semantics (i.e., sending
and receiving random samples drawn from probability distributions),
have simpler forms (i.e., no process spawning or higher-order channels),
and enjoy an efficient type-inference algorithm, which can also analyze
non-tail-recursive communication protocols.
Developing a truly concurrent probabilistic programming system, and
concurrent Bayesian inference algorithms with general session types,
would be interesting future work.

\section{Conclusion}
\label{Se:Conclusion}

We have presented a new probabilistic programming language that supports
programmable Bayesian inference, and guarantees model-guide absolute continuity,
thereby ensuring key soundness properties of multiple kinds of inference algorithms.
Our language implements the model and guide programs as \emph{coroutines},
and we develop \emph{guide types} to prescribe the communication protocols between
coroutines.
We have proved that well-typed model and guide coroutines execute safely, and they
are guaranteed to enjoy absolute continuity.
We have also developed an efficient type-inference algorithm that reconstructs
guide types directly from the code.
Finally, we have implemented our language with a prototype compiler to Pyro, and
evaluated our implementation on a suite of diverse probabilistic models.


\begin{acks}
This article is based on research supported, in part, by a gift from Rajiv and Ritu Batra;
by ONR under grants N00014-17-1-2889 and N00014-19-1-2318;
by DARPA under AA contract FA8750-18-C0092;
and by the NSF under SaTC award 1801369, SHF awards 1812876 and 2007784, and CAREER
award 1845514.
Any opinions, findings, and conclusions or recommendations
expressed in this publication are those of the authors,
and do not necessarily reflect the views of the sponsoring
agencies.
\end{acks}

\balance
\bibliography{db,misc}


\begin{thebibliography}{63}


\ifx \showCODEN    \undefined \def \showCODEN     #1{\unskip}     \fi
\ifx \showDOI      \undefined \def \showDOI       #1{#1}\fi
\ifx \showISBNx    \undefined \def \showISBNx     #1{\unskip}     \fi
\ifx \showISBNxiii \undefined \def \showISBNxiii  #1{\unskip}     \fi
\ifx \showISSN     \undefined \def \showISSN      #1{\unskip}     \fi
\ifx \showLCCN     \undefined \def \showLCCN      #1{\unskip}     \fi
\ifx \shownote     \undefined \def \shownote      #1{#1}          \fi
\ifx \showarticletitle \undefined \def \showarticletitle #1{#1}   \fi
\ifx \showURL      \undefined \def \showURL       {\relax}        \fi
\providecommand\bibfield[2]{#2}
\providecommand\bibinfo[2]{#2}
\providecommand\natexlab[1]{#1}
\providecommand\showeprint[2][]{arXiv:#2}

\bibitem[\protect\citeauthoryear{Ai, Arora, Dong, Gokkaya, Jiang, Kubendran,
  Kumar, Tingley, and Torabi}{Ai et~al\mbox{.}}{2019}]%
        {MAPL:AAD19}
\bibfield{author}{\bibinfo{person}{Jessica Ai}, \bibinfo{person}{Nimar~S.
  Arora}, \bibinfo{person}{Ning Dong}, \bibinfo{person}{Beliz Gokkaya},
  \bibinfo{person}{Thomas Jiang}, \bibinfo{person}{Anitha Kubendran},
  \bibinfo{person}{Arun Kumar}, \bibinfo{person}{Michael Tingley}, {and}
  \bibinfo{person}{Narjes Torabi}.} \bibinfo{year}{2019}\natexlab{}.
\newblock \showarticletitle{{HackPPL: A Universal Probabilistic Programming
  Language}}. In \bibinfo{booktitle}{\emph{Int.\ Workshop on Machine Learning
  and Prog.\ Lang.}} \emph{(\bibinfo{series}{MAPL'19})}.
\newblock
\urldef\tempurl%
\url{https://doi.org/10.1145/3315508.3329974}
\showDOI{\tempurl}


\bibitem[\protect\citeauthoryear{Anton and Thiemann}{Anton and
  Thiemann}{2010a}]%
        {APLAS:AT10}
\bibfield{author}{\bibinfo{person}{Konrad Anton} {and} \bibinfo{person}{Peter
  Thiemann}.} \bibinfo{year}{2010}\natexlab{a}.
\newblock \showarticletitle{{Towards Deriving Type Systems and Implementations
  for Coroutines}}. In \bibinfo{booktitle}{\emph{Asian Symp.\ on Prog.\ Lang.\
  and Systems}} \emph{(\bibinfo{series}{APLAS'10})}.
\newblock
\urldef\tempurl%
\url{https://doi.org/10.1007/978-3-642-17164-2_6}
\showDOI{\tempurl}


\bibitem[\protect\citeauthoryear{Anton and Thiemann}{Anton and
  Thiemann}{2010b}]%
        {TFP:AT10}
\bibfield{author}{\bibinfo{person}{Konrad Anton} {and} \bibinfo{person}{Peter
  Thiemann}.} \bibinfo{year}{2010}\natexlab{b}.
\newblock \showarticletitle{{Typing Coroutines}}. In
  \bibinfo{booktitle}{\emph{Trends in Functional Programming}}
  \emph{(\bibinfo{series}{TFP'10})}.
\newblock
\urldef\tempurl%
\url{https://doi.org/10.1007/978-3-642-22941-1_2}
\showDOI{\tempurl}


\bibitem[\protect\citeauthoryear{Atkinson, Yang, and Carbin}{Atkinson
  et~al\mbox{.}}{2018}]%
        {misc:AYC18}
\bibfield{author}{\bibinfo{person}{Eric Atkinson}, \bibinfo{person}{Cambridge
  Yang}, {and} \bibinfo{person}{Michael Carbin}.}
  \bibinfo{year}{2018}\natexlab{}.
\newblock \bibinfo{title}{{Verifying Handcoded Probabilistic Inference
  Procedures}}.
\newblock
\newblock
\urldef\tempurl%
\url{https://arxiv.org/abs/1805.01863}
\showURL{%
\tempurl}


\bibitem[\protect\citeauthoryear{Bhat, Agarwal, Vuduc, and Gray}{Bhat
  et~al\mbox{.}}{2012}]%
        {POPL:BAV12}
\bibfield{author}{\bibinfo{person}{Sooraj Bhat}, \bibinfo{person}{Ashish
  Agarwal}, \bibinfo{person}{Richard Vuduc}, {and} \bibinfo{person}{Alexander
  Gray}.} \bibinfo{year}{2012}\natexlab{}.
\newblock \showarticletitle{{A Type Theory for Probability Density Functions}}.
  In \bibinfo{booktitle}{\emph{Princ.\ of Prog.\ Lang.}}
  \emph{(\bibinfo{series}{POPL'12})}.
\newblock
\urldef\tempurl%
\url{https://doi.org/10.1145/2103656.2103721}
\showDOI{\tempurl}


\bibitem[\protect\citeauthoryear{Bhat, Borgstr{\"o}m, Gordon, and Russo}{Bhat
  et~al\mbox{.}}{2013}]%
        {TACAS:BBG13}
\bibfield{author}{\bibinfo{person}{Sooraj Bhat}, \bibinfo{person}{Johannes
  Borgstr{\"o}m}, \bibinfo{person}{Andrew~D. Gordon}, {and}
  \bibinfo{person}{Claudio Russo}.} \bibinfo{year}{2013}\natexlab{}.
\newblock \showarticletitle{{Deriving Probability Density Functions from
  Probabilistic Functional Programs}}. In \bibinfo{booktitle}{\emph{Tools and
  Algs.\ for the Construct.\ and Anal.\ of Syst.}}
  \emph{(\bibinfo{series}{TACAS'13})}.
\newblock
\urldef\tempurl%
\url{https://doi.org/10.1007/978-3-642-36742-7_35}
\showDOI{\tempurl}


\bibitem[\protect\citeauthoryear{Billingsley}{Billingsley}{2012}]%
        {book:Billingsley12}
\bibfield{author}{\bibinfo{person}{Patrick Billingsley}.}
  \bibinfo{year}{2012}\natexlab{}.
\newblock \bibinfo{booktitle}{\emph{{Probability and Measure}}}.
\newblock \bibinfo{publisher}{John Wiley \& Sons, Inc.}
\newblock


\bibitem[\protect\citeauthoryear{Bingham, Chen, Jankowiak, Obermeyer, Pradhan,
  Karaletsos, Singh, Szerlip, Horsfall, and Goodman}{Bingham
  et~al\mbox{.}}{2018}]%
        {JMLR:BCJ18}
\bibfield{author}{\bibinfo{person}{Eli Bingham}, \bibinfo{person}{Jonathan~P.
  Chen}, \bibinfo{person}{Martin Jankowiak}, \bibinfo{person}{Fritz Obermeyer},
  \bibinfo{person}{Neeraj Pradhan}, \bibinfo{person}{Theofanis Karaletsos},
  \bibinfo{person}{Rishabh Singh}, \bibinfo{person}{Paul Szerlip},
  \bibinfo{person}{Paul Horsfall}, {and} \bibinfo{person}{Noah~D. Goodman}.}
  \bibinfo{year}{2018}\natexlab{}.
\newblock \showarticletitle{{Pyro: Deep Universal Probabilistic Programming}}.
\newblock \bibinfo{journal}{\emph{J.\ Machine Learning Research}}
  \bibinfo{volume}{20}, \bibinfo{number}{1} (\bibinfo{date}{January}
  \bibinfo{year}{2018}).
\newblock
\urldef\tempurl%
\url{https://dl.acm.org/doi/10.5555/3322706.3322734}
\showURL{%
\tempurl}


\bibitem[\protect\citeauthoryear{Bonawitz}{Bonawitz}{2008}]%
        {phd:Bonawitz08}
\bibfield{author}{\bibinfo{person}{Keith~A. Bonawitz}.}
  \bibinfo{year}{2008}\natexlab{}.
\newblock \emph{\bibinfo{title}{{Composable Probabilistic Inference with
  Blaise}}}.
\newblock \bibinfo{thesistype}{Ph.D. Dissertation}.
  \bibinfo{school}{Massachusetts Institute of Technology}.
\newblock


\bibitem[\protect\citeauthoryear{Borgstr{\"o}m, Lago, Gordon, and
  Szymczak}{Borgstr{\"o}m et~al\mbox{.}}{2016}]%
        {ICFP:BLG16}
\bibfield{author}{\bibinfo{person}{Johannes Borgstr{\"o}m},
  \bibinfo{person}{Ugo~Dal Lago}, \bibinfo{person}{Andrew~D. Gordon}, {and}
  \bibinfo{person}{Marcin Szymczak}.} \bibinfo{year}{2016}\natexlab{}.
\newblock \showarticletitle{{A Lambda-Calculus Foundation for Universal
  Probabilistic Programming}}. In \bibinfo{booktitle}{\emph{Int.\ Conf.\ on
  Functional Programming}} \emph{(\bibinfo{series}{ICFP'16})}.
\newblock
\urldef\tempurl%
\url{https://doi.org/10.1145/2951913.2951942}
\showDOI{\tempurl}


\bibitem[\protect\citeauthoryear{Caires and Pfenning}{Caires and
  Pfenning}{2010}]%
        {CONCUR:CP10}
\bibfield{author}{\bibinfo{person}{Lu{\'i}s Caires} {and}
  \bibinfo{person}{Frank Pfenning}.} \bibinfo{year}{2010}\natexlab{}.
\newblock \showarticletitle{{Session Types as Intuitionistic Linear
  Propositions}}. In \bibinfo{booktitle}{\emph{Int.\ Conf.\ on Concurrency
  Theory}} \emph{(\bibinfo{series}{CONCUR'10})}.
\newblock
\urldef\tempurl%
\url{https://doi.org/10.1007/978-3-642-15375-4_16}
\showDOI{\tempurl}


\bibitem[\protect\citeauthoryear{Caires, Pfenning, and Toninho}{Caires
  et~al\mbox{.}}{2016}]%
        {JMSCS:CPT16}
\bibfield{author}{\bibinfo{person}{Lu{\'i}s Caires}, \bibinfo{person}{Frank
  Pfenning}, {and} \bibinfo{person}{Bernardo Toninho}.}
  \bibinfo{year}{2016}\natexlab{}.
\newblock \showarticletitle{{Linear Logic Propositions as Session Types}}.
\newblock \bibinfo{journal}{\emph{Math.\ Struct.\ Comp.\ Sci.}}
  \bibinfo{volume}{26}, \bibinfo{number}{3} (\bibinfo{date}{March}
  \bibinfo{year}{2016}).
\newblock
\urldef\tempurl%
\url{https://doi.org/10.1017/S0960129514000218}
\showURL{%
\tempurl}


\bibitem[\protect\citeauthoryear{Carpenter, Gelman, Hoffman, Lee, Goodrich,
  Betancourt, Brubaker, Guo, Li, and Riddell}{Carpenter et~al\mbox{.}}{2017}]%
        {JSS:CGH17}
\bibfield{author}{\bibinfo{person}{Bob Carpenter}, \bibinfo{person}{Andrew
  Gelman}, \bibinfo{person}{Matthew~D. Hoffman}, \bibinfo{person}{Daniel Lee},
  \bibinfo{person}{Ben Goodrich}, \bibinfo{person}{Michael Betancourt},
  \bibinfo{person}{Marcus Brubaker}, \bibinfo{person}{Jiqiang Guo},
  \bibinfo{person}{Peter Li}, {and} \bibinfo{person}{Allen Riddell}.}
  \bibinfo{year}{2017}\natexlab{}.
\newblock \showarticletitle{{Stan: A Probabilistic Programming Language}}.
\newblock \bibinfo{journal}{\emph{J.\ Statistical Softw.}}
  \bibinfo{volume}{76}, \bibinfo{number}{1} (\bibinfo{year}{2017}).
\newblock
\urldef\tempurl%
\url{https://doi.org/10.18637/jss.v076.i01}
\showDOI{\tempurl}


\bibitem[\protect\citeauthoryear{Castellan and Paquet}{Castellan and
  Paquet}{2019}]%
        {ESOP:CP19}
\bibfield{author}{\bibinfo{person}{Simon Castellan} {and} \bibinfo{person}{Hugo
  Paquet}.} \bibinfo{year}{2019}\natexlab{}.
\newblock \showarticletitle{{Probabilistic Programming Inference via
  Intensional Semantics}}. In \bibinfo{booktitle}{\emph{European Symp.\ on
  Programming}} \emph{(\bibinfo{series}{ESOP'19})}.
\newblock
\urldef\tempurl%
\url{https://doi.org/10.1007/978-3-030-17184-1_12}
\showDOI{\tempurl}


\bibitem[\protect\citeauthoryear{Chang and Pollard}{Chang and Pollard}{1997}]%
        {NSSOR:CP97}
\bibfield{author}{\bibinfo{person}{J.~T. Chang} {and} \bibinfo{person}{D.
  Pollard}.} \bibinfo{year}{1997}\natexlab{}.
\newblock \showarticletitle{{Conditioning as disintegration}}.
\newblock \bibinfo{journal}{\emph{Netherlands Society for Statistics and
  Operations Research}} \bibinfo{volume}{51}, \bibinfo{number}{3}
  (\bibinfo{date}{November} \bibinfo{year}{1997}).
\newblock
\urldef\tempurl%
\url{https://doi.org/10.1111/1467-9574.00056}
\showDOI{\tempurl}


\bibitem[\protect\citeauthoryear{Coppo, Dezani-Ciancaglini, Padovani, and
  Yoshida}{Coppo et~al\mbox{.}}{2015}]%
        {SFM:CDM15}
\bibfield{author}{\bibinfo{person}{Mario Coppo}, \bibinfo{person}{Mariangiola
  Dezani-Ciancaglini}, \bibinfo{person}{Luca Padovani}, {and}
  \bibinfo{person}{Nobuko Yoshida}.} \bibinfo{year}{2015}\natexlab{}.
\newblock \showarticletitle{{A Gentle Introduction to Multiparty Asynchronous
  Session Types}}. In \bibinfo{booktitle}{\emph{Formal Methods for Eternal
  Networked Software Systems}} \emph{(\bibinfo{series}{SFM'15})}.
\newblock
\urldef\tempurl%
\url{https://doi.org/10.1007/978-3-319-18941-3_4}
\showDOI{\tempurl}


\bibitem[\protect\citeauthoryear{Cusumano-Towner, Saad, Lew, and
  Mansinghka}{Cusumano-Towner et~al\mbox{.}}{2019}]%
        {PLDI:CSL19}
\bibfield{author}{\bibinfo{person}{Marco~F. Cusumano-Towner},
  \bibinfo{person}{Feras~A. Saad}, \bibinfo{person}{Alexander~K. Lew}, {and}
  \bibinfo{person}{Vikash~K. Mansinghka}.} \bibinfo{year}{2019}\natexlab{}.
\newblock \showarticletitle{{Gen: A General-Purpose Probabilistic Programming
  System with Programmable Inference}}. In \bibinfo{booktitle}{\emph{Prog.\
  Lang.\ Design and Impl.}} \emph{(\bibinfo{series}{PLDI'19})}.
\newblock
\urldef\tempurl%
\url{https://doi.org/10.1145/3314221.3314642}
\showDOI{\tempurl}


\bibitem[\protect\citeauthoryear{Das, DeYoung, Mordido, and Pfenning}{Das
  et~al\mbox{.}}{2020}]%
        {misc:DDM20}
\bibfield{author}{\bibinfo{person}{Ankush Das}, \bibinfo{person}{Henry
  DeYoung}, \bibinfo{person}{Andreia Mordido}, {and} \bibinfo{person}{Frank
  Pfenning}.} \bibinfo{year}{2020}\natexlab{}.
\newblock \bibinfo{title}{{Nested Session Types}}.
\newblock
\newblock
\urldef\tempurl%
\url{https://arxiv.org/abs/2010.06482}
\showURL{%
\tempurl}


\bibitem[\protect\citeauthoryear{Foster, Jankowiak, Bingham, Horsfall, Teh,
  Rainforth, and Goodman}{Foster et~al\mbox{.}}{2019}]%
        {NIPS:FJB19}
\bibfield{author}{\bibinfo{person}{Adam Foster}, \bibinfo{person}{Martin
  Jankowiak}, \bibinfo{person}{Eli Bingham}, \bibinfo{person}{Paul Horsfall},
  \bibinfo{person}{Yee~Whye Teh}, \bibinfo{person}{Tom Rainforth}, {and}
  \bibinfo{person}{Noah~D. Goodman}.} \bibinfo{year}{2019}\natexlab{}.
\newblock \showarticletitle{{Variational Bayesian Optimal Experimental Design:
  Efficient Automation of Adaptive Experiments}}. In
  \bibinfo{booktitle}{\emph{Neural Info.\ Processing Syst.}}
  \emph{(\bibinfo{series}{NIPS'19})}.
\newblock
\urldef\tempurl%
\url{https://arxiv.org/abs/1903.05480}
\showURL{%
\tempurl}


\bibitem[\protect\citeauthoryear{Ge, Xu, and Ghahramani}{Ge
  et~al\mbox{.}}{2018}]%
        {AISTATS:GXG18}
\bibfield{author}{\bibinfo{person}{Rong Ge}, \bibinfo{person}{Kai Xu}, {and}
  \bibinfo{person}{Zoubin Ghahramani}.} \bibinfo{year}{2018}\natexlab{}.
\newblock \showarticletitle{{Turing: A Language for Flexible Probabilistic
  Inference}}. In \bibinfo{booktitle}{\emph{Artificial Intelligence and
  Statistics}} \emph{(\bibinfo{series}{AISTATS'18})}.
\newblock


\bibitem[\protect\citeauthoryear{Gelman, Carlin, Stern, Dunson, Vehtari, and
  Rubin}{Gelman et~al\mbox{.}}{2013}]%
        {book:GCS13}
\bibfield{author}{\bibinfo{person}{Andrew Gelman}, \bibinfo{person}{John~B.
  Carlin}, \bibinfo{person}{Hal~S. Stern}, \bibinfo{person}{David~B. Dunson},
  \bibinfo{person}{Aki Vehtari}, {and} \bibinfo{person}{Donald~B. Rubin}.}
  \bibinfo{year}{2013}\natexlab{}.
\newblock \bibinfo{booktitle}{\emph{{Bayesian Data Analysis}}}.
\newblock \bibinfo{publisher}{Chapman and Hall/CRC}.
\newblock
\urldef\tempurl%
\url{https://doi.org/10.1201/b16018}
\showDOI{\tempurl}


\bibitem[\protect\citeauthoryear{Ghahramani}{Ghahramani}{2015}]%
        {NATURE:Ghahramani15}
\bibfield{author}{\bibinfo{person}{Zoubin Ghahramani}.}
  \bibinfo{year}{2015}\natexlab{}.
\newblock \showarticletitle{{Probabilistic machine learning and artificial
  intelligence}}.
\newblock \bibinfo{journal}{\emph{Nature}}  \bibinfo{volume}{521}
  (\bibinfo{date}{May} \bibinfo{year}{2015}).
\newblock
\urldef\tempurl%
\url{https://doi.org/10.1038/nature14541}
\showDOI{\tempurl}


\bibitem[\protect\citeauthoryear{Gilks, Thomas, and Spiegelhalter}{Gilks
  et~al\mbox{.}}{1994}]%
        {JRSS:GTS97}
\bibfield{author}{\bibinfo{person}{W.~R. Gilks}, \bibinfo{person}{A. Thomas},
  {and} \bibinfo{person}{D.~J. Spiegelhalter}.}
  \bibinfo{year}{1994}\natexlab{}.
\newblock \showarticletitle{{A Language and Program for Complex Bayesian
  Modelling}}.
\newblock \bibinfo{journal}{\emph{J.\ Royal Statistical Society}}
  \bibinfo{volume}{43}, \bibinfo{number}{1} (\bibinfo{date}{January}
  \bibinfo{year}{1994}).
\newblock
\urldef\tempurl%
\url{https://doi.org/10.2307/2348941}
\showDOI{\tempurl}


\bibitem[\protect\citeauthoryear{Giry}{Giry}{1982}]%
        {CATA:Giry82}
\bibfield{author}{\bibinfo{person}{Mich{\`e}le Giry}.}
  \bibinfo{year}{1982}\natexlab{}.
\newblock \showarticletitle{{A Categorical Approach to Probability Theory}}. In
  \bibinfo{booktitle}{\emph{Categorical Aspects of Topology and Analysis}}.
\newblock
\urldef\tempurl%
\url{https://doi.org/10.1007/BFb0092872}
\showDOI{\tempurl}


\bibitem[\protect\citeauthoryear{Goodman, Mansinghka, Roy, Bonawitz, and
  Tenenbaum}{Goodman et~al\mbox{.}}{2008}]%
        {UAI:GMR08}
\bibfield{author}{\bibinfo{person}{Noah~D. Goodman}, \bibinfo{person}{Vikash~K.
  Mansinghka}, \bibinfo{person}{Daniel Roy}, \bibinfo{person}{Keith~A.
  Bonawitz}, {and} \bibinfo{person}{Joshua~B. Tenenbaum}.}
  \bibinfo{year}{2008}\natexlab{}.
\newblock \showarticletitle{{Church: A language for generative models}}. In
  \bibinfo{booktitle}{\emph{Uncertainty in Artificial Intelligence}}
  \emph{(\bibinfo{series}{UAI'08})}.
\newblock
\urldef\tempurl%
\url{https://dl.acm.org/doi/10.5555/3023476.3023503}
\showURL{%
\tempurl}


\bibitem[\protect\citeauthoryear{Goodman and Stuhlm{\"u}ller}{Goodman and
  Stuhlm{\"u}ller}{2014}]%
        {misc:dippl}
\bibfield{author}{\bibinfo{person}{Noah~D. Goodman} {and}
  \bibinfo{person}{Andreas Stuhlm{\"u}ller}.} \bibinfo{year}{2014}\natexlab{}.
\newblock \bibinfo{title}{{The Design and Implementation of Probabilistic
  Programming Languages}}.
\newblock \bibinfo{howpublished}{Available on \url{http://dippl.org}}.
\newblock


\bibitem[\protect\citeauthoryear{Green}{Green}{1995}]%
        {BIOMETRIKA:Green95}
\bibfield{author}{\bibinfo{person}{Peter~J. Green}.}
  \bibinfo{year}{1995}\natexlab{}.
\newblock \showarticletitle{{Reversible Jump Markov Chain Monte Carlo
  Computation and Bayesian Model Determination}}.
\newblock \bibinfo{journal}{\emph{Biometrika}} \bibinfo{volume}{82},
  \bibinfo{number}{4} (\bibinfo{date}{December} \bibinfo{year}{1995}).
\newblock
\urldef\tempurl%
\url{https://doi.org/10.2307/2337340}
\showDOI{\tempurl}


\bibitem[\protect\citeauthoryear{Griffiths, Kemp, and Tenenbaum}{Griffiths
  et~al\mbox{.}}{2008}]%
        {kn:GKT08}
\bibfield{author}{\bibinfo{person}{Thomas~L. Griffiths},
  \bibinfo{person}{Charles Kemp}, {and} \bibinfo{person}{Joshua~B. Tenenbaum}.}
  \bibinfo{year}{2008}\natexlab{}.
\newblock \showarticletitle{{Bayesian Models of Cognition}}.
\newblock In \bibinfo{booktitle}{\emph{{The Cambridge Handbook of Computational
  Psychology}}}. \bibinfo{publisher}{Cambridge University Press}.
\newblock
\urldef\tempurl%
\url{https://doi.org/10.1017/CBO9780511816772.006}
\showDOI{\tempurl}


\bibitem[\protect\citeauthoryear{Harper}{Harper}{2016}]%
        {book:PFPL16}
\bibfield{author}{\bibinfo{person}{Robert Harper}.}
  \bibinfo{year}{2016}\natexlab{}.
\newblock \bibinfo{booktitle}{\emph{{Practical Foundations for Programming
  Languages}}}.
\newblock \bibinfo{publisher}{Cambridge University Press}.
\newblock
\urldef\tempurl%
\url{https://dl.acm.org/doi/book/10.5555/3002812}
\showURL{%
\tempurl}


\bibitem[\protect\citeauthoryear{Hoare}{Hoare}{1978}]%
        {CACM:Hoare78}
\bibfield{author}{\bibinfo{person}{C.~A.~R. Hoare}.}
  \bibinfo{year}{1978}\natexlab{}.
\newblock \showarticletitle{{Communicating Sequential Processes}}.
\newblock \bibinfo{journal}{\emph{Commun.\ ACM}} \bibinfo{volume}{21},
  \bibinfo{number}{8} (\bibinfo{date}{August} \bibinfo{year}{1978}).
\newblock
\urldef\tempurl%
\url{https://doi.org/10.1145/359576.359585}
\showDOI{\tempurl}


\bibitem[\protect\citeauthoryear{Honda}{Honda}{1993}]%
        {CONCUR:Honda93}
\bibfield{author}{\bibinfo{person}{Kohei Honda}.}
  \bibinfo{year}{1993}\natexlab{}.
\newblock \showarticletitle{{Types for Dyadic Interaction}}. In
  \bibinfo{booktitle}{\emph{Int.\ Conf.\ on Concurrency Theory}}
  \emph{(\bibinfo{series}{CONCUR'93})}.
\newblock
\urldef\tempurl%
\url{https://doi.org/10.1007/3-540-57208-2_35}
\showDOI{\tempurl}


\bibitem[\protect\citeauthoryear{Honda, Vasconcelos, and Kubo}{Honda
  et~al\mbox{.}}{1998}]%
        {ESOP:HVK98}
\bibfield{author}{\bibinfo{person}{Kohei Honda}, \bibinfo{person}{Vasco~T.
  Vasconcelos}, {and} \bibinfo{person}{Makoto Kubo}.}
  \bibinfo{year}{1998}\natexlab{}.
\newblock \showarticletitle{{Language Primitives and Type Discipline for
  Structured Communication-Based Programming}}. In
  \bibinfo{booktitle}{\emph{European Symp.\ on Programming}}
  \emph{(\bibinfo{series}{ESOP'98})}.
\newblock
\urldef\tempurl%
\url{https://doi.org/10.1007/BFb0053567}
\showDOI{\tempurl}


\bibitem[\protect\citeauthoryear{Honda, Yoshida, and Carbone}{Honda
  et~al\mbox{.}}{2008}]%
        {POPL:HYC08}
\bibfield{author}{\bibinfo{person}{Kohei Honda}, \bibinfo{person}{Nobuko
  Yoshida}, {and} \bibinfo{person}{Marco Carbone}.}
  \bibinfo{year}{2008}\natexlab{}.
\newblock \showarticletitle{{Multiparty Asynchronous Session Types}}. In
  \bibinfo{booktitle}{\emph{Princ.\ of Prog.\ Lang.}}
  \emph{(\bibinfo{series}{POPL'08})}.
\newblock
\urldef\tempurl%
\url{https://doi.org/10.1145/1328438.1328472}
\showDOI{\tempurl}


\bibitem[\protect\citeauthoryear{Huang, Tristan, and Morrisett}{Huang
  et~al\mbox{.}}{2017}]%
        {PLDI:HTM17}
\bibfield{author}{\bibinfo{person}{Daniel Huang},
  \bibinfo{person}{Jean-Baptiste Tristan}, {and} \bibinfo{person}{Greg
  Morrisett}.} \bibinfo{year}{2017}\natexlab{}.
\newblock \showarticletitle{{Compiling Markov Chain Monte Carlo Algorithms for
  Probabilistic Modeling}}. In \bibinfo{booktitle}{\emph{Prog.\ Lang.\ Design
  and Impl.}} \emph{(\bibinfo{series}{PLDI'17})}.
\newblock
\urldef\tempurl%
\url{https://doi.org/10.1145/3062341.3062375}
\showDOI{\tempurl}


\bibitem[\protect\citeauthoryear{Hur, Nori, Rajamani, and Samuel}{Hur
  et~al\mbox{.}}{2015}]%
        {LIPICS:HNR15}
\bibfield{author}{\bibinfo{person}{Chung-Kil Hur}, \bibinfo{person}{Aditya~V.
  Nori}, \bibinfo{person}{Sriram~K. Rajamani}, {and} \bibinfo{person}{Selva
  Samuel}.} \bibinfo{year}{2015}\natexlab{}.
\newblock \showarticletitle{{A Provably Correct Sampler for Probabilistic
  Programs}}. In \bibinfo{booktitle}{\emph{Leibniz International Proceedings in
  Informatics}} \emph{(\bibinfo{series}{LIPIcs'15})}.
\newblock
\urldef\tempurl%
\url{https://doi.org/10.4230/LIPIcs.FSTTCS.2015.475}
\showDOI{\tempurl}


\bibitem[\protect\citeauthoryear{Jelinek, Lafferty, and Mercer}{Jelinek
  et~al\mbox{.}}{1992}]%
        {SRU:JLM92}
\bibfield{author}{\bibinfo{person}{F. Jelinek}, \bibinfo{person}{J.~D.
  Lafferty}, {and} \bibinfo{person}{R.~L. Mercer}.}
  \bibinfo{year}{1992}\natexlab{}.
\newblock \showarticletitle{{Basic Methods of Probabilistic Context Free
  Grammars}}. In \bibinfo{booktitle}{\emph{Speech Recognition and
  Understanding}}.
\newblock
\urldef\tempurl%
\url{https://doi.org/10.1007/978-3-642-76626-8_35}
\showDOI{\tempurl}


\bibitem[\protect\citeauthoryear{Knuth}{Knuth}{1997}]%
        {book:KnuthArtVol2}
\bibfield{author}{\bibinfo{person}{Donald~E. Knuth}.}
  \bibinfo{year}{1997}\natexlab{}.
\newblock \bibinfo{booktitle}{\emph{{The Art of Computer Programming, Volume 2
  (3rd Ed.): Seminumerical Algorithms}}}.
\newblock \bibinfo{publisher}{Addison-Wesley}.
\newblock
\urldef\tempurl%
\url{https://dl.acm.org/doi/book/10.5555/270146}
\showURL{%
\tempurl}


\bibitem[\protect\citeauthoryear{Kozen}{Kozen}{1981}]%
        {JCSS:Kozen81}
\bibfield{author}{\bibinfo{person}{Dexter Kozen}.}
  \bibinfo{year}{1981}\natexlab{}.
\newblock \showarticletitle{{Semantics of Probabilistic Programs}}.
\newblock \bibinfo{journal}{\emph{J.\ Comput.\ Syst.\ Sci.}}
  \bibinfo{volume}{22}, \bibinfo{number}{3} (\bibinfo{date}{June}
  \bibinfo{year}{1981}).
\newblock
\urldef\tempurl%
\url{https://doi.org/10.1016/0022-0000(81)90036-2}
\showDOI{\tempurl}


\bibitem[\protect\citeauthoryear{Lee, Yu, Rival, and Yang}{Lee
  et~al\mbox{.}}{2019}]%
        {POPL:LYR19}
\bibfield{author}{\bibinfo{person}{Wonyeol Lee}, \bibinfo{person}{Hangyeol Yu},
  \bibinfo{person}{Xavier Rival}, {and} \bibinfo{person}{Hongseok Yang}.}
  \bibinfo{year}{2019}\natexlab{}.
\newblock \showarticletitle{{Towards Verified Stochastic Variational Inference
  for Probabilistic Programs}}.
\newblock \bibinfo{journal}{\emph{Proc.\ ACM Program.\ Lang.}}
  \bibinfo{volume}{4}, \bibinfo{number}{POPL} (\bibinfo{date}{December}
  \bibinfo{year}{2019}).
\newblock
\urldef\tempurl%
\url{https://doi.org/10.1145/3371084}
\showDOI{\tempurl}


\bibitem[\protect\citeauthoryear{Lew, Cusumano-Towner, Sherman, Carbin, and
  Mansinghka}{Lew et~al\mbox{.}}{2019}]%
        {POPL:LCS20}
\bibfield{author}{\bibinfo{person}{Alexander~K. Lew}, \bibinfo{person}{Marco~F.
  Cusumano-Towner}, \bibinfo{person}{Benjamin Sherman},
  \bibinfo{person}{Michael Carbin}, {and} \bibinfo{person}{Vikash~K.
  Mansinghka}.} \bibinfo{year}{2019}\natexlab{}.
\newblock \showarticletitle{{Trace Types and Denotational Semantics for Sound
  Programmable Inference in Probabilistic Languages}}.
\newblock \bibinfo{journal}{\emph{Proc.\ ACM Program.\ Lang.}}
  \bibinfo{volume}{4}, \bibinfo{number}{POPL} (\bibinfo{date}{December}
  \bibinfo{year}{2019}).
\newblock
\urldef\tempurl%
\url{https://doi.org/10.1145/3371087}
\showDOI{\tempurl}


\bibitem[\protect\citeauthoryear{Mansinghka, Schaechtle, Handa, Radul, Chen,
  and Rinard}{Mansinghka et~al\mbox{.}}{2018}]%
        {PLDI:MSH18}
\bibfield{author}{\bibinfo{person}{Vikash~K. Mansinghka},
  \bibinfo{person}{Ulrich Schaechtle}, \bibinfo{person}{Shivam Handa},
  \bibinfo{person}{Alexey Radul}, \bibinfo{person}{Yutian Chen}, {and}
  \bibinfo{person}{Martin~C. Rinard}.} \bibinfo{year}{2018}\natexlab{}.
\newblock \showarticletitle{{Probabilistic Programming with Programmable
  Inference}}. In \bibinfo{booktitle}{\emph{Prog.\ Lang.\ Design and Impl.}}
  \emph{(\bibinfo{series}{PLDI'18})}.
\newblock
\urldef\tempurl%
\url{https://doi.org/10.1145/3296979.3192409}
\showDOI{\tempurl}


\bibitem[\protect\citeauthoryear{Milner}{Milner}{1989}]%
        {book:Milner89}
\bibfield{author}{\bibinfo{person}{Robin Milner}.}
  \bibinfo{year}{1989}\natexlab{}.
\newblock \bibinfo{booktitle}{\emph{{Communication and Concurrency}}}.
\newblock \bibinfo{publisher}{Prentice-Hall, Inc.}
\newblock
\urldef\tempurl%
\url{https://dl.acm.org/doi/book/10.5555/534666}
\showURL{%
\tempurl}


\bibitem[\protect\citeauthoryear{Milner, Parrow, and Walker}{Milner
  et~al\mbox{.}}{1992a}]%
        {JIC:MPW92a}
\bibfield{author}{\bibinfo{person}{Robin Milner}, \bibinfo{person}{Joachim
  Parrow}, {and} \bibinfo{person}{David Walker}.}
  \bibinfo{year}{1992}\natexlab{a}.
\newblock \showarticletitle{{A Calculus of Mobile Processes, I}}.
\newblock \bibinfo{journal}{\emph{Information and Computation}}
  \bibinfo{volume}{100}, \bibinfo{number}{1} (\bibinfo{date}{September}
  \bibinfo{year}{1992}).
\newblock
\urldef\tempurl%
\url{https://doi.org/10.1016/0890-5401(92)90008-4}
\showDOI{\tempurl}


\bibitem[\protect\citeauthoryear{Milner, Parrow, and Walker}{Milner
  et~al\mbox{.}}{1992b}]%
        {JIC:MPW92b}
\bibfield{author}{\bibinfo{person}{Robin Milner}, \bibinfo{person}{Joachim
  Parrow}, {and} \bibinfo{person}{David Walker}.}
  \bibinfo{year}{1992}\natexlab{b}.
\newblock \showarticletitle{{A Calculus of Mobile Processes, II}}.
\newblock \bibinfo{journal}{\emph{Information and Computation}}
  \bibinfo{volume}{100}, \bibinfo{number}{1} (\bibinfo{date}{September}
  \bibinfo{year}{1992}).
\newblock
\urldef\tempurl%
\url{https://doi.org/10.1016/0890-5401(92)90009-5}
\showDOI{\tempurl}


\bibitem[\protect\citeauthoryear{Moggi}{Moggi}{1989}]%
        {LICS:Moggi89}
\bibfield{author}{\bibinfo{person}{Eugenio Moggi}.}
  \bibinfo{year}{1989}\natexlab{}.
\newblock \showarticletitle{{Computational lambda-calculus and monads}}. In
  \bibinfo{booktitle}{\emph{Logic in Computer Science}}
  \emph{(\bibinfo{series}{LICS'89})}.
\newblock
\urldef\tempurl%
\url{https://doi.org/10.1109/LICS.1989.39155}
\showDOI{\tempurl}


\bibitem[\protect\citeauthoryear{Murray}{Murray}{2015}]%
        {JSS:Murray15}
\bibfield{author}{\bibinfo{person}{Lawrence~M. Murray}.}
  \bibinfo{year}{2015}\natexlab{}.
\newblock \showarticletitle{{Bayesian State-Space Modelling on High-Performance
  Hardware Using LibBi}}.
\newblock \bibinfo{journal}{\emph{J.\ Statistical Softw.}}
  \bibinfo{volume}{67}, \bibinfo{number}{10} (\bibinfo{year}{2015}).
\newblock
\urldef\tempurl%
\url{https://doi.org/10.18637/jss.v067.i10}
\showDOI{\tempurl}


\bibitem[\protect\citeauthoryear{Narayanan, Carette, Romano, Shan, and
  Zinkov}{Narayanan et~al\mbox{.}}{2016}]%
        {FLOPS:NCR16}
\bibfield{author}{\bibinfo{person}{Praveen Narayanan}, \bibinfo{person}{Jacques
  Carette}, \bibinfo{person}{Wren Romano}, \bibinfo{person}{Chung-chieh Shan},
  {and} \bibinfo{person}{Robert Zinkov}.} \bibinfo{year}{2016}\natexlab{}.
\newblock \showarticletitle{{Probabilistic Inference by Program Transformation
  in Hakaru (System Description)}}. In \bibinfo{booktitle}{\emph{{Int.\ Symp.\
  on Functional and Logic Programming}}} \emph{(\bibinfo{series}{FLOPS'16})}.
\newblock
\urldef\tempurl%
\url{https://doi.org/10.1007/978-3-319-29604-3_5}
\showDOI{\tempurl}


\bibitem[\protect\citeauthoryear{Panangaden}{Panangaden}{1999}]%
        {ENTCS:Panangaden99}
\bibfield{author}{\bibinfo{person}{Prakash Panangaden}.}
  \bibinfo{year}{1999}\natexlab{}.
\newblock \showarticletitle{{The Category of Markov Kernels}}.
\newblock \bibinfo{journal}{\emph{Electr.\ Notes Theor.\ Comp.\ Sci.}}
  \bibinfo{volume}{22} (\bibinfo{year}{1999}).
\newblock
\urldef\tempurl%
\url{https://doi.org/10.1016/S1571-0661(05)80602-4}
\showDOI{\tempurl}


\bibitem[\protect\citeauthoryear{Plummer}{Plummer}{2003}]%
        {DSC:Plummer03}
\bibfield{author}{\bibinfo{person}{Martyn Plummer}.}
  \bibinfo{year}{2003}\natexlab{}.
\newblock \showarticletitle{{JAGS: A Program for Analysis of Bayesian Graphical
  Models using Gibbs Sampling}}. In \bibinfo{booktitle}{\emph{Int.\ Workshop on
  Distributed Statistical Comp.}} \emph{(\bibinfo{series}{DSC'03})}.
\newblock


\bibitem[\protect\citeauthoryear{Saad, Cusumano-Towner, Schaechtle, Rinard, and
  Mansinghka}{Saad et~al\mbox{.}}{2019}]%
        {POPL:SCS19}
\bibfield{author}{\bibinfo{person}{Feras~A. Saad}, \bibinfo{person}{Marco~F.
  Cusumano-Towner}, \bibinfo{person}{Ulrich Schaechtle},
  \bibinfo{person}{Martin~C. Rinard}, {and} \bibinfo{person}{Vikash~K.
  Mansinghka}.} \bibinfo{year}{2019}\natexlab{}.
\newblock \showarticletitle{{Bayesian Synthesis of Probabilistic Programs for
  Automatic Data Modeling}}.
\newblock \bibinfo{journal}{\emph{Proc.\ ACM Program.\ Lang.}}
  \bibinfo{volume}{3}, \bibinfo{number}{POPL} (\bibinfo{date}{January}
  \bibinfo{year}{2019}).
\newblock
\urldef\tempurl%
\url{https://doi.org/10.1145/3290350}
\showDOI{\tempurl}


\bibitem[\protect\citeauthoryear{Scalas and Yoshida}{Scalas and
  Yoshida}{2019}]%
        {POPL:SY19}
\bibfield{author}{\bibinfo{person}{Alceste Scalas} {and}
  \bibinfo{person}{Nobuko Yoshida}.} \bibinfo{year}{2019}\natexlab{}.
\newblock \showarticletitle{{Less Is More: Multiparty Session Types
  Revisited}}.
\newblock \bibinfo{journal}{\emph{Proc.\ ACM Program.\ Lang.}}
  \bibinfo{volume}{3}, \bibinfo{number}{POPL} (\bibinfo{date}{January}
  \bibinfo{year}{2019}).
\newblock
\urldef\tempurl%
\url{https://doi.org/10.1145/3290343}
\showDOI{\tempurl}


\bibitem[\protect\citeauthoryear{{\'S}cibior, Ghahramani, and
  Gordon}{{\'S}cibior et~al\mbox{.}}{2015}]%
        {HASKELL:SGG15}
\bibfield{author}{\bibinfo{person}{Adam {\'S}cibior}, \bibinfo{person}{Zoubin
  Ghahramani}, {and} \bibinfo{person}{Andrew~D. Gordon}.}
  \bibinfo{year}{2015}\natexlab{}.
\newblock \showarticletitle{{Practical Probabilistic Programming with Monads}}.
  In \bibinfo{booktitle}{\emph{Symp.\ on Haskell}}
  \emph{(\bibinfo{series}{Haskell'15})}.
\newblock
\urldef\tempurl%
\url{https://doi.org/10.1145/2887747.2804317}
\showDOI{\tempurl}


\bibitem[\protect\citeauthoryear{{\'S}cibior, Kammar, V{\'a}k{\'a}r, Staton,
  Yang, Cai, Ostermann, Moss, Heunen, and Ghahramani}{{\'S}cibior
  et~al\mbox{.}}{2017}]%
        {POPL:SKV17}
\bibfield{author}{\bibinfo{person}{Adam {\'S}cibior}, \bibinfo{person}{Ohad
  Kammar}, \bibinfo{person}{Matthijs V{\'a}k{\'a}r}, \bibinfo{person}{Sam
  Staton}, \bibinfo{person}{Hongseok Yang}, \bibinfo{person}{Yufei Cai},
  \bibinfo{person}{Klaus Ostermann}, \bibinfo{person}{Sean~K. Moss},
  \bibinfo{person}{Chris Heunen}, {and} \bibinfo{person}{Zoubin Ghahramani}.}
  \bibinfo{year}{2017}\natexlab{}.
\newblock \showarticletitle{{Denotational Validation of Higher-Order Bayesian
  Inference}}.
\newblock \bibinfo{journal}{\emph{Proc.\ ACM Program.\ Lang.}}
  \bibinfo{volume}{2}, \bibinfo{number}{POPL} (\bibinfo{date}{December}
  \bibinfo{year}{2017}).
\newblock
\urldef\tempurl%
\url{https://doi.org/10.1145/3158148}
\showDOI{\tempurl}


\bibitem[\protect\citeauthoryear{Shan and Ramsey}{Shan and Ramsey}{2017}]%
        {POPL:SR17}
\bibfield{author}{\bibinfo{person}{Chung-chieh Shan} {and}
  \bibinfo{person}{Norman Ramsey}.} \bibinfo{year}{2017}\natexlab{}.
\newblock \showarticletitle{{Exact Bayesian Inference by Symbolic
  Disintegration}}. In \bibinfo{booktitle}{\emph{Princ.\ of Prog.\ Lang.}}
  \emph{(\bibinfo{series}{POPL'17})}.
\newblock
\urldef\tempurl%
\url{https://doi.org/10.1145/3009837.3009852}
\showDOI{\tempurl}


\bibitem[\protect\citeauthoryear{Thiemann and Vasconcelos}{Thiemann and
  Vasconcelos}{2016}]%
        {ICFP:TV16}
\bibfield{author}{\bibinfo{person}{Peter Thiemann} {and}
  \bibinfo{person}{Vasco~T. Vasconcelos}.} \bibinfo{year}{2016}\natexlab{}.
\newblock \showarticletitle{{Context-Free Session Types}}. In
  \bibinfo{booktitle}{\emph{Int.\ Conf.\ on Functional Programming}}
  \emph{(\bibinfo{series}{ICFP'16})}.
\newblock
\urldef\tempurl%
\url{https://doi.org/10.1145/2951913.2951926}
\showDOI{\tempurl}


\bibitem[\protect\citeauthoryear{Thiemann and Vasconcelos}{Thiemann and
  Vasconcelos}{2019}]%
        {POPL:TV19}
\bibfield{author}{\bibinfo{person}{Peter Thiemann} {and}
  \bibinfo{person}{Vasco~T. Vasconcelos}.} \bibinfo{year}{2019}\natexlab{}.
\newblock \showarticletitle{{Label-Dependent Session Types}}.
\newblock \bibinfo{journal}{\emph{Proc.\ ACM Program.\ Lang.}}
  \bibinfo{volume}{4}, \bibinfo{number}{POPL} (\bibinfo{date}{December}
  \bibinfo{year}{2019}).
\newblock
\urldef\tempurl%
\url{https://doi.org/10.1145/3371135}
\showDOI{\tempurl}


\bibitem[\protect\citeauthoryear{Tran, Hoffman, Saurous, Brevdo, Murphy, and
  Blei}{Tran et~al\mbox{.}}{2017}]%
        {ICLR:THS17}
\bibfield{author}{\bibinfo{person}{Dustin Tran}, \bibinfo{person}{Matthew~D.
  Hoffman}, \bibinfo{person}{Rif~A. Saurous}, \bibinfo{person}{Eugene Brevdo},
  \bibinfo{person}{Kevin Murphy}, {and} \bibinfo{person}{David~M. Blei}.}
  \bibinfo{year}{2017}\natexlab{}.
\newblock \showarticletitle{{Deep Probabilistic Programming}}. In
  \bibinfo{booktitle}{\emph{Int.\ Conf.\ on Learning Representations}}
  \emph{(\bibinfo{series}{ICLR'17})}.
\newblock


\bibitem[\protect\citeauthoryear{Wadler}{Wadler}{2012}]%
        {ICFP:Wadler12}
\bibfield{author}{\bibinfo{person}{Philip Wadler}.}
  \bibinfo{year}{2012}\natexlab{}.
\newblock \showarticletitle{{Propositions as Sessions}}. In
  \bibinfo{booktitle}{\emph{Int.\ Conf.\ on Functional Programming}}
  \emph{(\bibinfo{series}{ICFP'12})}.
\newblock
\urldef\tempurl%
\url{https://doi.org/10.1145/2364527.2364568}
\showDOI{\tempurl}


\bibitem[\protect\citeauthoryear{Webb, Golinski, Zinkov, Siddharth, Rainforth,
  Teh, and Wood}{Webb et~al\mbox{.}}{2018}]%
        {NIPS:WGZ18}
\bibfield{author}{\bibinfo{person}{Stefan Webb}, \bibinfo{person}{Adam
  Golinski}, \bibinfo{person}{Robert Zinkov}, \bibinfo{person}{N. Siddharth},
  \bibinfo{person}{Tom Rainforth}, \bibinfo{person}{Yee~Whye Teh}, {and}
  \bibinfo{person}{Frank Wood}.} \bibinfo{year}{2018}\natexlab{}.
\newblock \showarticletitle{{Faithful Inversion of Generative Models for
  Effective Amortized Inference}}. In \bibinfo{booktitle}{\emph{Neural Info.\
  Processing Syst.}} \emph{(\bibinfo{series}{NIPS'18})}.
\newblock
\urldef\tempurl%
\url{https://dl.acm.org/doi/10.5555/3327144.3327229}
\showURL{%
\tempurl}


\bibitem[\protect\citeauthoryear{Website}{Website}{2020}]%
        {misc:greenlet}
\bibfield{author}{\bibinfo{person}{Website}.} \bibinfo{year}{2020}\natexlab{}.
\newblock \bibinfo{title}{{greenlet: Lightweight concurrent programming}}.
\newblock \bibinfo{howpublished}{Available on
  \url{https://greenlet.readthedocs.io}}.
\newblock


\bibitem[\protect\citeauthoryear{Williams}{Williams}{1991}]%
        {book:Williams91}
\bibfield{author}{\bibinfo{person}{David Williams}.}
  \bibinfo{year}{1991}\natexlab{}.
\newblock \bibinfo{booktitle}{\emph{{Probability with Martingales}}}.
\newblock \bibinfo{publisher}{Cambridge University Press}.
\newblock
\urldef\tempurl%
\url{https://doi.org/10.1017/CBO9780511813658}
\showDOI{\tempurl}


\bibitem[\protect\citeauthoryear{Wood, van~de Meent, and Mansinghka}{Wood
  et~al\mbox{.}}{2014}]%
        {AISTATS:WMM14}
\bibfield{author}{\bibinfo{person}{Frank Wood}, \bibinfo{person}{Jan~Willem
  van~de Meent}, {and} \bibinfo{person}{Vikash~K. Mansinghka}.}
  \bibinfo{year}{2014}\natexlab{}.
\newblock \showarticletitle{{A New Approach to Probabilistic Programming
  Inference}}. In \bibinfo{booktitle}{\emph{Artificial Intelligence and
  Statistics}} \emph{(\bibinfo{series}{AISTATS'14})}.
\newblock


\bibitem[\protect\citeauthoryear{Zinkov and Shan}{Zinkov and Shan}{2017}]%
        {UAI:ZS17}
\bibfield{author}{\bibinfo{person}{Robert Zinkov} {and}
  \bibinfo{person}{Chung-chieh Shan}.} \bibinfo{year}{2017}\natexlab{}.
\newblock \showarticletitle{{Composing Inference Algorithms as Program
  Transformations}}. In \bibinfo{booktitle}{\emph{Uncertainty in Artificial
  Intelligence}} \emph{(\bibinfo{series}{UAI'17})}.
\newblock
\urldef\tempurl%
\url{https://arxiv.org/abs/1603.01882}
\showURL{%
\tempurl}


\end{thebibliography}

\iflong
\clearpage
\appendix
\onecolumn
\setlist{nosep}

\section{Preliminaries on Measure Theory}
\label{Se:MeasureTheory}

Interested readers can refer to textbooks and notes in the
literature~\cite{book:Billingsley12,book:Williams91} for more details.

A \emph{measurable space} is a pair $(S,\calS)$, where $S$ is a nonempty set, and $\calS$ is a \emph{$\sigma$-algebra} on $S$, i.e., a family of subsets of $S$ that contains $\emptyset$ and is closed under complement and countable unions.
The smallest $\sigma$-algebra that contains a family $\calA$ of subsets of $S$ is said to be \emph{generated} by $\calA$, denoted by $\sigma(\calA)$.
Every topological space $(S,\tau)$ admits a \emph{Borel $\sigma$-algebra}, given by $\sigma(\tau)$.
This gives canonical $\sigma$-algebras on $\bbR$, $\bbQ$, $\bbN$, etc.
A measurable space $(S,\calS)$ is said to be a \emph{standard Borel space}, if $\calS$ is a Borel
$\sigma$-algebra generated by a complete metric space on $S$.
A measurable space $(S,\calS)$ is a standard Borel space if and only if it is isomorphic to $\bbR$ or
a subset of $\bbN$.
A function $f : S \to T$, where $(S,\calS)$ and $(T,\calT)$ are measurable spaces, is said to be \emph{$(\calS,\calT)$-measurable}, if $f^{-1}(B) \in \calS$ for each $B \in \calT$.
If $T = \bbR$, we tacitly assume that the Borel $\sigma$-algebra is defined on $T$, and we simply call $f$ \emph{measurable}, or a \emph{random variable}.
Measurable functions form a vector space, and products, maxima, and limiting operations preserve measurability.

A \emph{measure} $\mu$ on a measurable space $(S,\calS)$ is a mapping from $\calS$ to $[0,\infty]$ such that
(i) $\mu(\emptyset) = 0$, and
(ii) for all pairwise-disjoint $\{A_n\}_{n \in \bbZ^+}$ in $\calS$, it holds that $\mu(\bigcup_{n \in \bbZ^+} A_i) = \sum_{n \in \bbZ^+} \mu(A_i)$.
The triple $(S,\calS,\mu)$ is called a \emph{measure space}.
A measure $\mu$ is called a \emph{probability} measure, if $\mu(S) = 1$.
A measure $\mu$ is called a \emph{sub-probability} measure, if $\mu(S) \le 1$.
A measure $\mu$ is called \emph{$\sigma$-finite}, if $S$ is the countable union of
measurable sets with finite measure.
We denote the collection of probability measures on $(S,\calS)$ by $\bbD(S,\calS)$.
For each $x \in S$, the \emph{Dirac measure} $\delta(x)$ is defined as $\lambda A. [x \in A]$.
For measures $\mu$ and $\nu$, we write $\mu + \nu$ for the measure $\lambda A. \mu(A) + \nu(A)$.
For measure $\mu$ and scalar $c \ge 0$, we write $c \cdot \mu$ for the measure $\lambda A. c \cdot \mu(A)$.

The \emph{integral} of a measurable function $f$ on $A \in \calS$ with respect to a measure $\mu$ on $(S,\calS)$ is defined following Lebesgue's theory and is denoted by $\mu(f;A)$, $\int_A f d\mu$, or $\int_A f(x) \mu(dx)$.
If $A = S$, we tacitly omit $A$ from the notations.
For each $A \in \calS$, it holds that $\mu(f; A) = \mu(f \mathrm{I}_A)$, where $\mathrm{I}_A$ is the indicator function for $A$. 
%

Let $f$ be a nonnegative measurable function on $(S,\calS)$.
We can transform a measure $\mu$ on $(S,\calS)$ through $f$ by integration:
$f\mu \defeq \lambda A. \mu(f; A)$.
If $\nu$ denotes the measure $f\mu$, we say that $\nu$ has \emph{density} $f$ relative to $\mu$,
and express this by $\frac{d \nu}{d \mu} = f$.
In this case, we have for $A \in \calS$, $\mu(A) = 0$ implies that $\nu(A) = 0$,
i.e., $\nu$ is \emph{absolutely continuous} with respect to $\mu$.

A \emph{kernel} from a measurable space $(S,\calS)$ to another $(T,\calT)$ is a mapping from $S \times \calT$ to $[0,\infty]$ such that:
(i) for each $x \in S$, the function $\lambda B. \kappa(x,B)$ is a measure on $(T,\calT)$, and
(ii) for each $B \in \calT$, the function $\lambda x. \kappa(x,B)$ is measurable.
We write $\kappa : (S,\calS) \rightsquigarrow (T,\calT)$ to declare that $\kappa$ is a kernel from $(S,\calS)$ to $(T,\calT)$.
Intuitively, kernels describe measure transformers from one measurable space to another.
A kernel $\kappa$ is called a \emph{probability} kernel, if $\kappa(x,T) = 1$ for all $x \in S$.
We denote the collection of probability kernels from $(S,\calS)$ to $(T,\calT)$ by $\bbK((S,\calS),(T,\calT))$.
If the two measurable spaces coincide, we simply write $\bbK(S,\calS)$.
We can ``push-forward'' a measure $\mu$ on $(S,\calS)$ to a measure on $(T,\calT)$ through a kernel $\kappa: (S,\calS) \rightsquigarrow (T,\calT)$ by integration:\footnote{
We use a \emph{monad bind} notation $\bind$ here.
Indeed, the category of measurable spaces admits a monad with sub-probability measures~\cite{CATA:Giry82,ENTCS:Panangaden99}.
}
$
\mu \bind \kappa \defeq  \lambda B. \int_S  \kappa(x,B) \mu(dx).
$

The \emph{product} of two measurable spaces $(S,\calS)$ and $(T,\calT)$ is defined as $(S,\calS) \otimes (T,\calT) \defeq (S \times T, \calS \otimes \calT)$, where $\calS \otimes \calT$ is the smallest $\sigma$-algebra that makes coordinate maps measurable, i.e., $\sigma( \{ \pi_1^{-1}(A) \mid A \in \calS\} \cup \{ \pi_2^{-1}(B) \mid B \in \calT \}  ) )$, where $\pi_i$ is the $i$-the coordinate map.
If $\mu_1$ and $\mu_2$ are two measures on $(S,\calS)$ and $(T,\calT)$, respectively,
then there exists a measure on $(S,\calS) \otimes (T,\calT)$, called the \emph{product measure} and written
$\mu_1 \otimes \mu_2$, such that $(\mu_1 \otimes \mu_2)(A \times B) = \mu_1(A)\mu_2(B)$.
When $\mu_1$ and $\mu_2$ are $\sigma$-finite, the product measure is uniquely defined and also $\sigma$-finite.

The \emph{coproduct} (i.e., disjoint union) of two measurable space $(S,\calS)$ and $(T,\calT)$ is defined as
$(S,\calS) \amalg (T,\calT) \defeq (S \sqcup T, \calS \amalg \calT)$, where
$S \sqcup T \defeq \{ i_1(x) \mid x \in S \} \cup \{ i_1(y) \mid y \in T \}$, $i_1 \defeq \lambda x. \tuple{1,x}$,
$i_2 \defeq \lambda y. \tuple{2,y}$, and
$\calS \amalg \calT$ is the smallest $\sigma$-algebra that makes injection maps measurable, i.e.,
$\sigma( \{ i_1^{-1}(A) \mid A \in \calS \} \cup \{ i_2^{-1}(B) \mid B \in \calT \}) = \{ i_1^{-1}(A) \cup i_2^{-1}(B) \mid A \in \calS \wedge B \in \calT \}$.
If $\mu_1$ and $\mu_2$ are two measures on $(S,\calS)$ and $(T,\calT)$, respectively,
then we can define their \emph{coproduct measure}, written $\mu_1 \amalg \mu_2$,
as $(\mu_1 \amalg \mu_2)(i_1^{-1}(A) \cup i_2^{-1}(B)) \defeq \mu_1(A) + \mu_2(B)$, for any $A \in \calS$,
$B \in \calT$.
We can easily extend the binary coproducts to arbitary coproducts.
Particularly, the countable coproduct of $\sigma$-finite measures is still $\sigma$-finite.

Standard Borel spaces are closed under countable products and coproducts.
We will use this property in our construction of semantic domains in \cref{Se:FullInference}.


\section{Full Development of Guide Types}
\label{Se:FullSpec}

\cref{Fig:FullEvalElab} presents a complete list of evaluation rules for expressions and commands.
\cref{Fig:FullTypingElab} presents a complete list of typing rules for expressions, commands, and programs.
\cref{Fig:TypingAux} presents typing rules for values, environments, and guidance traces.
In the rest of this section, we prove type safety of guide types.

\begin{figure*}
\centering 
\begin{mathpar}\small
  \Rule{EE:Var}
  { v =  V(x)
  }
  { V \vdash x \evalto v }
  \and
  \Rule{EE:Triv}
  {
  }
  { V \vdash \etriv \evalto \etriv }
  \and
  \Rule{EE:True}
  { }
  { V \vdash \etrue \evalto \etrue }
  \and
  \Rule{EE:False}
  {
  }
  { V \vdash \efalse \evalto \efalse }
  \and
  \Rule{EE:Cond:True}
  { V \vdash e \evalto \etrue \\
    V \vdash e_1 \evalto v
  }
  { V \vdash \econd{e}{e_1}{e_2} \evalto v }
  \and
  \Rule{EE:Cond:False}
  { V \vdash e \evalto \efalse \\
    V \vdash e_2 \evalto v
  }
  { V \vdash \econd{e}{e_1}{e_2} \evalto v }
  \and
  \Rule{EE:Real}
  {
  }
  { V \vdash \bar{r} \evalto \bar{r} }
  \and
  \Rule{EE:Nat}
  {
  }
  { V \vdash \bar{n} \evalto \bar{n} }
  \and
  \Rule{EE:Op}
  { V \vdash e_1 \evalto v_1 \\
    V \vdash e_2 \evalto v_2 \\
    v = v_1 \mathbin{\Diamond} v_2
  }
  { V \vdash \ebinop{\Diamond}{e_1}{e_2} \evalto v }
  \and
  \Rule{EE:Abs}
  {
  }
  { V \vdash \eabs{x}{\tau}{e} \evalto \vclo{V}{\eabs{x}{\tau}{e}} }
  \and
  \Rule{EE:App}
  { V \vdash e_1 \evalto \vclo{V_o}{\eabs{x}{\tau}{e_o}} \\
    V \vdash e_2 \evalto v_2 \\
    V_o[x \mapsto v_2] \vdash e_o \evalto v
  }
  { V \vdash \eapp{e_1}{e_2} \evalto v }
  \and
  \Rule{EE:Let}
  { V \vdash e_1 \evalto v_1 \\
    V[x \mapsto v_1] \vdash e_2 \evalto v_2
  }
  { V \vdash \elet{e_1}{x}{e_2} \evalto v_2 }
  \and
  \Rule{EE:Ber}
  { V \vdash e \evalto v
  }
  { V \vdash \eber{e} \evalto \eber{v} }
  \and
  \Rule{EE:Unif}
  {
  }
  { V \vdash \eunif \evalto \eunif }
  \and
  \Rule{EE:Beta}
  { V \vdash e_1 \evalto v_1 \\
    V \vdash e_2 \evalto v_2
  }
  { V \vdash \ebeta{e_1}{e_2} \evalto \ebeta{v_1}{v_2} }
  \and
  \Rule{EE:Gamma}
  { V \vdash e_1 \evalto v_1 \\
    V \vdash e_2 \evalto v_2
  }
  { V \vdash \egamma{e_1}{e_2} \evalto \egamma{v_1}{v_2} }
  \and
  \Rule{EE:Normal}
  { V \vdash e_1 \evalto v_1 \\
    V \vdash e_2 \evalto v_2
  }
  { V \vdash \enormal{e_1}{e_2} \evalto \enormal{v_1}{v_2} }
  \and
  \Rule{EE:Cat}
  { \Forall{i \in \{1,\cdots, n\}} V \vdash e_i \evalto v_i }
  { V \vdash \ecat{e_1, \cdots, e_n} \evalto \ecat{v_1,\cdots,v_n} }
  \and
  \Rule{EE:Geo}
  { V \vdash e \evalto v
  }
  { V \vdash \egeo{e} \evalto \egeo{v} }
  \and
  \Rule{EE:Pois}
  { V \vdash e \evalto v
  }
  { V \vdash \epois{e} \evalto \epois{v} }
  \\
  \Rule{EM:Ret}
  { V \vdash e \evalto v
  }
  { V \mid \chtype{a}{[]} ; \chtype{b}{[]} \vdash \mret{e} \evalp{1} v }
  \and
  \Rule{EM:Bnd}
  { V \mid \chtype{a}{\sigma_a} ; \chtype{b}{\sigma_b} \vdash m_1 \evalp{w_1} v_1 \\
    V[x \mapsto v_1] \mid \chtype{a}{\sigma_a'} ; \chtype{b}{\sigma_b'} \vdash m_2 \evalp{w_2} v_2
  }
  { V \mid \chtype{a}{\sigma_a \concat \sigma_a'} ; \chtype{b}{\sigma_b \concat \sigma_b'} \vdash \mbnd{m_1}{x}{m_2} \evalp{w_1 \cdot w_2} v_2 }
  \and
  \Rule{EM:Sample:Recv:L}
  { V \vdash e \evalto d \\
    v \in d.\mathrm{support} \\
    w = d.\mathrm{density}(v)
  }
  { V \mid \chtype{a}{[\msgobjl{v}]} ; \chtype{b}{[]} \vdash \msamplei{e}{a} \evalp{w} v }
  \and
  \Rule{EM:Sample:Send:R}
  { V \vdash e \evalto d \\
    v \in d.\mathrm{support} \\
    w = d.\mathrm{density}(v)
  }
  { V \mid \chtype{a}{[]} ; \chtype{b}{[\msgobjl{v}]} \vdash \msampleo{e}{b} \evalp{w} v }
  \and
  \Rule{EM:Sample:Send:L}
  { V \vdash e \evalto d \\
    v \in d.\mathrm{support} \\
    w = d.\mathrm{density}(v)
  }
  { V \mid \chtype{a}{[\msgobjr{v}]} ; \chtype{b}{[]} \vdash \msampleo{e}{a} \evalp{w} v }
  \and
  \Rule{EM:Sample:Recv:R}
  { V \vdash e \evalto d \\
    v \in d.\mathrm{support} \\
    w = d.\mathrm{density}(v)
  }
  { V \mid \chtype{a}{[]} ; \chtype{b}{[\msgobjr{v}]} \vdash \msamplei{e}{b} \evalp{w} v }
  \and
  \Rule{EM:Cond:Recv:L}
  { i = \m{ite}(v_a, 1, 2) \\
    V \mid \chtype{a}{\sigma_a} ; \chtype{b}{\sigma_b} \vdash m_i \evalp{w} v
  }
  { V \mid \chtype{a}{[\dirobjl{v_a}] \concat \sigma_a} ; \chtype{b}{\sigma_b} \vdash \mbranchi{m_1}{m_2}{a} \evalp{w} v }
  \and
  \Rule{EM:Cond:Send:R}
  { V \vdash e \evalto v_e \\
    i = \m{ite}(v_b, 1, 2) \\
    V \mid \chtype{a}{\sigma_a} ; \chtype{b}{\sigma_b} \vdash m_i \evalp{w} v
  }
  { V \mid \chtype{a}{\sigma_a} ; \chtype{b}{[\dirobjl{v_b}] \concat \sigma_b} \vdash \mbrancho{e}{m_1}{m_2}{b} \evalp{w \cdot [v_b = v_e]} v }
  \and
  \Rule{EM:Cond:Send:L}
  { V \vdash e \evalto v_e \\
    i = \m{ite}(v_a, 1, 2) \\
    V \mid \chtype{a}{\sigma_a} ; \chtype{b}{\sigma_b} \vdash m_i \evalp{w} v
  }
  { V \mid \chtype{a}{[\dirobjr{v_a}] \concat \sigma_a} ; \chtype{b}{\sigma_b} \vdash \mbrancho{e}{m_1}{m_2}{a} \evalp{w \cdot [v_a = v_e]} v }
  \and
  \Rule{EM:Cond:Recv:R}
  { i = \m{ite}(v_b, 1, 2) \\
    V \mid \chtype{a}{\sigma_a} ; \chtype{b}{\sigma_b} \vdash m_i \evalp{w} v
  }
  { V \mid \chtype{a}{\sigma_a} ; \chtype{b}{[\dirobjr{v_b}] \concat \sigma_b} \vdash \mbranchi{m_1}{m_2}{b} \evalp{w} v }
  \and
  \Rule{EM:Call}
  { \calD(f) = \fundec{f}{\tau_1}{\tau_2}{x_f}{m_f}{a}{b} \\
    V \vdash e \evalto v_1 \\
    \emptyset[x_f \mapsto v_1] \mid \chtype{a}{\sigma_a} ; \chtype{b}{\sigma_b} \vdash m_f \evalp{w} v_2
  }
  { V \mid \chtype{a}{[\foldobj] \concat \sigma_a} ; \chtype{b}{[\foldobj] \concat \sigma_b} \vdash \mcall{f}{e} \evalp{w} v_2 }  
\end{mathpar}
\caption{Evaluation rules for expressions and commands.}
\label{Fig:FullEvalElab}
\end{figure*}

\begin{figure*}
\centering
\begin{mathpar}\small
  \Rule{TE:Var}
  {
  }
  { \Gm, x : \tau \vdash x : \tau }
  \and
  \Rule{TE:Triv}
  {
  }
  { \Gm \vdash \etriv : \tunit }
  \and
  \Rule{TE:True}
  { 
  }
  { \Gm \vdash \etrue : \tbool }
  \and
  \Rule{TE:False}
  {
  }
  { \Gm \vdash \efalse : \tbool }
  \and
  \Rule{TE:Cond}
  { \Gm \vdash e : \tbool \\
    \Gm \vdash e_1 : \tau \\
    \Gm \vdash e_2 : \tau
  }
  { \Gm \vdash \econd{e}{e_1}{e_2} : \tau }
  \and
  \Rule{TE:UReal}
  { r \in (0, 1)
  }
  { \Gm \vdash \bar{r} : \tureal }
  \and
  \Rule{TE:PReal}
  { r > 0
  }
  { \Gm \vdash \bar{r} : \tpreal }
  \and
  \Rule{TE:Real}
  {
  }
  { \Gm \vdash \bar{r} : \treal }
  \and
  \Rule{TE:FNat}
  { n < m
  }
  { \Gm \vdash \bar{n} : \tnat_m }
  \and
  \Rule{TE:Nat}
  {
  }
  { \Gm \vdash \bar{n} : \tnat }
  \and
  \Rule{TE:Op}
  { \Gm \vdash e_1 : \Diamond.\mathrm{arg}_1 \\
    \Gm \vdash e_2 : \Diamond.\mathrm{arg}_2
  }
  { \Gm \vdash \ebinop{\Diamond}{e_1}{e_2} : \Diamond.\mathrm{res} }
  \and
  \Rule{TE:Abs}
  { \Gm, x : \tau \vdash e : \tau'
  }
  { \Gm \vdash \eabs{x}{\tau}{e} : \tau \to \tau' }
  \and
  \Rule{TE:App}
  { \Gm \vdash e_1 : \tau_1 \to \tau_2 \\
    \Gm \vdash e_2 : \tau_1
  }
  { \Gm \vdash \eapp{e_1}{e_2} : \tau_2 }
  \and
  \Rule{TE:Let}
  { \Gm \vdash e_1  :\tau_1 \\
    \Gm,x:\tau_1 \vdash e_2 : \tau_2
  }
  { \Gm \vdash \elet{e_1}{x}{e_2}  : \tau_2 }
  \and
  \Rule{TE:Ber}
  { \Gm \vdash e  : \tureal
  }
  { \Gm \vdash \eber{e} : \tdist{\tbool} }
  \and
  \Rule{TE:Unif}
  {
  }
  { \Gm \vdash \eunif : \tdist{\tureal} }
  \and
  \Rule{TE:Beta}
  { \Gm \vdash e_1 : \tpreal \\
    \Gm \vdash e_2 : \tpreal
  }
  { \Gm \vdash \ebeta{e_1}{e_2} : \tdist{\tureal} }
  \and
  \Rule{TE:Gamma}
  { \Gm \vdash e_1 : \tpreal \\
    \Gm \vdash e_2 : \tpreal
  }
  { \Gm \vdash \egamma{e_1}{e_2} : \tdist{\tpreal} }
  \and
  \Rule{TE:Normal}
  { \Gm \vdash e_1 : \treal \\
    \Gm \vdash e_2 : \tpreal
  }
  { \Gm \vdash \enormal{e_1}{e_2} : \tdist{\treal} }
  \and
  \Rule{TE:Cat}
  { \Forall{i \in \{1, \cdots, n\}} \Gm \vdash e_i : \tpreal
  }
  { \Gm \vdash \ecat{e_1,\cdots,e_n} : \tdist{\tnat_n} }
  \and
  \Rule{TE:Geo}
  { \Gm \vdash e : \tureal
  }
  { \Gm \vdash \egeo{e} : \tdist{\tnat} }
  \and
  \Rule{TE:Pois}
  { \Gm \vdash e : \tpreal
  }
  { \Gm \vdash \epois{e} : \tdist{\tnat} }
  \\
  \Rule{TM:Ret}
  { \Gm \vdash e : \tau
  }
  { \Gm \mid \chtype{a}{A} ; \chtype{b}{B} \vdash \mret{e} \dotsim \tau \mid \chtype{a}{A} ; \chtype{b}{B}  }
  \and
  \Rule{TM:Bnd}
  { \Gm \mid \chtype{a}{A} ; \chtype{b}{B} \vdash m_1 \dotsim \tau_1 \mid \chtype{a}{A'} ; \chtype{b}{B'} \\\\
    \Gm, x : \tau_1 \mid \chtype{a}{A'} ; \chtype{b}{B'} \vdash m_2 \dotsim \tau_2 \mid \chtype{a}{A''} ; \chtype{b}{B''}
  }
  { \Gm \mid \chtype{a}{A} ; \chtype{b}{B} \vdash \mbnd{m_1}{x}{m_2} \dotsim \tau_2 \mid \chtype{a}{A''} ; \chtype{b}{B''} }
  \and
  \Rule{TM:Sample:Recv:L}
  { \Gm \vdash e : \tdist{\tau}
  }
  { \Gm \mid \chtype{a}{\tau \wedge A} ; \chtype{b}{B} \vdash \msamplei{e}{a} \dotsim \tau \mid \chtype{a}{A} ; \chtype{b}{B} }
  \and
  \Rule{TM:Sample:Send:R}
  { \Gm \vdash e : \tdist{\tau}
  }
  { \Gm \mid \chtype{a}{A} ; \chtype{b}{ \tau \wedge B} \vdash \msampleo{e}{b} \dotsim \tau \mid \chtype{a}{A} ; \chtype{b}{B} }
  \and
  \Rule{TM:Sample:Send:L}
  { \Gm \vdash e : \tdist{\tau}
  }
  { \Gm \mid \chtype{a}{\tau \supset A} ; \chtype{b}{B} \vdash \msampleo{e}{a} \dotsim \tau \mid \chtype{a}{A} ; \chtype{b}{B} }
  \and
  \Rule{TM:Sample:Recv:R}
  { \Gm \vdash e : \tdist{\tau}
  }
  { \Gm \mid \chtype{a}{A} ; \chtype{b}{\tau \supset B} \vdash \msamplei{e}{b} \dotsim \tau \mid \chtype{a}{A} ; \chtype{b}{B} }
  \and
  \Rule{TM:Cond:Recv:L}
  { \Gm \mid \chtype{a}{A_1} ; \chtype{b}{B} \vdash m_1 \dotsim \tau \mid \chtype{a}{A'} ; \chtype{b}{B'} \\\\
    \Gm \mid \chtype{a}{A_2} ; \chtype{b}{B} \vdash m_2 \dotsim \tau \mid \chtype{a}{A'} ; \chtype{b}{B'}
  }
  { \Gm \mid \chtype{a}{A_1 \ichoice A_2} ; \chtype{b}{B} \vdash \mbranchi{m_1}{m_2}{a} \dotsim \tau \mid \chtype{a}{A'} ; \chtype{b}{B'} }
  \enskip
  \Rule{TM:Cond:Send:R}
  { \Gm \vdash e : \tbool \\
    \Gm \mid \chtype{a}{A} ; \chtype{b}{B_1} \vdash m_1 \dotsim \tau \mid \chtype{a}{A'} ; \chtype{b}{B'} \\\\
    \Gm \mid \chtype{a}{A} ; \chtype{b}{B_2} \vdash m_2 \dotsim \tau \mid \chtype{a}{A'} ; \chtype{b}{B'}
  }
  { \Gm \mid \chtype{a}{A} ; \chtype{b}{B_1 \ichoice B_2} \vdash \mbrancho{e}{m_1}{m_2}{b} \dotsim \tau \mid \chtype{a}{A'} ; \chtype{b}{B'} }
  \and
  \Rule{TM:Cond:Send:L}
  { \Gm \vdash e : \tbool \\
    \Gm \mid \chtype{a}{A_1} ; \chtype{b}{B} \vdash m_1 \dotsim \tau \mid \chtype{a}{A'} ; \chtype{b}{B'} \\\\
    \Gm \mid \chtype{a}{A_2} ; \chtype{b}{B} \vdash m_2 \dotsim \tau \mid \chtype{a}{A'} ; \chtype{b}{B'}
  }
  { \Gm \mid \chtype{a}{A_1 \echoice A_2} ; \chtype{b}{B} \vdash \mbrancho{e}{m_1}{m_2}{a} \dotsim \tau \mid \chtype{a}{A'} ; \chtype{b}{B'} }
  \enskip
  \Rule{TM:Cond:Recv:R}
  { \Gm \mid \chtype{a}{A} ; \chtype{b}{B_1} \vdash m_1 \dotsim \tau \mid \chtype{a}{A'} ; \chtype{b}{B'} \\\\
    \Gm \mid \chtype{a}{A} ; \chtype{b}{B_2} \vdash m_2 \dotsim \tau \mid \chtype{a}{A'} ; \chtype{b}{B'}
  }
  { \Gm \mid \chtype{a}{A} ; \chtype{b}{B_1 \echoice B_2} \vdash \mbranchi{m_1}{m_2}{b} \dotsim \tau \mid \chtype{a}{A'} ; \chtype{b}{B'} }
  \and
  \Rule{TM:Call}
  { \Sg(f) = \tau_1 \leadsto \tau_2 \mid \chtype{a}{T_a} ; \chtype{b}{T_b} \\
    \Gm \vdash e : \tau_1
  }
  { \Gm \mid \chtype{a}{T_a[A]} ; \chtype{b}{T_b[B]} \vdash \mcall{f}{e} \dotsim \tau_2 \mid \chtype{a}{A} ; \chtype{b}{B} }
  \\
  \Rule{TP:Dec}
  { \m{typedef}(T_a.X_a.A), \m{typedef}(T_b.X_b.B) \in \calT \\\\
    x_f : \tau_1 \mid \chtype{a}{A} ; \chtype{b}{B} \vdash_\Sg m_f \dotsim \tau_2 \mid \chtype{a}{X_a} ; \chtype{b}{X_b}
  }
  { \vdash_\Sg \fundec{f}{\tau_1}{\tau_2}{x_f}{m_f}{a}{b} : \tau_1 \leadsto \tau_2 \mid \chtype{a}{T_a} ; \chtype{b}{T_b} }
  \and
  \Rule{TP:Proc}
  { \calD = \many{\fundec{f_i}{\tau_i}{\rho_i}{x_i}{m_i}{a_i}{b_i}} \\
    \Forall{i} 
    {} \vdash_{\Sg} \fundec{f_i}{\tau_i}{\rho_i}{x_i}{m_i}{a_i}{b_i} : \Sg(f_i)
  }
  { \vdash \calD : \Sg }
\end{mathpar}
\caption{Typing rules for expressions, commands, and programs.}
\label{Fig:FullTypingElab}
\end{figure*}

\begin{figure*}
\centering
\begin{mathpar}\small
  \Rule{TV:Triv}
  {
  }
  { \etriv : \tunit }
  \and
  \Rule{TV:True}
  {
  }
  { \etrue : \tbool }
  \and
  \Rule{TV:False}
  {
  }
  { \efalse : \tbool }
  \and
  \Rule{TV:UReal}
  { r \in (0, 1) }
  { \bar{r} : \tureal }
  \and
  \Rule{TV:PReal}
  { r > 0
  }
  { \bar{r} : \tpreal }
  \and
  \Rule{TV:Real}
  {
  }
  { \bar{r} : \treal }
  \and
  \Rule{TV:FNat}
  { n < m
  }
  { \bar{n} : \tnat_m }
  \and
  \Rule{TV:Nat}
  {
  }
  { \bar{n} : \tnat }
  \and
  \Rule{TV:Clo}
  { V : \Gm \\
    \Gm \vdash \eabs{x}{\tau}{e} : \tau \to \tau'
  }
  { \vclo{V}{\eabs{x}{\tau}{e}} : \tau \to \tau' }
  \and
  \Rule{TV:Ber}
  { v  : \tureal
  }
  { \eber{v} : \tdist{\tbool} }
  \and
  \Rule{TV:Unif}
  {
  }
  { \eunif : \tdist{\tureal} }
  \and
  \Rule{TV:Beta}
  { v_1 : \tpreal \\
    v_2 : \tpreal
  }
  { \ebeta{v_1}{v_2} : \tdist{\tureal} }
  \and
  \Rule{TV:Gamma}
  { v_1 : \tpreal \\
    v_2 : \tpreal
  }
  { \egamma{v_1}{v_2} : \tdist{\tpreal} }
  \and
  \Rule{TV:Normal}
  { v_1 : \treal \\
    v_2 : \tpreal
  }
  { \enormal{v_1}{v_2} : \tdist{\treal} }
  \and
  \Rule{TE:Cat}
  { \Forall{i \in \{1, \cdots, n\}} v_i : \tpreal
  }
  { \ecat{v_1,\cdots,v_n} : \tdist{\tnat_n} }
  \and
  \Rule{TV:Geo}
  { v : \tureal
  }
  { \egeo{v} : \tdist{\tnat} }
  \and
  \Rule{TV:Pois}
  { v : \tpreal
  }
  { \epois{v} : \tdist{\tnat} }
  \\
  \Rule{TC:Empty}
  {
  }
  { \emptyset : \cdot }
  \and
  \Rule{TC:Extend}
  { V : \Gm \\
    v : \tau
  }
  { V[x \mapsto v] : \Gm, x : \tau }
  \\
  \Rule{TT:$\one$}
  {
  }
  { [] : \one }
  \and
  \Rule{TT:$\wedge$}
  { v : \tau \\
    \sigma : A
  }
  { [\msgobjl{v}] \concat \sigma : \tau \wedge A  }
  \and
  \Rule{TT:$\supset$}
  { v : \tau \\
    \sigma : A
  }
  { [\msgobjr{v}] \concat \sigma : \tau \supset A }
  \and
  \Rule{TT:$\ichoice$}
  { v : \tbool \\
    A = \m{ite}(v, A_1, A_2) \\
    \sigma : A
  }
  { [\dirobjl{v}] \concat \sigma : A_1 \ichoice A_2 }
  \and
  \Rule{TT:$\echoice$}
  { v : \tbool \\
    A = \m{ite}(v, A_1, A_2) \\
    \sigma : A
  }
  { [\dirobjr{v}] \concat \sigma : A_1 \echoice A_2 }
  \and
  \Rule{TT:$\mu$}
  { \m{typedef}(T.X.A) \in \calT \\ \sigma : [B / X] A
  }
  { [\foldobj] \concat \sigma : T[B] }  
\end{mathpar}
\caption{Type rules for values, environments, and guidance traces.}
\label{Fig:TypingAux}
\end{figure*}

\begin{proposition}\label{Lem:SoundnessDist}
  If $d : \tdist{\tau}$ and $v$ is a value, then $v : \tau$ if and only if $v \in d.\mathrm{support}$ (i.e., $d.\mathrm{density}(v) > 0$).
\end{proposition}
\begin{proof}
  Appeal to mathematical properties of primitive distributions.
\end{proof}

\begin{proposition}\label{The:SoundnessExp}\
  \begin{itemize}
    \item If $\Gm \vdash e : \tau$, $V \vdash e \evalto v$, and $V : \Gm$, then $v : \tau$.
    \item If $\Gm \vdash e : \tau$, $V : \Gm$, then there exists a value $v$ such that $V \vdash e \evalto v$.
  \end{itemize}
\end{proposition}
\begin{proof}
  Appeal to type soundness and strong normalization of the simply-typed lambda calculus.
\end{proof}

\begin{lemma}[Substitution]\label{Lem:GuideSubst}
  If $\Gm \mid \chtype{a}{A} ; \chtype{b}{B} \vdash m \dotsim \tau \mid \chtype{a}{A'} ; \chtype{b}{B'}$,
  then for any $X_a,X_b,A_o,B_o$, it holds that
  $\Gm \vdash \chtype{a}{[A_o / X_a]A} ; \chtype{b}{[B_o / X_b] B} \vdash m \dotsim \tau \mid \chtype{a}{[A_o / X_a] A'} ; \chtype{b}{[B_o / X_b] B'}$.
\end{lemma}
\begin{proof}
  By induction on the derivation of $\Gm \mid \chtype{a}{A} ; \chtype{b}{B} \vdash m \dotsim \tau \mid \chtype{a}{A'} ; \chtype{b}{B'}$.
  We show several nontrivial cases; others are similar to one of these cases.
  
  \begin{description}[labelindent=\parindent]
    \item[Case:]
    \begin{mathpar}\small
    \Rule{TM:Ret}
    { \Gm \vdash e : \tau
    }
    { \Gm \mid \chtype{a}{A} ; \chtype{b}{B} \vdash \mret{e} \dotsim \tau \mid \chtype{a}{A} ; \chtype{b}{B}  }  
    \end{mathpar}
    
    $\Gm \mid \chtype{a}{[A_o/X_a]A} ; \chtype{b}{[B_o/X_b]B} \vdash \mret{e} \dotsim \tau \mid \chtype{a}{[A_o/X_a]A} ; \chtype{b}{[B_o/X_b]B}$ \hfill (\textsc{TM:Ret})
    
    \item[Case:]
    \begin{mathpar}\small
    \Rule{TM:Bnd}
    { \Gm \mid \chtype{a}{A} ; \chtype{b}{B} \vdash m_1 \dotsim \tau_1 \mid \chtype{a}{A''} ; \chtype{b}{B''} \\
      \Gm, x : \tau_1 \mid \chtype{a}{A''} ; \chtype{b}{B''} \vdash m_2 \dotsim \tau \mid \chtype{a}{A'} ; \chtype{b}{B'}
    }
    { \Gm \mid \chtype{a}{A} ; \chtype{b}{B} \vdash \mbnd{m_1}{x}{m_2} \dotsim \tau \mid \chtype{a}{A'} ; \chtype{b}{B'} }  
    \end{mathpar}

    $\Gm \mid \chtype{a}{[A_o/X_a]A} ; \chtype{b}{[B_o/X_b]B} \vdash m_1 \dotsim \tau_1 \mid \chtype{a}{[A_o/X_a]A''}; \chtype{b}{[B_o/X_b]B''}$ \hfill (I.H.)
    
    $\Gm, x:\tau_1 \mid \chtype{a}{[A_o/X_a]A''} ; \chtype{b}{[B_o/X_b]B''} \vdash m_2 \dotsim \tau \mid \chtype{a}{[A_o/X_a]A'} ; \chtype{b}{[B_o/X_b]B'}$ \hfill (I.H.)
    
    $\Gm \mid \chtype{a}{[A_o/X_a]A} ; \chtype{b}{[B_o/X_b]B} \vdash \mbnd{m_1}{x}{m_2} \dotsim \tau \mid \chtype{a}{[A_o/X_a]A''} ; \chtype{b}{[B_o/X_b]B''}$ \hfill (\textsc{TM:Bnd})
    
    \item[Case:]
    \begin{mathpar}\small
    \Rule{TM:Call}
    { \Sg(f) = \tau_1 \leadsto \tau \mid \chtype{a}{T_a} ; \chtype{b}{T_b} \\
      \Gm \vdash e : \tau_1
    }
    { \Gm \mid \chtype{a}{T_a[A']} ; \chtype{b}{T_b[B']} \vdash \mcall{f}{e} \dotsim \tau \mid \chtype{a}{A'} ; \chtype{b}{B'} }  
    \end{mathpar}
    
    $\Gm \mid \chtype{a}{T_a[[A_o/X_a]A']} ; \chtype{b}{T_b[[B_o/X_b]B']} \vdash \mcall{f}{e} \dotsim \tau \mid \chtype{a}{[A_o/X_a]A'} ; \chtype{b}{[B_o/X_b]B'}$ \hfill (\textsc{TM:Call})
    
    $T_a[[A_o/X_a]A'] = [A_o/X_a] (T_a[A'])$, $T_b[[B_o/X_b]B'] = [B_o/X_b] (T_b[B'])$
    
    $A = T_a[A']$, $B = T_b[B']$ \hfill (assumption)
  \end{description}
\end{proof}

\begin{theorem}[Well-typed programs evaluate to well-typed values]\label{Lem:SoundnessValue}
  If
  $\Gm \mid \chtype{a}{A} ; \chtype{b}{B} \vdash m \dotsim \tau \mid \chtype{a}{A'} ; \chtype{b}{B'}$,
  $V \mid \chtype{a}{\sigma_a} ; \chtype{b}{\sigma_b} \vdash m \evalp{w} v$, and
  $V : \Gm$,
  then
  $v : \tau$.
\end{theorem}
\begin{proof}
  By induction on the derivation of $V \mid \chtype{a}{\sigma_a} ; \chtype{b}{\sigma_b} \vdash m \evalp{w} v$, followed by inversion on $\Gm \mid \chtype{a}{A} ; \chtype{b}{B} \vdash m \dotsim \tau \mid \chtype{a}{A'} ; \chtype{b}{B'}$.
  We show several nontrivial cases; others are similar to one of these cases.
  
  \begin{description}[labelindent=\parindent]
    \item[Case:]
    \begin{mathpar}\small
    \Rule{EM:Ret}
    { V \vdash e \evalto v
    }
    { V \mid \chtype{a}{[]} ; \chtype{b}{[]} \vdash \mret{e} \evalp{1} v }
    \and
    \Rule{TM:Ret}
    { \Gm \vdash e : \tau
    }
    { \Gm \mid \chtype{a}{A} ; \chtype{b}{B} \vdash \mret{e} \dotsim \tau \mid \chtype{a}{A} ; \chtype{b}{B} }
    \end{mathpar}
    
    $\Gm \vdash e : \tau$, $V \vdash e \evalto v$ \hfill (assumption)
    
    $v : \tau$ \hfill (\cref{The:SoundnessExp})
    
    \item[Case:]
    \begin{mathpar}\small
      \Rule{EM:Bnd}
      { V \mid \chtype{a}{\sigma_{a,1}} ; \chtype{b}{\sigma_{b,1}} \vdash m_1 \evalp{w_1} v_1 \\\\
        V[x \mapsto v_1] \mid \chtype{a}{\sigma_{a,2}} ; \chtype{b}{\sigma_{b,2}} \vdash m_2 \evalp{w_2} v
      }
      { V \mid \chtype{a}{\sigma_{a,1} \concat \sigma_{a,2}} ; \chtype{b}{\sigma_{b,1} \concat \sigma_{b,2}} \vdash \mbnd{m_1}{x}{m_2} \evalp{w_1 \cdot w_2} v }
      \and
      \Rule{TM:Bnd}
      { \Gm \mid \chtype{a}{A} ; \chtype{b}{B} \vdash m_1 \dotsim \tau_1 \mid \chtype{a}{A''} ; \chtype{b}{B''} \\\\
        \Gm, x : \tau_1 \mid \chtype{a}{A''} ; \chtype{b}{B''} \vdash m_2 \dotsim \tau \mid \chtype{a}{A'} ; \chtype{b}{B'}
      }
      { \Gm \mid \chtype{a}{A} ; \chtype{b}{B} \vdash \mbnd{m_1}{x}{m_2} \dotsim \tau \mid \chtype{a}{A'} ; \chtype{b}{B'} }
    \end{mathpar}
    
    $\Gm \mid \chtype{a}{A} ; \chtype{b}{B} \vdash m_1 \dotsim \tau_1 \mid \chtype{a}{A''} ; \chtype{b}{B''}$, $V \mid \chtype{a}{\sigma_{a,1}} ; \chtype{b}{\sigma_{b,1}} \vdash m_1 \evalp{w_1} v_1$, $V : \Gm$ \hfill (assumption)
    
    $v_1 : \tau_1$ \hfill (I.H.)
    
    $V[x \mapsto v_1] : (\Gm,x:\tau_1)$ \hfill (\textsc{TC:Extend})
    
    $\Gm,x : \tau_1 \mid \chtype{a}{A''} ; \chtype{b}{B''} \vdash m_2 \dotsim \tau \mid \chtype{a}{A'} ; \chtype{b}{B'}$, $V[x \mapsto v_1] \mid \chtype{a}{\sigma_{a,2}} ; \chtype{b}{\sigma_{b,2}} \vdash m_2 \evalp{w_2} v$ \hfill (assumption)  
    
    $v : \tau$ \hfill (I.H.)
    
    \item[Case:]
    \begin{mathpar}\small
      \Rule{EM:Call}
      { V \vdash e \evalto v_1 \\
        \calD(f) = \fundec{f}{\tau_1}{\tau}{x_f}{m_f}{a}{b} \\
        \emptyset[x_f \mapsto v_1] \mid \chtype{a}{\sigma_{a,1}} ; \chtype{b}{\sigma_{b,1}} \vdash m_f \evalp{w} v 
      }
      { V \mid \chtype{a}{[\foldobj] \concat \sigma_{a,1}} ; \chtype{b}{[\foldobj] \concat \sigma_{b,1}} \vdash \mcall{f}{e} \evalp{w} v }
      \and
      \Rule{TM:Call}
      { \Sg(f) = \tau_1 \leadsto \tau \mid \chtype{a}{T_a} ; \chtype{b}{T_b} \\
        \Gm \vdash e : \tau_1
      }
      { \Gm \mid \chtype{a}{T_a[A']} ; \chtype{b}{T_b[B']} \vdash \mcall{f}{e} \dotsim \tau \mid \chtype{a}{A'} ; \chtype{b}{B'} }
      \and
      \Rule{TP:Dec}
      { \m{typedef}(T_a.X_a.A''), \m{typedef}(T_b.X_b.B'') \in \calT \\\\
        x_f : \tau_1 \mid \chtype{a}{A''} ; \chtype{b}{B''} \vdash_{\Sg} m_f \dotsim \tau \mid \chtype{a}{X_a} ; \chtype{b}{X_b} 
      }
      { \vdash_{\Sg} \fundec{f}{\tau_1}{\tau}{x_f}{m_f}{a}{b} : \tau_1 \leadsto \tau \mid \chtype{a}{T_a} ; \chtype{b}{T_b} }
    \end{mathpar}
    
    $\Gm \vdash e : \tau_1$, $V \vdash e \evalto v_1$, $V : \Gm$ \hfill (assumption)
    
    $v_1 : \tau_1$ \hfill (\cref{The:SoundnessExp})
    
    $\emptyset[x_f \mapsto v_1] : (x_f : \tau_1)$ \hfill (\textsc{TC:Extend})
    
    $x_f : \tau_1 \mid \chtype{a}{A''} ; \chtype{b}{B''} \vdash_{\Sg} m_f \dotsim \tau \mid \chtype{a}{X_a} ; \chtype{b}{X_b}$ \hfill (assumption)
  
    $x_f : \tau_1 \mid \chtype{a}{[A'/X_a]A''} ; \chtype{b}{[B'/X_b]B''} \vdash m_f \dotsim \tau \mid \chtype{a}{A'} ; \chtype{b}{B'}$ \hfill (\cref{Lem:GuideSubst})
    
    $\emptyset[x_f \mapsto v_1] \mid \chtype{a}{\sigma_{a,1}} ; \chtype{b}{\sigma_{b,1}} \vdash m_f \evalp{w} v$ \hfill (assumption)
    
    $v : \tau$ \hfill (I.H.)
    
    \item[Case:]
    The reasoning below also works for \textsc{(TM:Sample:Recv:R)} and \textsc{(TM:Sample:Send:*)}.
    \begin{mathpar}\small
      \Rule{EM:Sample:Recv:L}
      { V \vdash e \evalto d \\
        v \in d.\mathrm{support} \\
        w = d.\mathrm{density}(v)
      }
      { V \mid \chtype{a}{[\msgobjl{v}]} ; \chtype{b}{[]} \vdash \msamplei{e}{a} \evalp{w} v }  
      \and
      \Rule{TM:Sample:Recv:L}
      { \Gm \vdash e : \tdist{\tau}
      }
      { \Gm \mid \chtype{a}{\tau \wedge A'} ; \chtype{b}{B'} \vdash \msamplei{e}{a} \dotsim \tau \mid \chtype{a}{A'} ; \chtype{b}{B'} }
    \end{mathpar}
    
    $\Gm \vdash e : \tdist{\tau}$, $V \vdash e \evalto d$, $V : \Gm$ \hfill (assumption)
    
    $d : \tdist{\tau}$ \hfill (\cref{The:SoundnessExp})
    
    $v \in d.\mathrm{support}$ \hfill (assumption)
    
    $v : \tau$ \hfill (\cref{Lem:SoundnessDist})
    
    \item[Case:]
    The reasoning below also works for \textsc{(TM:Cond:Recv:R)}.
    \begin{mathpar}\small
    \Rule{EM:Cond:Recv:L}
    { i = \m{ite}(v_a,1,2) \\
      V \mid \chtype{a}{\sigma_{a,1}} ; \chtype{b}{\sigma_b} \vdash m_i \evalp{w} v
    }
    { V \mid \chtype{a}{[ \dirobjl{v_a} ] \concat \sigma_{a,1}} ; \chtype{b}{\sigma_b} \vdash \mbranchi{m_1}{m_2}{a} \evalp{w} v }
    \and
    \Rule{TM:Cond:Recv:L}
    { \Gm \mid \chtype{a}{A_1} ; \chtype{b}{B} \vdash m_1 \dotsim \tau \mid \chtype{a}{A'} ; \chtype{b}{B'} \\
      \Gm \mid \chtype{a}{A_2} ; \chtype{b}{B} \vdash m_2 \dotsim \tau \mid \chtype{a}{A'} ; \chtype{b}{B'}
    }
    { \Gm \mid \chtype{a}{A_1 \ichoice A_2} ; \chtype{b}{B} \vdash \mbranchi{m_1}{m_2}{a} \dotsim \tau \mid \chtype{a}{A'} ; \chtype{b}{B'} }
    \end{mathpar}
    
    $V \mid \chtype{a}{\sigma_{a,1}} ; \chtype{b}{\sigma_b} \vdash m_i \evalp{w} v$,
    $\Gm \vdash \chtype{a}{A_i} ; \chtype{b}{B} \vdash m_i \dotsim \tau \mid \chtype{a}{A'} ; \chtype{b}{B'}$ \hfill (assumption)
    
    $v : \tau$ \hfill (I.H.)
    
    \item[Case:]
    The reasoning below also works for (\textsc{TM:Cond:Send:R}).
    \begin{mathpar}\small
    \Rule{EM:Cond:Send:L}
    { V \vdash e \evalto v_e \\
      i = \m{ite}(v_a, 1, 2) \\
      V \mid \chtype{a}{\sigma_{a,1}} ; \chtype{b}{\sigma_b} \vdash m_i \evalp{w} v
    }
    { V \mid \chtype{a}{[\dirobjr{v_a}] \concat \sigma_{a,1}} ; \chtype{b}{\sigma_b} \vdash \mbrancho{e}{m_1}{m_2}{a} \evalp{w \cdot [v_a=v_e]} v }
    \and
    \Rule{TM:Cond:Send:L}
    { \Gm \vdash e : \tbool \\
      \Gm \mid \chtype{a}{A_1} ; \chtype{b}{B} \vdash m_1 \dotsim \tau \mid \chtype{a}{A'} ; \chtype{b}{B'} \\
      \Gm \mid \chtype{a}{A_2} ; \chtype{b}{B} \vdash m_2 \dotsim \tau \mid \chtype{a}{A'} ; \chtype{b}{B'}
    }
    { \Gm \mid \chtype{a}{A_1 \echoice A_2} ; \chtype{b}{B} \vdash \mbrancho{e}{m_1}{m_2}{a} \dotsim \tau \mid \chtype{a}{A'} ; \chtype{b}{B'} }
    \end{mathpar}
    
    $V \mid \chtype{a}{\sigma_{a,1}} ; \chtype{b}{\sigma_b} \vdash m_i \evalp{w} v$,
    $\Gm \vdash \chtype{a}{A_i} ; \chtype{b}{B} \vdash m_i \dotsim \tau \mid \chtype{a}{A'} ; \chtype{b}{B'}$ \hfill (assumption)
    
    $v : \tau$ \hfill (I.H.)
  \end{description}
\end{proof}

\begin{theorem}[Well-typed programs produce well-typed traces]\label{Lem:SoundnessW}
  If
  $\Gm \mid \chtype{a}{A} ; \chtype{b}{B} \vdash m \dotsim \tau \mid \chtype{a}{A'} ; \chtype{b}{B'}$,
  $V \mid \chtype{a}{\sigma_a} ; \chtype{b}{\sigma_b} \vdash m \evalp{w} v$,
  $V : \Gm$,
  $\sigma_a' : A'$, and
  $\sigma_b' : B'$,
  then
  $(\sigma_a \concat \sigma_a') : A$ and
  $(\sigma_b \concat \sigma_b') : B$.
\end{theorem}
\begin{proof}
  By induction on the derivation of $V \mid \chtype{a}{\sigma_a} ; \chtype{b}{\sigma_b} \vdash m \evalp{w} v$, followed by inversion on $\Gm \mid \chtype{a}{A} ; \chtype{b}{B} \vdash m \dotsim \tau \mid \chtype{a}{A'} ; \chtype{b}{B'}$.
  We show several nontrivial cases; others are similar to one of these cases.
  
  \begin{description}[labelindent=\parindent]
    \item[Case:]
    \begin{mathpar}\small
    \Rule{EM:Ret}
    { V \vdash e \evalto v
    }
    { V \mid \chtype{a}{[]} ; \chtype{b}{[]} \vdash \mret{e} \evalp{1} v }
    \and
    \Rule{TM:Ret}
    { \Gm \vdash e : \tau
    }
    { \Gm \mid \chtype{a}{A} ; \chtype{b}{B} \vdash \mret{e} \dotsim \tau \mid \chtype{a}{A} ; \chtype{b}{B} }
    \end{mathpar}
    
    $\sigma_a = []$, $\sigma_b = []$ \hfill (assumption)
    
    $(\sigma_a \concat \sigma_a') = \sigma_a'$, $(\sigma_b \concat \sigma_b') = \sigma_b'$
    
    $\sigma_a' : A'$, $\sigma_b' : B'$, $A = A'$, $B = B'$ \hfill (assumption)
    
    $(\sigma_a \concat \sigma_a') : A$, $(\sigma_b \concat \sigma_b') : B$
    
    \item[Case:]
    \begin{mathpar}\small
      \Rule{EM:Bnd}
      { V \mid \chtype{a}{\sigma_{a,1}} ; \chtype{b}{\sigma_{b,1}} \vdash m_1 \evalp{w_1} v_1 \\\\
        V[x \mapsto v_1] \mid \chtype{a}{\sigma_{a,2}} ; \chtype{b}{\sigma_{b,2}} \vdash m_2 \evalp{w_2} v
      }
      { V \mid \chtype{a}{\sigma_{a,1} \concat \sigma_{a,2}} ; \chtype{b}{\sigma_{b,1} \concat \sigma_{b,2}} \vdash \mbnd{m_1}{x}{m_2} \evalp{w_1 \cdot w_2} v }
      \and
      \Rule{TM:Bnd}
      { \Gm \mid \chtype{a}{A} ; \chtype{b}{B} \vdash m_1 \dotsim \tau_1 \mid \chtype{a}{A''} ; \chtype{b}{B''} \\\\
        \Gm, x : \tau_1 \mid \chtype{a}{A''} ; \chtype{b}{B''} \vdash m_2 \dotsim \tau \mid \chtype{a}{A'} ; \chtype{b}{B'}
      }
      { \Gm \mid \chtype{a}{A} ; \chtype{b}{B} \vdash \mbnd{m_1}{x}{m_2} \dotsim \tau \mid \chtype{a}{A'} ; \chtype{b}{B'} }
    \end{mathpar}
    
    $v_1 : \tau_1$ \hfill (\cref{Lem:SoundnessValue})
    
    $V : \Gm$ \hfill (assumption)
    
    $V[x \mapsto v_1] : (\Gm, x : \tau_1)$ \hfill (\textsc{TC:Extend})
    
    $\sigma_a' :A'$, $\sigma_b' : B'$ \hfill (assumption)
    
    $V[x \mapsto v_1] \mid \chtype{a}{\sigma_{a,2}} ; \chtype{b}{\sigma_{b,2}} \vdash m_2 \evalp{w_2} v$, $\Gm,x:\tau_1 \mid \chtype{a}{A''} ; \chtype{b}{B''} \vdash m_2 \dotsim \tau \mid \chtype{a}{A'} ; \chtype{b}{B'}$ \hfill (assumption)
    
    $(\sigma_{a,2} \concat \sigma_a') : A''$, $(\sigma_{b,2} \concat \sigma_b') : B''$ \hfill (I.H.)
    
    $V : \Gm$, $V \mid \chtype{a}{\sigma_{a,1}} ; \chtype{b}{\sigma_{b,1}} \vdash m_1 \evalp{w_1} v_1$, $\Gm \mid \chtype{a}{A} ; \chtype{b}{B} \vdash m_1 \dotsim \tau_1 \mid \chtype{a}{A''} ; \chtype{b}{B''}$ \hfill (assumption)
    
    $(\sigma_{a,1} \concat \sigma_{a,2} \concat \sigma_a') : A$, $(\sigma_{b,1} \concat \sigma_{b,2} \concat \sigma_b') : B$ \hfill (I.H.)
    
    $\sigma_a = \sigma_{a,1} \concat \sigma_{a,2}$, $\sigma_b = \sigma_{b,1} \concat \sigma_{b,2}$ \hfill (assumption)
    
    \item[Case:]
    \begin{mathpar}\small
      \Rule{EM:Call}
      { V \vdash e \evalto v_1 \\
        \calD(f) = \fundec{f}{\tau_1}{\tau}{x_f}{m_f}{a}{b} \\
        \emptyset[x_f \mapsto v_1] \mid \chtype{a}{\sigma_{a,1}} ; \chtype{b}{\sigma_{b,1}} \vdash m_f \evalp{w} v 
      }
      { V \mid \chtype{a}{[\foldobj] \concat \sigma_{a,1}} ; \chtype{b}{[\foldobj] \concat \sigma_{b,1}} \vdash \mcall{f}{e} \evalp{w} v }
      \and
      \Rule{TM:Call}
      { \Sg(f) = \tau_1 \leadsto \tau \mid \chtype{a}{T_a} ; \chtype{b}{T_b} \\
        \Gm \vdash e : \tau_1
      }
      { \Gm \mid \chtype{a}{T_a[A']} ; \chtype{b}{T_b[B']} \vdash \mcall{f}{e} \dotsim \tau \mid \chtype{a}{A'} ; \chtype{b}{B'} }
      \and
      \Rule{TP:Dec}
      { \m{typedef}(T_a.X_a.A''), \m{typedef}(T_b.X_b.B'') \in \calT \\\\
        x_f : \tau_1 \mid \chtype{a}{A''} ; \chtype{b}{B''} \vdash_{\Sg} m_f \dotsim \tau \mid \chtype{a}{X_a} ; \chtype{b}{X_b} 
      }
      { \vdash_{\Sg} \fundec{f}{\tau_1}{\tau}{x_f}{m_f}{a}{b} : \tau_1 \leadsto \tau \mid \chtype{a}{T_a} ; \chtype{b}{T_b} }
    \end{mathpar}
    
    $\Gm \vdash e : \tau_1$, $V \vdash e \evalto v_1$, $V : \Gm$ \hfill (assumption)
    
    $v_1 : \tau_1$ \hfill (\cref{The:SoundnessExp})
    
    $\emptyset[x_f \mapsto v_1] : (x_f : \tau_1)$ \hfill (\textsc{TC:Extend})
    
    $x_f : \tau_1 \mid \chtype{a}{A''} ; \chtype{b}{B''} \vdash_{\Sg} m_f \dotsim \tau \mid \chtype{a}{X_a} ; \chtype{b}{X_b}$ \hfill (assumption)
  
    $x_f : \tau_1 \mid \chtype{a}{[A'/X_a]A''} ; \chtype{b}{[B'/X_b]B''} \vdash m_f \dotsim \tau \mid \chtype{a}{A'} ; \chtype{b}{B'}$ \hfill (\cref{Lem:GuideSubst})
    
    $\sigma_a' : A'$, $\sigma_b' : B'$, $\emptyset[x_f \mapsto v_1] \mid \chtype{a}{\sigma_{a,1}} ; \chtype{b}{\sigma_{b,1}} \vdash m_f \evalp{w} v$ \hfill (assumption)
    
    $(\sigma_{a,1} \concat \sigma_a') : [A'/X_a]A''$, $(\sigma_{b,1} \concat \sigma_b') : [B'/X_b]B''$ \hfill (I.H.)
    
    $([\foldobj] \concat \sigma_{a,1} \concat \sigma_a') : T_a[A']$, $[\foldobj] \concat \sigma_{b,1} \concat \sigma_b') : T_b[B']$ \hfill (\textsc{TT:$\mu$})
    
    $\sigma_a = [\foldobj] \concat \sigma_{a,1}$, $\sigma_b = [\foldobj] \concat \sigma_{b,1}$ \hfill (assumption)
    
    \item[Case:]
    The reasoning below also works for \textsc{(TM:Sample:Recv:R)} and \textsc{(TM:Sample:Send:*)}.
    \begin{mathpar}\small
      \Rule{EM:Sample:Recv:L}
      { V \vdash e \evalto d \\
        v \in d.\mathrm{support} \\
        w = d.\mathrm{density}(v)
      }
      { V \mid \chtype{a}{[\msgobjl{v}]} ; \chtype{b}{[]} \vdash \msamplei{e}{a} \evalp{w} v }  
      \and
      \Rule{TM:Sample:Recv:L}
      { \Gm \vdash e : \tdist{\tau}
      }
      { \Gm \mid \chtype{a}{\tau \wedge A'} ; \chtype{b}{B'} \vdash \msamplei{e}{a} \dotsim \tau \mid \chtype{a}{A'} ; \chtype{b}{B'} }
    \end{mathpar}
    
    $v : \tau$ \hfill (\cref{Lem:SoundnessValue})
    
    $\sigma_a' : A'$, $\sigma_b' : B'$ \hfill (assumption)
    
    $([\msgobjl{v}] \concat \sigma_a') : \tau \wedge A'$, $([] \concat \sigma_b') : B'$ \hfill (\textsc{TT:$\wedge$})
    
    $A = \tau \wedge A'$, $B = B'$ \hfill (assumption)
    
    \item[Case:]
    The reasoning below also works for \textsc{(TM:Cond:Recv:R)}.
    \begin{mathpar}\small
    \Rule{EM:Cond:Recv:L}
    { i = \m{ite}(v_a,1,2) \\
      V \mid \chtype{a}{\sigma_{a,1}} ; \chtype{b}{\sigma_b} \vdash m_i \evalp{w} v
    }
    { V \mid \chtype{a}{[ \dirobjl{v_a} ] \concat \sigma_{a,1}} ; \chtype{b}{\sigma_b} \vdash \mbranchi{m_1}{m_2}{a} \evalp{w} v }
    \and
    \Rule{TM:Cond:Recv:L}
    { \Gm \mid \chtype{a}{A_1} ; \chtype{b}{B} \vdash m_1 \dotsim \tau \mid \chtype{a}{A'} ; \chtype{b}{B'} \\
      \Gm \mid \chtype{a}{A_2} ; \chtype{b}{B} \vdash m_2 \dotsim \tau \mid \chtype{a}{A'} ; \chtype{b}{B'}
    }
    { \Gm \mid \chtype{a}{A_1 \ichoice A_2} ; \chtype{b}{B} \vdash \mbranchi{m_1}{m_2}{a} \dotsim \tau \mid \chtype{a}{A'} ; \chtype{b}{B'} }
    \end{mathpar}
    
    $V \mid \chtype{a}{\sigma_{a,1}} ; \chtype{b}{\sigma_b} \vdash m_i \evalp{w} v$, $\Gm \mid \chtype{a}{A_i} ; \chtype{b}{B} \vdash m_i \dotsim \tau \mid \chtype{a}{A'} ; \chtype{b}{B'}$ \hfill (assumption)
    
    $(\sigma_{a,1} \concat \sigma_a') : A_i$, $(\sigma_b \concat \sigma_b') : B$ \hfill (I.H.)
    
    $v_a = \etrue$ or $v_a = \efalse$ (for $\m{ite}(-)$ to be well-defined) \hfill (assumption)
    
    $v_a : \tbool$
    
    $A_i = \m{ite}(v_a,A_1,A_2)$ \hfill (assumption)
    
    $([\dirobjl{v_a}] \concat \sigma_{a,1} \concat \sigma_a') : A_1 \ichoice A_2$ \hfill (\textsc{TT:$\ichoice$})
    
    $A = A_1 \ichoice A_2$ \hfill (assumption)
    
    \item[Case:]
    The reasoning below also works for (\textsc{TM:Cond:Send:R}).
    \begin{mathpar}\small
    \Rule{EM:Cond:Send:L}
    { V \vdash e \evalto v_e \\
      i = \m{ite}(v_a, 1, 2) \\
      V \mid \chtype{a}{\sigma_{a,1}} ; \chtype{b}{\sigma_b} \vdash m_i \evalp{w} v
    }
    { V \mid \chtype{a}{[\dirobjr{v_a}] \concat \sigma_{a,1}} ; \chtype{b}{\sigma_b} \vdash \mbrancho{e}{m_1}{m_2}{a} \evalp{w \cdot [v_a=v_e]} v }
    \and
    \Rule{TM:Cond:Send:L}
    { \Gm \vdash e : \tbool \\
      \Gm \mid \chtype{a}{A_1} ; \chtype{b}{B} \vdash m_1 \dotsim \tau \mid \chtype{a}{A'} ; \chtype{b}{B'} \\
      \Gm \mid \chtype{a}{A_2} ; \chtype{b}{B} \vdash m_2 \dotsim \tau \mid \chtype{a}{A'} ; \chtype{b}{B'}
    }
    { \Gm \mid \chtype{a}{A_1 \echoice A_2} ; \chtype{b}{B} \vdash \mbrancho{e}{m_1}{m_2}{a} \dotsim \tau \mid \chtype{a}{A'} ; \chtype{b}{B'} }
    \end{mathpar}
    
    $V \mid \chtype{a}{\sigma_{a,1}} ; \chtype{b}{\sigma_b} \vdash m_i \evalp{w} v$, $\Gm \mid \chtype{a}{A_i} ; \chtype{b}{B} \vdash m_i \dotsim \tau \mid \chtype{a}{A'} ; \chtype{b}{B'}$ \hfill (assumption)
    
    $(\sigma_{a,1} \concat \sigma_a') : A_i$, $(\sigma_b \concat \sigma_b') : B$ \hfill (I.H.)
    
    $v_a = \etrue$ or $v_a = \efalse$ (for $\m{ite}(-)$ to be well-defined) \hfill (assumption)
    
    $v_a : \tbool$
    
    $A_i = \m{ite}(v_a,A_1,A_2)$ \hfill (assumption)
    
    $([\dirobjr{v_a}] \concat \sigma_{a,1} \concat \sigma_a') : A_1 \echoice A_2$ \hfill (\textsc{TT:$\echoice$})
    
    $A = A_1 \echoice A_2$ \hfill (assumption)
  \end{description}
\end{proof}

\begin{corollary*}[\cref{Cor:SoundnessW}]
  If $\cdot \mid \chtype{a}{A} ; \chtype{b}{B} \vdash_\Sg m \dotsim \tau \mid \chtype{a}{\one} ; \chtype{b}{\one}$
  and $\emptyset \mid \chtype{a}{\sigma_a} ; \chtype{b}{\sigma_b} \vdash m \evalp{w} v$,
  then
  $\sigma_a : A$,
  $\sigma_b : B$,
  and $v : \tau$.
\end{corollary*}
\begin{proof}
  Appeal to \cref{Lem:SoundnessValue,Lem:SoundnessW}.
\end{proof}

\begin{theorem}[Normalization, part I]\label{The:SoundnessI}
  If
  $\Gm \mid \chtype{a}{A} ; \chtype{b}{B} \vdash m \dotsim \tau \mid \chtype{a}{A'} ; \chtype{b}{B'}$,
  $V : \Gm$,
  $\sigma_a : A$,
  and $\sigma_b : B$,
  then
  there exist $w,v,\sigma_a'',\sigma_a',\sigma_b'',\sigma_b'$ such that
  $V \mid \chtype{a}{\sigma_a''} ; \chtype{b}{\sigma_b''} \vdash m \evalp{w} v$,
  $\sigma_a = \sigma_a'' \concat \sigma_a'$,
  $\sigma_b = \sigma_b'' \concat \sigma_b'$,
  $\sigma_a' : A'$,
  and $\sigma_b' : B'$.
\end{theorem}
\begin{proof}
  By nested induction on the derivation of $\sigma_a : A$, $\sigma_b : B$, and $\Gm \mid \chtype{a}{A} ; \chtype{b}{B} \vdash m \dotsim \tau \mid \chtype{a}{A'} ; \chtype{b}{B'}$.
  We show several nontrivial cases; others are similar to one of these cases.
  
  \begin{description}[labelindent=\parindent]
    \item[Case:]
    \[\small
    \Rule{TM:Ret}
    { \Gm \vdash e : \tau
    }
    { \Gm \mid \chtype{a}{A} ; \chtype{b}{B} \vdash \mret{e} \dotsim \tau \mid \chtype{a}{A} ; \chtype{b}{B} }
    \]
    
    $V : \Gm$, $\Gm \vdash e : \tau$ \hfill (assumption)
    
    $V \vdash e \evalto v$ for some $v$ \hfill (\cref{The:SoundnessExp})
    
    $V \mid \chtype{a}{[]} ; \chtype{b}{[]} \vdash \mret{e} \evalp{1} v$ \hfill (\textsc{EM:Ret})
    
    $\sigma_a : A$, $\sigma_b : B$ \hfill (assumption)
    
    \item[Case:]
    \[\small
    \Rule{TM:Bnd}
    { \Gm \mid \chtype{a}{A} ; \chtype{b}{B} \vdash m_1 \dotsim \tau_1 \mid \chtype{a}{A''} ; \chtype{b}{B''} \\
      \Gm, x : \tau_1 \mid \chtype{a}{A''} ; \chtype{b}{B''} \vdash m_2 \dotsim \tau \mid \chtype{a}{A'} ; \chtype{b}{B'}
    }
    { \Gm \mid \chtype{a}{A} ; \chtype{b}{B} \vdash \mbnd{m_1}{x}{m_2} \dotsim \tau \mid \chtype{a}{A'} ; \chtype{b}{B'} }
    \]
    
    $V \mid \chtype{a}{\sigma_{a,1}} ; \chtype{b}{\sigma_{b,1}} \vdash m_1 \evalp{w_1} v_1$ s.t. 
    $\sigma_a = \sigma_{a,1} \concat \sigma_{a,2}$, $\sigma_b = \sigma_{b,1} \concat \sigma_{b,2}$, $\sigma_{a,2} : A''$, $\sigma_{b,2} : B''$ \hfill (I.H.)
    
    $v_1 : \tau_1$ \hfill (\cref{Lem:SoundnessValue})
    
    $V[x \mapsto v_1] : (\Gm, x : \tau_1)$ \hfill (\textsc{TC:Extend})
    
    $V[x \mapsto v_1] \mid \chtype{a}{\sigma_{a,2}'} \mid \chtype{b}{\sigma_{b,2}'} \vdash m_2 \evalp{w_2} v$ s.t.
    $\sigma_{a,2} = \sigma_{a,2}' \concat \sigma_a'$, $\sigma_{b,2} = \sigma_{b,2}' \concat \sigma_b'$, $\sigma_a' : A'$, $\sigma_b' : B'$ \hfill (I.H.)
    
    Let $\sigma_a'' \defeq \sigma_{a,1} \concat \sigma_{a,2}'$, $\sigma_b'' \defeq \sigma_{b,1} \concat \sigma_{b,2}'$, thus $\sigma_a = \sigma_{a,1} \concat \sigma_{a,2} = \sigma_{a,1} \concat \sigma_{a,2}' \concat \sigma_a' = \sigma_a'' \concat \sigma_a'$, $\sigma_b = \sigma_b'' \concat \sigma_b'$
    
    $V \mid \chtype{a}{\sigma_a''} ; \chtype{b}{\sigma_b''} \vdash \mbnd{m_1}{x}{m_2} \evalp{w_1 \cdot w_2} v$ \hfill (\textsc{EM:Bnd})
    
    \item[Case:]
    \[\small
    \Rule{TM:Call}
    { \Sg(f) = \tau_1 \leadsto \tau \mid \chtype{a}{T_a} ; \chtype{b}{T_b} \\
      \Gm \vdash e : \tau_1
    }
    { \Gm \mid \chtype{a}{T_a[A']} ; \chtype{b}{T_b[B']} \vdash \mcall{f}{e} \dotsim \tau \mid \chtype{a}{A'} ; \chtype{b}{B'} }
    \]
    
    $V : \Gm$, $\Gm \vdash e : \tau_1$ \hfill (assumption)
    
    $V \vdash e \evalto v_1$ for some $v_1$ s.t. $v_1 : \tau_1$ \hfill (\cref{The:SoundnessExp})
    
    Let $\calD(f) = \fundec{f}{\tau_1}{\tau}{x_f}{m_f}{a}{b}$
    
    $\emptyset[x_f \mapsto v_1] : (x_f : \tau_1)$ \hfill (\textsc{TC:Extend})
    
    $x_f : \tau_1 \mid \chtype{a}{A_o} ; \chtype{b}{B_o} \vdash m_f \dotsim \tau \mid \chtype{a}{X_a} ; \chtype{b}{X_b}$, $\m{typedef}(T_a.X_a.A_o),\m{typedef}(T_b.X_b.B_o) \in \calT$ \hfill (assumption)
    
    $x_f : \tau_1 \mid \chtype{a}{[A'/X_a]A_o} ; \chtype{b}{[B'/X_b]B_o} \vdash m_f \dotsim \tau \mid \chtype{a}{A'} ; \chtype{b}{B'}$ \hfill (\cref{Lem:GuideSubst})
    
    $A = T_a[A']$, $\sigma_a : T_a[A']$ \hfill (assumption)
    
    $\sigma_a = [\foldobj] \concat \sigma_{a,\m{t}}$, $\sigma_{a,\m{t}} : [A' / X_a] A_o$ \hfill (inversion)
    
    $B = T_b[B']$, $\sigma_b : T_b[B']$ \hfill (assumption)
    
    $\sigma_b = [\foldobj] \concat \sigma_{b,\m{t}}$, $\sigma_{b,\m{t}} : [B' / X_b] B_o$ \hfill (inversion)
        
    $\emptyset[x_f \mapsto v_1] \mid \chtype{a}{\sigma_{a,\m{t}}'} ; \chtype{b}{\sigma_{b,\m{t}}'} \vdash m_f \evalp{w} v$ s.t.
    $\sigma_{a,\m{t}} = \sigma_{a,\m{t}}' \concat \sigma_a'$, $\sigma_{b,\m{t}} = \sigma_{b,\m{t}}' \concat \sigma_b'$, $\sigma_a' : A'$, $\sigma_b' : B'$ \hfill (I.H.)
    
    Let $\sigma_a'' \defeq [\foldobj] \concat \sigma_{a,\m{t}}'$, $\sigma_b'' \defeq [\foldobj] \concat \sigma_{b,\m{t}}'$
    
    $V \mid \chtype{a}{\sigma_a''} ; \chtype{b}{\sigma_b''} \vdash \mcall{f}{e} \evalp{w} v$ \hfill (\textsc{EM:Call})
    
    \item[Case:]
    The reasoning below also works for \textsc{(TM:Cond:Recv:R)}.
    \[\small
    \Rule{TM:Cond:Recv:L}
    { \Gm \mid \chtype{a}{A_1} ; \chtype{b}{B} \vdash m_1 \dotsim \tau \mid \chtype{a}{A'} ; \chtype{b}{B'} \\
      \Gm \mid \chtype{a}{A_2} ; \chtype{b}{B} \vdash m_2 \dotsim \tau \mid \chtype{a}{A'} ; \chtype{b}{B'}
    }
    { \Gm \mid \chtype{a}{A_1 \ichoice A_2} ; \chtype{b}{B} \vdash \mbranchi{m_1}{m_2}{a} \dotsim \tau \mid \chtype{a}{A'} ; \chtype{b}{B'} }
    \]
        
    $\sigma_a : A_1 \ichoice A_2$ \hfill (assumption)
    
    $\sigma_a = [\dirobjl{v_a}] \concat \sigma_{a,\m{t}}$, $\sigma_{a,\m{t}} : \m{ite}(v_a,A_1,A_2)$ \hfill (inversion)
    
    \begin{description}
      \item[Subcase:] $v_a = \etrue$, $\sigma_{a,\m{t}} : A_1$
            
      $V \mid \chtype{a}{\sigma_{a,\m{t}}'} ; \chtype{b}{\sigma_b''} \vdash m_1 \evalp{w} v$ s.t.
      $\sigma_{a,\m{t}} = \sigma_{a,\m{t}}' \concat \sigma_a'$, $\sigma_b = \sigma_b'' \concat \sigma_b'$, $\sigma_a' : A'$, $\sigma_b' : B'$ \hfill (I.H.)
      
      Let $\sigma_a'' \defeq [\dirobjl{\etrue}] \concat \sigma_{a,\m{t}}'$
      
      $V \mid \chtype{a}{\sigma_a''} ; \chtype{b}{\sigma_b''} \vdash \mbranchi{m_1}{m_2}{a} \evalp{w} v$ \hfill (\textsc{EM:Cond:Recv:L})
      
      \item[Subcase:] $v_a = \efalse$, $\sigma_{a,\m{t}} : A_2$
      
      $V \mid \chtype{a}{\sigma_{a,\m{t}}'} ; \chtype{b}{\sigma_b''} \vdash m_2 \evalp{w} v$ s.t.
      $\sigma_{a,\m{t}} = \sigma_{a,\m{t}}' \concat \sigma_a'$, $\sigma_b = \sigma_b'' \concat \sigma_b'$, $\sigma_a' : A'$, $\sigma_b' : B'$ \hfill (I.H.)
      
      Let $\sigma_a'' \defeq [\dirobjl{\efalse}] \concat \sigma_{a,\m{t}}'$
      
      $V \mid \chtype{a}{\sigma_a''} ; \chtype{b}{\sigma_b''} \vdash \mbranchi{m_1}{m_2}{a} \evalp{w} v$ \hfill (\textsc{EM:Cond:Recv:L})
    \end{description}
    
    \item[Case:] The reasoning below also works for (\textsc{TM:Cond:Send:R}).
    \[\small
    \Rule{TM:Cond:Send:L}
    { \Gm \vdash e : \tbool \\
      \Gm \mid \chtype{a}{A_1} ; \chtype{b}{B} \vdash m_1 \dotsim \tau \mid \chtype{a}{A'} ; \chtype{b}{B'} \\
      \Gm \mid \chtype{a}{A_2} ; \chtype{b}{B} \vdash m_2 \dotsim \tau \mid \chtype{a}{A'} ; \chtype{b}{B'}
    }
    { \Gm \mid \chtype{a}{A_1 \echoice A_2} ; \chtype{b}{B} \vdash \mbrancho{e}{m_1}{m_2}{a} \dotsim \tau \mid \chtype{a}{A'} ; \chtype{b}{B'} }
    \]
    
    $V : \Gm$, $\Gm \vdash e : \tbool$ \hfill (assumption)
    
    $V \vdash e \evalto v_e$ for some $v_e$ s.t. $v_e : \tbool$ \hfill (\cref{The:SoundnessExp})
    
    $\sigma_a : A_1 \echoice A_2$ \hfill (assumption) \\
    
    $\sigma_a = [\dirobjr{v_a}] \concat \sigma_{a,\m{t}}$, $\sigma_{a,\m{t}} : \m{ite}(v_a,A_1,A_2)$ \hfill (inversion)
    
    \begin{description}
      \item[Subcase:] $v_a = \etrue$, $\sigma_{a,\m{t}} : A_1$
      
      $V \mid \chtype{a}{\sigma_{a,\m{t}}'} ; \chtype{b}{\sigma_b''} \vdash m_1 \evalp{w} v$ s.t.
      $\sigma_{a,\m{t}} = \sigma_{a,\m{t}}' \concat \sigma_a'$, $\sigma_b = \sigma_b'' \concat \sigma_b'$, $\sigma_a' : A'$, $\sigma_b' : B'$ \hfill (I.H.)
      
      Let $\sigma_a'' \defeq [\dirobjl{\etrue}] \concat \sigma_{a,\m{t}}'$
      
      $V \mid \chtype{a}{\sigma_a''} ; \chtype{b}{\sigma_b''} \vdash \mbranchi{m_1}{m_2}{a} \evalp{w \cdot [v_e = \etrue]} v$ \hfill (\textsc{EM:Cond:Send:L})
      
      \item[Subcase:] $v_a = \efalse$, $\sigma_{a,\m{t}} : A_2$
      
      $V \mid \chtype{a}{\sigma_{a,\m{t}}'} ; \chtype{b}{\sigma_b''} \vdash m_2 \evalp{w} v$ s.t.
      $\sigma_{a,\m{t}} = \sigma_{a,\m{t}}' \concat \sigma_a'$, $\sigma_b = \sigma_b'' \concat \sigma_b'$, $\sigma_a' : A'$, $\sigma_b' : B'$ \hfill (I.H.)
      
      Let $\sigma_a'' \defeq [\dirobjl{\efalse}] \concat \sigma_{a,\m{t}}'$
      
      $V \mid \chtype{a}{\sigma_a''} ; \chtype{b}{\sigma_b''} \vdash \mbranchi{m_1}{m_2}{a} \evalp{w \cdot [v_e = \efalse]} v$ \hfill (\textsc{EM:Cond:Recv:L})
    \end{description}
    
    \item[Case:]
    The reasoning below also works for \textsc{(TM:Sample:Recv:R)} and \textsc{(TM:Sample:Send:*)}.
    \[\small
    \Rule{TM:Sample:Recv:L}
    { \Gm \vdash e : \tdist{\tau}
    }
    { \Gm \mid \chtype{a}{\tau \wedge A'} ; \chtype{b}{B} \vdash \msamplei{e}{a} \dotsim \tau \mid \chtype{a}{A'} ; \chtype{b}{B'} }
    \]
    
    $V : \Gm$, $\Gm \vdash e : \tdist{\tau}$ \hfill (assumption)
    
    $V \vdash e \evalto d$ for some $d$ s.t. $d : \tdist{\tau}$ \hfill (\cref{The:SoundnessExp})
    
    $\sigma_a : \tau \wedge A'$ \hfill (assumption)
    
    $\sigma_a = [ \msgobjl{v} ] \concat \sigma_a'$, $v : \tau$, $\sigma_a' : A'$ \hfill (inversion)
    
    $v \in d.\mathrm{support}$ \hfill (\cref{Lem:SoundnessDist})
    
    Let $w \defeq d.\mathrm{density}(v)$
    
    $V \mid \chtype{a}{[\msgobjl{v} ]} ; \chtype{b}{[]} \vdash \msamplei{e}{a} \evalp{w} v$ \hfill (\textsc{EM:Sample:Recv:L})
    
    $\sigma_a' : A'$, $\sigma_b : B$ \hfill (assumption)
  \end{description}
\end{proof}

\begin{corollary*}[\cref{Cor:SoundnessI}]
  If
  $\cdot \mid \chtype{a}{A} ; \chtype{b}{B} \vdash m \dotsim \tau \mid \chtype{a}{\one} ; \chtype{b}{\one}$,
  $\sigma_a : A$,
  and $\sigma_b : B$,
  then
  there exist $w,v$ such that
  $\emptyset \mid \chtype{a}{\sigma_a} ; \chtype{b}{\sigma_b} \vdash m \evalp{w} v$ and
  $v : \tau$.
\end{corollary*}
\begin{proof}
  Appeal to \cref{The:SoundnessI,Lem:SoundnessValue}.
\end{proof}

\begin{theorem}[Normalization, part II]\label{The:SoundnessII}
  If
  $\Gm \mid \chtype{a}{A} ; \chtype{b}{B} \vdash m \dotsim \tau \mid \chtype{a}{A'} ; \chtype{b}{B'}$,
  $V : \Gm$,
  $A$ is $\echoice$-free,
  $B$ is $\ichoice$-free,
  $\sigma_a : A$,
  and $\sigma_b : B$,
  then
  $A'$ is $\echoice$-free,
  $B'$ is $\ichoice$-free,
  and there exist $w,v,\sigma_a'',\sigma_a',\sigma_b'',\sigma_b'$ such that
  $V \mid \chtype{a}{\sigma_a''} ; \chtype{b}{\sigma_b''} \vdash m \evalp{w} v$,
  $w>0$,
  $\sigma_a = \sigma_a'' \concat \sigma_a'$,
  $\sigma_b = \sigma_b'' \concat \sigma_b'$,
  $\sigma_a' : A'$,
  and $\sigma_b' : B'$.
\end{theorem}
\begin{proof}
  By nested induction on the derivation of $\sigma_a : A$, $\sigma_b : B$, and $\Gm \mid \chtype{a}{A} ; \chtype{b}{B} \vdash m \dotsim \tau \mid \chtype{a}{A'} ; \chtype{b}{B'}$.
  We show several nontrivial cases; others are similar to one of these cases.
  
  \begin{description}[labelindent=\parindent]
    \item[Case:]
    \[\small
    \Rule{TM:Ret}
    { \Gm \vdash e : \tau
    }
    { \Gm \mid \chtype{a}{A} ; \chtype{b}{B} \vdash \mret{e} \dotsim \tau \mid \chtype{a}{A} ; \chtype{b}{B} }
    \]
    
    $V : \Gm$, $\Gm \vdash e : \tau$ \hfill (assumption)
    
    $V \vdash e \evalto v$ for some $v$ \hfill (\cref{The:SoundnessExp})
    
    $V \mid \chtype{a}{[]} ; \chtype{b}{[]} \vdash \mret{e} \evalp{1} v$ \hfill (\textsc{EM:Ret})
    
    $A$ is $\echoice$-free, $B$ is $\ichoice$-free, $\sigma_a : A$, $\sigma_b : B$ \hfill (assumption)
    
    \item[Case:]
    \[\small
    \Rule{TM:Bnd}
    { \Gm \mid \chtype{a}{A} ; \chtype{b}{B} \vdash m_1 \dotsim \tau_1 \mid \chtype{a}{A''} ; \chtype{b}{B''} \\
      \Gm, x : \tau_1 \mid \chtype{a}{A''} ; \chtype{b}{B''} \vdash m_2 \dotsim \tau \mid \chtype{a}{A'} ; \chtype{b}{B'}
    }
    { \Gm \mid \chtype{a}{A} ; \chtype{b}{B} \vdash \mbnd{m_1}{x}{m_2} \dotsim \tau \mid \chtype{a}{A'} ; \chtype{b}{B'} }
    \]
    
    $A''$ is $\echoice$-free, $B''$ is $\ichoice$-free \hfill (I.H.)
    
    $V \mid \chtype{a}{\sigma_{a,1}} ; \chtype{b}{\sigma_{b,1}} \vdash m_1 \evalp{w_1} v_1$ s.t. $w_1>0$,
    $\sigma_a = \sigma_{a,1} \concat \sigma_{a,2}$, $\sigma_b = \sigma_{b,1} \concat \sigma_{b,2}$, $\sigma_{a,2} : A''$, $\sigma_{b,2} : B''$, and
    
    \quad $A''$ is $\echoice$-free, $B''$ is $\ichoice$-free \hfill (I.H.)
    
    $v_1 : \tau_1$ \hfill (\cref{Lem:SoundnessValue})
    
    $V[x \mapsto v_1] : (\Gm, x : \tau_1)$ \hfill (\textsc{TC:Extend})
    
    $V[x \mapsto v_1] \mid \chtype{a}{\sigma_{a,2}'} \mid \chtype{b}{\sigma_{b,2}'} \vdash m_2 \evalp{w_2} v$ s.t. $w_2>0$,
    $\sigma_{a,2} = \sigma_{a,2}' \concat \sigma_a'$, $\sigma_{b,2} = \sigma_{b,2}' \concat \sigma_b'$, $\sigma_a' : A'$, $\sigma_b' : B'$, and
    
    \quad $A'$ is $\echoice$-free, $B'$ is $\ichoice$-free \hfill (I.H.)
    
    Let $\sigma_a'' \defeq \sigma_{a,1} \concat \sigma_{a,2}'$, $\sigma_b'' \defeq \sigma_{b,1} \concat \sigma_{b,2}'$
    
    $V \mid \chtype{a}{\sigma_a''} ; \chtype{b}{\sigma_b''} \vdash \mbnd{m_1}{x}{m_2} \evalp{w_1 \cdot w_2} v$ \hfill (\textsc{EM:Bnd})
    
    $w_1 > 0, w_2 > 0$ thus $w_1 \cdot w_2 > 0$
    
    \item[Case:]
    \[\small
    \Rule{TM:Call}
    { \Sg(f) = \tau_1 \leadsto \tau \mid \chtype{a}{T_a} ; \chtype{b}{T_b} \\
      \Gm \vdash e : \tau_1
    }
    { \Gm \mid \chtype{a}{T_a[A']} ; \chtype{b}{T_b[B']} \vdash \mcall{f}{e} \dotsim \tau \mid \chtype{a}{A'} ; \chtype{b}{B'} }
    \]
    
    $V : \Gm$, $\Gm \vdash e : \tau_1$ \hfill (assumption)
    
    $V \vdash e \evalto v_1$ for some $v_1$ s.t. $v_1 : \tau_1$ \hfill (\cref{The:SoundnessExp})
    
    Let $\calD(f) = \fundec{f}{\tau_1}{\tau}{x_f}{m_f}{a}{b}$
    
    $\emptyset[x_f \mapsto v_1] : (x_f : \tau_1)$ \hfill (\textsc{TC:Extend})
    
    $x_f : \tau_1 \mid \chtype{a}{A_o} ; \chtype{b}{B_o} \vdash m_f \dotsim \tau \mid \chtype{a}{X_a} ; \chtype{b}{X_b}$, $\m{typedef}(T_a.X_a.A_o),\m{typedef}(T_b.X_b.B_o) \in \calT$ \hfill (assumption)
    
    $x_f : \tau_1 \mid \chtype{a}{[A'/X_a]A_o} ; \chtype{b}{[B'/X_b]B_o} \vdash m_f \dotsim \tau \mid \chtype{a}{A'} ; \chtype{b}{B'}$ \hfill (\cref{Lem:GuideSubst})
    
    $A = T_a[A']$, $\sigma_a : T_a[A']$ \hfill (assumption)
    
    $\sigma_a = [\foldobj] \concat \sigma_{a,\m{t}}$, $\sigma_{a,\m{t}} : [A' / X_a] A_o$ \hfill (inversion)
    
    $B = T_b[B']$, $\sigma_b : T_b[B']$ \hfill (assumption)
    
    $\sigma_b = [\foldobj] \concat \sigma_{b,\m{t}}$, $\sigma_{b,\m{t}} : [B' / X_b] B_o$ \hfill (inversion)
    
    $T_a[A']$ is $\echoice$-free, $T_b[B']$ is $\ichoice$-free $\implies$ $[A'/X_a]A_o$ is $\echoice$-free, $[B'/X_b]B_o$ is $\ichoice$-free
    
    $A'$ is $\echoice$-free, $B'$ is $\ichoice$-free \hfill (I.H.)
    
    $\emptyset[x_f \mapsto v_1] \mid \chtype{a}{\sigma_{a,\m{t}}'} ; \chtype{b}{\sigma_{b,\m{t}}'} \vdash m_f \evalp{w} v$ s.t. $w>0$, and
    $\sigma_{a,\m{t}} = \sigma_{a,\m{t}}' \concat \sigma_a'$, $\sigma_{b,\m{t}} = \sigma_{b,\m{t}}' \concat \sigma_b'$, $\sigma_a' : A'$, $\sigma_b' : B'$ \hfill (I.H.)
    
    Let $\sigma_a'' \defeq [\foldobj] \concat \sigma_{a,\m{t}}'$, $\sigma_b'' \defeq [\foldobj] \concat \sigma_{b,\m{t}}'$
    
    $V \mid \chtype{a}{\sigma_a''} ; \chtype{b}{\sigma_b''} \vdash \mcall{f}{e} \evalp{w} v$ \hfill (\textsc{EM:Call})
    
    \item[Case:]
    The reasoning below also works for \textsc{(TM:Cond:Recv:R)}.
    \[\small
    \Rule{TM:Cond:Recv:L}
    { \Gm \mid \chtype{a}{A_1} ; \chtype{b}{B} \vdash m_1 \dotsim \tau \mid \chtype{a}{A'} ; \chtype{b}{B'} \\
      \Gm \mid \chtype{a}{A_2} ; \chtype{b}{B} \vdash m_2 \dotsim \tau \mid \chtype{a}{A'} ; \chtype{b}{B'}
    }
    { \Gm \mid \chtype{a}{A_1 \ichoice A_2} ; \chtype{b}{B} \vdash \mbranchi{m_1}{m_2}{a} \dotsim \tau \mid \chtype{a}{A'} ; \chtype{b}{B'} }
    \]
    
    $(A_1 \ichoice A_2)$ $\echoice$-free implies $A_1$ $\echoice$-free, $A_2$ $\echoice$-free \hfill (assumption)
    
    $\sigma_a : A_1 \ichoice A_2$ \hfill (assumption)
    
    $\sigma_a = [\dirobjl{v_a}] \concat \sigma_{a,\m{t}}$, $\sigma_{a,\m{t}} : \m{ite}(v_a,A_1,A_2)$ \hfill (inversion)
    
    \begin{description}
      \item[Subcase:] $v_a = \etrue$, $\sigma_{a,\m{t}} : A_1$
      
      $A'$ is $\echoice$-free, $B'$ is $\ichoice$-free \hfill (I.H.) 
      
      $V \mid \chtype{a}{\sigma_{a,\m{t}}'} ; \chtype{b}{\sigma_b''} \vdash m_1 \evalp{w} v$ s.t. $w>0$, and
      $\sigma_{a,\m{t}} = \sigma_{a,\m{t}}' \concat \sigma_a'$, $\sigma_b = \sigma_b'' \concat \sigma_b'$, $\sigma_a' : A'$, $\sigma_b' : B'$ \hfill (I.H.)
      
      Let $\sigma_a'' \defeq [\dirobjl{\etrue}] \concat \sigma_{a,\m{t}}'$
      
      $V \mid \chtype{a}{\sigma_a''} ; \chtype{b}{\sigma_b''} \vdash \mbranchi{m_1}{m_2}{a} \evalp{w} v$ \hfill (\textsc{EM:Cond:Recv:L})
      
      \item[Subcase:] $v_a = \efalse$, $\sigma_{a,\m{t}} : A_2$
      
      $A'$ is $\echoice$-free, $B'$ is $\ichoice$-free \hfill (I.H.) 
      
      $V \mid \chtype{a}{\sigma_{a,\m{t}}'} ; \chtype{b}{\sigma_b''} \vdash m_2 \evalp{w} v$ s.t. $w>0$, and
      $\sigma_{a,\m{t}} = \sigma_{a,\m{t}}' \concat \sigma_a'$, $\sigma_b = \sigma_b'' \concat \sigma_b'$, $\sigma_a' : A'$, $\sigma_b' : B'$ \hfill (I.H.)
      
      Let $\sigma_a'' \defeq [\dirobjl{\efalse}] \concat \sigma_{a,\m{t}}'$
      
      $V \mid \chtype{a}{\sigma_a''} ; \chtype{b}{\sigma_b''} \vdash \mbranchi{m_1}{m_2}{a} \evalp{w} v$ \hfill (\textsc{EM:Cond:Recv:L})
    \end{description}
    
    \item[Case:]
    The reasoning below also works for \textsc{(TM:Sample:Recv:R)} and \textsc{(TM:Sample:Send:*)}.
    \[\small
    \Rule{TM:Sample:Recv:L}
    { \Gm \vdash e : \tdist{\tau}
    }
    { \Gm \mid \chtype{a}{\tau \wedge A'} ; \chtype{b}{B} \vdash \msamplei{e}{a} \dotsim \tau \mid \chtype{a}{A'} ; \chtype{b}{B'} }
    \]
    
    $V : \Gm$, $\Gm \vdash e : \tdist{\tau}$ \hfill (assumption)
    
    $V \vdash e \evalto d$ for some $d$ s.t. $d : \tdist{\tau}$ \hfill (\cref{The:SoundnessExp})
    
    $\sigma_a : \tau \wedge A'$ \hfill (assumption)
    
    $\sigma_a = [ \msgobjl{v} ] \concat \sigma_a'$, $v : \tau$, $\sigma_a' : A'$ \hfill (inversion)
    
    $v \in d.\mathrm{support}$, $d.\mathrm{density}(v) > 0$ \hfill (\cref{Lem:SoundnessDist})
    
    Let $w \defeq d.\mathrm{density}(v)$
    
    $V \mid \chtype{a}{[\msgobjl{v} ]} ; \chtype{b}{[]} \vdash \msamplei{e}{a} \evalp{w} v$ \hfill (\textsc{EM:Sample:Recv:L})
    
    $A'$ is $\echoice$-free, $B$ is $\ichoice$-free, $\sigma_a' : A'$, $\sigma_b : B$ \hfill (assumption)
  \end{description}
\end{proof}

\begin{corollary*}[\cref{Cor:SoundnessII}]
  If
  $\cdot \mid \chtype{a}{A} ; \chtype{b}{B} \vdash m \dotsim \tau \mid \chtype{a}{\one} ; \chtype{b}{\one}$,
  $A$ is $\echoice$-free,
  $B$ is $\ichoice$-free,
  $\sigma_a : A$,
  and $\sigma_b : B$,
  then
  there exist $w,v$ such that
  $\emptyset \mid \chtype{a}{\sigma_a} ; \chtype{b}{\sigma_b} \vdash m \evalp{w} v$,
  $v : \tau$,
  and $w>0$.
\end{corollary*}
\begin{proof}
  Appeal to \cref{The:SoundnessII,Lem:SoundnessValue}.
\end{proof}

To justify that our operational semantics correctly keeps track of \emph{possible} traces for running a program,
we introduce a reduction relation $V \mid \chtype{a}{\sigma_a} ; \chtype{b}{\sigma_b} \red m \evalto v$ that ignores all the information about probabilities.
\cref{Fig:RedRules} presents reduction rules for commands.

\begin{figure}
\begin{mathpar}\small
  \Rule{RM:Ret}
  { V \vdash e \evalto v
  }
  { V \mid \chtype{a}{[]} ; \chtype{b}{[]} \red \mret{e} \evalto v }
  \and
  \Rule{RM:Bnd}
  { V \mid \chtype{a}{\sigma_a} ; \chtype{b}{\sigma_b} \red m_1 \evalto v_1 \\
    V[x \mapsto v_1] \mid \chtype{a}{\sigma_a'} ; \chtype{b}{\sigma_b'} \red m_2 \evalto v_2
  }
  { V \mid \chtype{a}{\sigma_a \concat \sigma_a'} ; \chtype{b}{\sigma_b \concat \sigma_b'} \red \mbnd{m_1}{x}{m_2} \evalto v_2 }
  \and
  \Rule{RM:Sample:Recv:L}
  { V \vdash e \evalto d \\
    v \in d.\mathrm{support}
  }
  { V \mid \chtype{a}{[\msgobjl{v}]} ; \chtype{b}{[]} \red \msamplei{e}{a} \evalto v }
  \and
  \Rule{RM:Sample:Send:R}
  { V \vdash e \evalto d \\
    v \in d.\mathrm{support}
  }
  { V \mid \chtype{a}{[]} ; \chtype{b}{[\msgobjl{v}]} \red \msampleo{e}{b} \evalto v }
  \and
  \Rule{RM:Sample:Send:L}
  { V \vdash e \evalto d \\
    v \in d.\mathrm{support}
  }
  { V \mid \chtype{a}{[\msgobjr{v}]} ; \chtype{b}{[]} \red \msampleo{e}{a} \evalto v }
  \and
  \Rule{RM:Sample:Recv:R}
  { V \vdash e \evalto d \\
    v \in d.\mathrm{support}
  }
  { V \mid \chtype{a}{[]} ; \chtype{b}{[\msgobjr{v}]} \red \msamplei{e}{b} \evalto v }
  \and
  \Rule{RM:Cond:Recv:L}
  { i = \m{ite}(v_a, 1, 2) \\
    V \mid \chtype{a}{\sigma_a} ; \chtype{b}{\sigma_b} \red m_i \evalto v
  }
  { V \mid \chtype{a}{[\dirobjl{v_a}] \concat \sigma_a} ; \chtype{b}{\sigma_b} \red \mbranchi{m_1}{m_2}{a} \evalto v }
  \and
  \Rule{RM:Cond:Send:R}
  { V \vdash e \evalto v_e \\
    i = \m{ite}(v_e, 1, 2) \\
    V \mid \chtype{a}{\sigma_a} ; \chtype{b}{\sigma_b} \red m_i \evalto v
  }
  { V \mid \chtype{a}{\sigma_a} ; \chtype{b}{[\dirobjl{v_e}] \concat \sigma_b} \red \mbrancho{e}{m_1}{m_2}{b} \evalto v }
  \and
  \Rule{RM:Cond:Send:L}
  { V \vdash e \evalto v_e \\
    i = \m{ite}(v_e, 1, 2) \\
    V \mid \chtype{a}{\sigma_a} ; \chtype{b}{\sigma_b} \red m_i \evalto v
  }
  { V \mid \chtype{a}{[\dirobjr{v_e}] \concat \sigma_a} ; \chtype{b}{\sigma_b} \red \mbrancho{e}{m_1}{m_2}{a} \evalto v }
  \and
  \Rule{RM:Cond:Recv:R}
  { i = \m{ite}(v_b, 1, 2) \\
    V \mid \chtype{a}{\sigma_a} ; \chtype{b}{\sigma_b} \red m_i \evalto v
  }
  { V \mid \chtype{a}{\sigma_a} ; \chtype{b}{[\dirobjr{v_b}] \concat \sigma_b} \red \mbranchi{m_1}{m_2}{b} \evalto v }
  \and
  \Rule{RM:Call}
  { \calD(f) = \fundec{f}{\tau_1}{\tau_2}{x_f}{m_f}{a}{b} \\
    V \vdash e \evalto v_1 \\
    \emptyset[x_f \mapsto v_1] \mid \chtype{a}{\sigma_a} ; \chtype{b}{\sigma_b} \red m_f \evalto v_2
  }
  { V \mid \chtype{a}{[\foldobj] \concat \sigma_a} ; \chtype{b}{[\foldobj] \concat \sigma_b} \red \mcall{f}{e} \evalto v_2 }
\end{mathpar}
\caption{Reduction rules for commands.}\label{Fig:RedRules}
\end{figure}

\begin{theorem}\label{Lem:SoundnessIII}
  Suppose that
  $\Gm \mid \chtype{a}{A} ; \chtype{b}{B} \vdash m \dotsim \tau \mid \chtype{a}{A'} ; \chtype{b}{B'}$ and
  $V : \Gm$.
  Then for 
  any $\sigma_a, \sigma_b, v$,
  we have
  $V \mid \chtype{a}{\sigma_a} ; \chtype{b}{\sigma_b} \red m \evalto v$
  if and only if
  $V \mid \chtype{a}{\sigma_a} ; \chtype{b}{\sigma_b} \vdash m \evalp{w} v$ for some $w > 0$.
\end{theorem}
\begin{proof}
  By induction on the derivation of the evaluation judgment,
  followed by inversion on $\Gm \mid \chtype{a}{A} ; \chtype{b}{B} \vdash m \dotsim \tau \mid \chtype{a}{A'} ; \chtype{b}{B'}$.
  
  We show several nontrivial cases; others are similar to one of these cases.
  
  \begin{description}[labelindent=\parindent]
    \item[Case:]
    \begin{mathpar}\small
    \Rule{TM:Ret}
    { \Gm \vdash e : \tau
    }
    { \Gm \mid \chtype{a}{A} ; \chtype{b}{B} \vdash \mret{e} \dotsim \tau \mid \chtype{a}{A} ; \chtype{b}{B} }
    \end{mathpar}
    
    \begin{itemize}
      \item The ``if'' direction:
      \begin{mathpar}\small
      \Rule{EM:Ret}
      { V \vdash e \evalto v
      }
      { V \mid \chtype{a}{[]} ; \chtype{b}{[]} \vdash \mret{e} \evalp{1} v }  
      \end{mathpar}
      
      $V \vdash e \evalto v$ \hfill (assumption)
      
      $V \mid \chtype{a}{[]} ; \chtype{b}{[]} \red \mret{e} \evalto v$ \hfill (\textsc{RM:Ret})
      
      \item The ``only if'' direction:
      \begin{mathpar}\small
      \Rule{RM:Ret}
      { V \vdash e \evalto v
      }
      { V \mid \chtype{a}{[]} ; \chtype{b}{[]} \red \mret{e} \evalto v }  
      \end{mathpar}
      
      $V \vdash e \evalto v$ \hfill (assumption)
      
      $V \mid \chtype{a}{[]} ; \chtype{b}{[]} \vdash \mret{e} \evalp{1} v$ with $1 > 0$ \hfill (\textsc{EM:Ret})
    \end{itemize}
    
    \item[Case:]
    \begin{mathpar}\small
      \Rule{TM:Bnd}
      { \Gm \mid \chtype{a}{A} ; \chtype{b}{B} \vdash m_1 \dotsim \tau_1 \mid \chtype{a}{A''} ; \chtype{b}{B''} \\\\
        \Gm, x : \tau_1 \mid \chtype{a}{A''} ; \chtype{b}{B''} \vdash m_2 \dotsim \tau \mid \chtype{a}{A'} ; \chtype{b}{B'}
      }
      { \Gm \mid \chtype{a}{A} ; \chtype{b}{B} \vdash \mbnd{m_1}{x}{m_2} \dotsim \tau \mid \chtype{a}{A'} ; \chtype{b}{B'} }
    \end{mathpar}
    
    \begin{itemize}
      \item The ``if'' direction:
      \begin{mathpar}\small
        \Rule{EM:Bnd}
        { V \mid \chtype{a}{\sigma_{a,1}} ; \chtype{b}{\sigma_{b,1}} \vdash m_1 \evalp{w_1} v_1 \\
          V[x \mapsto v_1] \mid \chtype{a}{\sigma_{a,2}} ; \chtype{b}{\sigma_{b,2}} \vdash m_2 \evalp{w_2} v
        }
        { V \mid \chtype{a}{\sigma_{a,1} \concat \sigma_{a,2}} ; \chtype{b}{\sigma_{b,1} \concat \sigma_{b,2}} \vdash \mbnd{m_1}{x}{m_2} \evalp{w_1 \cdot w_2} v }  
      \end{mathpar}

      $w_1 \cdot w_2 > 0$ \hfill (assumption)
      
      $w_1 > 0$ and $w_2 > 0$
      
      $V \mid \chtype{a}{\sigma_{a,1}} ; \chtype{b}{\sigma_{b,1}} \vdash m_1 \evalp{w_1} v_1$,
      $\Gm \mid \chtype{a}{A} ; \chtype{b}{B} \vdash m_1 \dotsim \tau_1 \mid \chtype{a}{A''} ; \chtype{b}{B''}$, $V : \Gm$ \hfill (assumption)
      
      $V \mid \chtype{a}{\sigma_{a,1}} ; \chtype{b}{\sigma_{b,1} } \red m_1 \evalto v_1$ \hfill (I.H.)
      
      $v_1 : \tau_1$ \hfill (\cref{Lem:SoundnessValue})
      
      $V[x \mapsto v_1] : (\Gm, x : \tau_1)$ \hfill (\textsc{TC:Extend})
      
      $V[x \mapsto v_1] \mid \chtype{a}{\sigma_{a,2}} ; \chtype{b}{\sigma_{b,2}} \vdash m_2 \evalp{w_2} v$, $\Gm,x:\tau_1 \mid \chtype{a}{A''} ; \chtype{b}{B''} \vdash m_2 \dotsim \tau \mid \chtype{a}{A'} ; \chtype{b}{B'}$ \hfill (assumption)
      
      $V[x \mapsto v_1] \mid \chtype{a}{\sigma_{a,2}} ; \chtype{b}{\sigma_{b,2}} \red m_2 \evalto v$ \hfill (I.H.)
      
      $V \mid \chtype{a}{\sigma_{a,1} \concat \sigma_{a,2}} ; \chtype{b}{\sigma_{b,1} \concat \sigma_{b,2}} \red \mbnd{m_1}{x}{m_2} \evalto v$ \hfill (\textsc{RM:Bnd})
      
      \item The ``only if'' direction:
      \begin{mathpar}\small
        \Rule{RM:Bnd}
        { V \mid \chtype{a}{\sigma_{a,1}} ; \chtype{b}{\sigma_{b,1}} \red m_1 \evalto v_1 \\
          V[x \mapsto v_1] \mid \chtype{a}{\sigma_{a,2}} ; \chtype{b}{\sigma_{b,2}} \red m_2 \evalto v
        }
        { V \mid \chtype{a}{\sigma_{a,1} \concat \sigma_{a,2}} ; \chtype{b}{\sigma_{b,1} \concat \sigma_{b,2}} \red \mbnd{m_1}{x}{m_2} \evalto v }  
      \end{mathpar}
      
      $V \mid \chtype{a}{\sigma_{a,1}} ; \chtype{b}{\sigma_{b,1}} \red m_1 \evalto v_1$, $\Gm \mid \chtype{a}{A} ; \chtype{b}{B} \vdash m_1 \dotsim \tau_1 \mid \chtype{a}{A''} ; \chtype{b}{B''}$, $V : \Gm$ \hfill (assumption)
      
      $V \mid \chtype{a}{\sigma_{a,1}} ; \chtype{b}{\sigma_{b,1}} \vdash m_1 \evalp{w_1} v_1$ for some $w_1 > 0$ \hfill (I.H.)
      
      $v_1 : \tau_1$ \hfill (\cref{Lem:SoundnessValue})
      
      $V[x \mapsto v_1] : (\Gm, x : \tau_1)$ \hfill (\textsc{TC:Extend})
      
      $V[x \mapsto v_1] \mid \chtype{a}{\sigma_{a,2}} ; \chtype{b}{\sigma_{b,2}} \red m_2 \evalto v$, $\Gm, x:\tau_1 \mid \chtype{a}{A''} ; \chtype{b}{B'' } \vdash m_2 \dotsim \tau \mid \chtype{a}{A'} ; \chtype{b}{B'}$ \hfill (assumption)
      
      $V[x \mapsto v_1] \mid \chtype{a}{\sigma_{a,2}} ; \chtype{b}{\sigma_{b,2}} \vdash m_2 \evalp{w_2} v$ for some $w_2 > 0$ \hfill (I.H.)

      $V \mid \chtype{a}{\sigma_{a,1} \concat \sigma_{a,2}} ; \chtype{b}{\sigma_{b,1} \concat \sigma_{b,2}} \vdash \mbnd{m_1}{x}{m_2} \evalp{w_1 \cdot w_2} v$ \hfill (\textsc{EM:Bnd})
      
      $w_1 > 0$, $w_2 > 0$ thus $w_1 \cdot w_2 > 0$
    \end{itemize}
        
    \item[Case:]
    \begin{mathpar}\small
      \Rule{TM:Call}
      { \Sg(f) = \tau_1 \leadsto \tau \mid \chtype{a}{T_a} ; \chtype{b}{T_b} \\
        \Gm \vdash e : \tau_1
      }
      { \Gm \mid \chtype{a}{T_a[A']} ; \chtype{b}{T_b[B']} \vdash \mcall{f}{e} \dotsim \tau \mid \chtype{a}{A'} ; \chtype{b}{B'} }
      \and
      \Rule{TP:Dec}
      { \m{typedef}(T_a.X_a.A''), \m{typedef}(T_b.X_b.B'') \in \calT \\\\
        x_f : \tau_1 \mid \chtype{a}{A''} ; \chtype{b}{B''} \vdash_{\Sg} m_f \dotsim \tau \mid \chtype{a}{X_a} ; \chtype{b}{X_b} 
      }
      { \vdash_{\Sg} \fundec{f}{\tau_1}{\tau}{x_f}{m_f}{a}{b} : \tau_1 \leadsto \tau \mid \chtype{a}{T_a} ; \chtype{b}{T_b} }
    \end{mathpar}
    
    \begin{itemize}
      \item The ``if'' direction:
      \begin{mathpar}\small
        \Rule{EM:Call}
        { V \vdash e \evalto v_1 \\
          \calD(f) = \fundec{f}{\tau_1}{\tau}{x_f}{m_f}{a}{b} \\
          \emptyset[x_f \mapsto v_1] \mid \chtype{a}{\sigma_{a,1}} ; \chtype{b}{\sigma_{b,1}} \vdash m_f \evalp{w} v 
        }
        { V \mid \chtype{a}{[\foldobj] \concat \sigma_{a,1}} ; \chtype{b}{[\foldobj] \concat \sigma_{b,1}} \vdash \mcall{f}{e} \evalp{w} v }
      \end{mathpar}
      
      $\Gm \vdash e : \tau_1$, $V \vdash e \evalto v_1$, $V : \Gm$ \hfill (assumption)
      
      $v_1 : \tau_1$ \hfill (\cref{The:SoundnessExp})
      
      $\emptyset[x_f \mapsto v_1] : (x_f : \tau_1)$ \hfill (\textsc{TC:Extend})
      
      $x_f : \tau_1 \mid \chtype{a}{A''} ; \chtype{b}{B''} \vdash_\Sg m_f \dotsim \tau \mid \chtype{a}{X_a} ; \chtype{b}{X_b}$ \hfill (assumption)
      
      $x_f : \tau_1 \mid \chtype{a}{[A'/X_a]A''} ; \chtype{b}{[B'/X_b]B''} \vdash m_f \dotsim \tau \mid \chtype{a}{A'} ; \chtype{b}{B'}$ \hfill (\cref{Lem:GuideSubst})
      
      $\emptyset[x_f \mapsto v_1] \mid \chtype{a}{\sigma_{a,1}} ; \chtype{b}{\sigma_{b,1}} \vdash m_f \evalp{w} v$ \hfill (assumption)
      
      $\emptyset[x_f \mapsto v_1] \mid \chtype{a}{\sigma_{a,1}} ; \chtype{b}{\sigma_{b,1}} \red m_f \evalto v$ \hfill (I.H.)
      
      $V \mid \chtype{a}{[\foldobj] \concat \sigma_{a,1}} ; \chtype{b}{[\foldobj] \concat \sigma_{b,1}} \red \mcall{f}{e} \evalto v$ \hfill (\textsc{RM:Call})

      \item The ``only if'' direction:
      \begin{mathpar}\small
        \Rule{RM:Call}
        { V \vdash e \evalto v \\
          \calD(f) = \fundec{f}{\tau_1}{\tau}{x_f}{m_f}{a}{b} \\
          \emptyset[x_f \mapsto v_1] \mid \chtype{a}{\sigma_{a,1}} ; \chtype{b}{\sigma_{b,1}} \red m_f \evalto v
        }
        { V \mid \chtype{a}{[\foldobj] \concat \sigma_{a,1}} ; \chtype{b}{[\foldobj] \concat \sigma_{b,1}} \red \mcall{f}{e} \evalto v }
      \end{mathpar}
      
      $\Gm \vdash e : \tau_1$, $V \vdash e \evalto v$, $V : \Gm$ \hfill (assumption)
      
      $v_1 : \tau_1$ \hfill (\cref{The:SoundnessExp})
      
      $\emptyset[x_f \mapsto v_1] : (x_f : \tau_1)$ \hfill (\textsc{TC:Extend})
      
      $x_f : \tau_1 \mid \chtype{a}{A''} ; \chtype{b}{B''} \vdash_\Sg m_f \dotsim \tau \mid \chtype{a}{X_a} ; \chtype{b}{X_b}$ \hfill (assumption)
      
      $x_f : \tau_1 \mid \chtype{a}{[A'/X_a]A''} ; \chtype{b}{[B'/X_b]B''} \vdash m_f \dotsim \tau \mid \chtype{a}{A'} ; \chtype{b}{B'}$ \hfill (\cref{Lem:GuideSubst})
      
      $\emptyset[x_f \mapsto v_1] \mid \chtype{a}{\sigma_{a,1}} ; \chtype{b}{\sigma_{b,1}} \red m_f \evalto v$ \hfill (assumption)
      
      $\emptyset[x_f \mapsto v_1] \mid \chtype{a}{\sigma_{a,1}} ; \chtype{b}{\sigma_{b,1}} \vdash m_f \evalp{w} v$ for some $w > 0$ \hfill (I.H.)
      
      $V \mid \chtype{a}{[\foldobj] \concat \sigma_{a,1}} ; \chtype{b}{[\foldobj] \concat \sigma_{b,1}} \vdash m_f \evalp{w} v$ \hfill (\textsc{EM:Call})
    \end{itemize}
    
    \item[Case:]
    The reasoning below also works for \textsc{(TM:Sample:Recv:R)} and \textsc{(TM:Sample:Send:*)}.
    \begin{mathpar}\small
      \Rule{TM:Sample:Recv:L}
      { \Gm \vdash e : \tdist{\tau}
      }
      { \Gm \mid \chtype{a}{\tau \wedge A'} ; \chtype{b}{B'} \vdash \msamplei{e}{a} \dotsim \tau \mid \chtype{a}{A'} ; \chtype{b}{B'} }
    \end{mathpar}
    
    \begin{itemize}
      \item The ``if'' direction:
      \begin{mathpar}\small
        \Rule{EM:Sample:Recv:L}
        { V \vdash e \evalto d \\
          v \in d.\mathrm{support} \\
          w = d.\mathrm{density}(v)
        }
        { V \mid \chtype{a}{[\msgobjl{v}]} ; \chtype{b}{[]} \vdash \msamplei{e}{a} \evalp{w} v } 
      \end{mathpar}
      
      $V \vdash e \evalto d$, $v \in d.\mathrm{support}$ \hfill (assumption)
      
      $V \mid \chtype{a}{[\msgobjl{v}]} ; \chtype{b}{[]} \red \msamplei{e}{a} \evalto v$ \hfill (\textsc{RM:Sample:Recv:L})
      
      \item The ``only if'' direction:
      \begin{mathpar}\small
        \Rule{RM:Sample:Recv:L}
        { V \vdash e \evalto d \\
          v \in d.\mathrm{support}
        }
        { V \mid \chtype{a}{[\msgobjl{v}] } ; \chtype{b}{[]} \red \msamplei{e}{a} \evalto v }
      \end{mathpar}
      
      $\Gm \vdash e : \tdist{\tau}$, $V \vdash e \evalto d$, $V : \Gm$ \hfill (assumption)
      
      $d : \tdist{\tau}$ \hfill (\cref{The:SoundnessExp})
      
      $v \in d.\mathrm{support}$ \hfill (assumption)
      
      $v : \tau$ and $d.\mathrm{density}(v) > 0$ \hfill (\cref{Lem:SoundnessDist})
      
      $V \mid \chtype{a}{[\msgobjl{v}]} ; \chtype{b}{[]} \vdash \msamplei{e}{a} \evalp{d.\mathrm{density}(v)} v$ \hfill (\textsc{EM:Sample:Recv:L})
    \end{itemize}
    
    \item[Case:]
    The reasoning below also works for \textsc{(TM:Cond:Recv:R)}.
    \begin{mathpar}\small
    \Rule{TM:Cond:Recv:L}
    { \Gm \mid \chtype{a}{A_1} ; \chtype{b}{B} \vdash m_1 \dotsim \tau \mid \chtype{a}{A'} ; \chtype{b}{B'} \\
      \Gm \mid \chtype{a}{A_2} ; \chtype{b}{B} \vdash m_2 \dotsim \tau \mid \chtype{a}{A'} ; \chtype{b}{B'}
    }
    { \Gm \mid \chtype{a}{A_1 \ichoice A_2} ; \chtype{b}{B} \vdash \mbranchi{m_1}{m_2}{a} \dotsim \tau \mid \chtype{a}{A'} ; \chtype{b}{B'} }
    \end{mathpar}
    
    \begin{itemize}
      \item The ``if'' direction:
      \begin{mathpar}\small
        \Rule{EM:Cond:Recv:L}
        { i = \m{ite}(v_a,1,2) \\
         V \mid \chtype{a}{\sigma_{a,1}} ; \chtype{b}{\sigma_b} \vdash m_i \evalp{w} v
        }
        { V \mid \chtype{a}{[ \dirobjl{v_a} ] \concat \sigma_{a,1}} ; \chtype{b}{\sigma_b} \vdash \mbranchi{m_1}{m_2}{a} \evalp{w} v }  
      \end{mathpar}
      
      $V \mid \chtype{a}{\sigma_{a,1}} ; \chtype{b}{\sigma_b} \vdash m_i \evalp{w} v$, $\Gm \mid \chtype{a}{A_i} ; \chtype{b}{B} \vdash m_i \dotsim \tau \mid \chtype{a}{A'} ; \chtype{b}{B'}$, $V : \Gm$ \hfill (assumption)
      
      $V \mid \chtype{a}{\sigma_{a,1}} ; \chtype{b}{\sigma_b} \red m_i \evalto v$ \hfill (I.H.)
      
      $V \mid \chtype{a}{[\dirobjl{v_a}] \concat \sigma_{a,1}} ; \chtype{b}{\sigma_b} \red \mbranchi{m_1}{m_2}{a} \evalto v$ \hfill (\textsc{RM:Cond:Recv:L})
      
      \item The ``only if'' direction:
      \begin{mathpar}\small
        \Rule{RM:Cond:Recv:L}
        { i = \m{ite}(v_a,1,2) \\
          V \mid \chtype{a}{\sigma_{a,1}} ; \chtype{b}{\sigma_b} \red m_i \evalto v
        }
        {  V \mid \chtype{a}{[\dirobjl{v_a}] \concat \sigma_{a,1}} ; \chtype{b}{\sigma_b} \red \mbranchi{m_1}{m_2}{a} \evalto v  }
      \end{mathpar}
      
      $V \mid \chtype{a}{\sigma_{a,1}}; \chtype{b}{\sigma_b} \red m_i \evalto v$, $\Gm \mid \chtype{a}{A_i} ; \chtype{b}{B} \vdash m_i \dotsim \tau \mid \chtype{a}{A'} ; \chtype{b}{B'}$, $V : \Gm$ \hfill (assumption)
      
      $V \mid \chtype{a}{\sigma_{a,1}} ; \chtype{b}{\sigma_b} \vdash m_i \evalp{w} v$ for some $w>0$ \hfill (I.H.)
      
      $V \mid \chtype{a}{[\dirobjl{v_a}] \concat \sigma_{a,1}} ; \chtype{b}{\sigma_b} \vdash \mbranchi{m_1}{m_2}{a} \evalp{w} v$ \hfill (\textsc{EM:Cond:Recv:L})
    \end{itemize}
        
    \item[Case:]
    The reasoning below also works for (\textsc{TM:Cond:Send:R}).
    \begin{mathpar}\small
    \Rule{TM:Cond:Send:L}
    { \Gm \vdash e : \tbool \\
      \Gm \mid \chtype{a}{A_1} ; \chtype{b}{B} \vdash m_1 \dotsim \tau \mid \chtype{a}{A'} ; \chtype{b}{B'} \\
      \Gm \mid \chtype{a}{A_2} ; \chtype{b}{B} \vdash m_2 \dotsim \tau \mid \chtype{a}{A'} ; \chtype{b}{B'}
    }
    { \Gm \mid \chtype{a}{A_1 \echoice A_2} ; \chtype{b}{B} \vdash \mbrancho{e}{m_1}{m_2}{a} \dotsim \tau \mid \chtype{a}{A'} ; \chtype{b}{B'} }
    \end{mathpar}
    
    \begin{itemize}
      \item The ``if'' direction:
      \begin{mathpar}\small
        \Rule{EM:Cond:Send:L}
        { V \vdash e \evalto v_e \\
          i = \m{ite}(v_a, 1, 2) \\
          V \mid \chtype{a}{\sigma_{a,1}} ; \chtype{b}{\sigma_b} \vdash m_i \evalp{w} v
        }
        { V \mid \chtype{a}{[\dirobjr{v_a}] \concat \sigma_{a,1}} ; \chtype{b}{\sigma_b} \vdash \mbrancho{e}{m_1}{m_2}{a} \evalp{w \cdot [v_a=v_e]} v }  
      \end{mathpar}
      
      $w \cdot [v_a = v_e] > 0$ \hfill (assumption)
      
      $w > 0$ and $v_a = v_e$
      
      $V \mid \chtype{a}{\sigma_{a,1}} ; \chtype{b}{\sigma_b} \vdash m_i \evalp{w} v$, $\Gm \mid \chtype{a}{A_i} ; \chtype{b}{B} \vdash m_i \dotsim \tau \mid \chtype{a}{A'} ; \chtype{b}{B'}$, $V : \Gm$ \hfill (assumption)
      
      $V \mid \chtype{a}{\sigma_{a,1}} ; \chtype{b}{\sigma_b} \red m_i \evalto v$ \hfill (I.H.)
      
      $V \mid \chtype{a}{[\dirobjr{v_e}] \concat \sigma_{a,1}} ; \chtype{b}{\sigma_b} \red \mbrancho{e}{m_1}{m_2}{a} \evalto v$ \hfill (\textsc{RM:Cond:Send:L})

      \item The ``only if'' direction:
      \begin{mathpar}\small
        \Rule{RM:Cond:Send:L}
        { V \vdash e \evalto v_e \\
          i = \m{ite}(v_e, 1, 2) \\
          V \mid \chtype{a}{\sigma_{a,1}} ; \chtype{b}{\sigma_b} \red m_i \evalto v
        }
        { V \mid \chtype{a}{[\dirobjr{v_e}] \concat \sigma_{a,1}} ; \chtype{b}{\sigma_b} \red \mbrancho{e}{m_1}{m_2}{a} \evalto v }
      \end{mathpar}
      
      $V \mid \chtype{a}{\sigma_{a,1}} ; \chtype{b}{\sigma_b} \red m_i \evalto v$, $\Gm \mid \chtype{a}{A_i} ; \chtype{b}{B} \vdash m_i \dotsim \tau \mid \chtype{a}{A'} ; \chtype{b}{B'}$, $V : \Gm$ \hfill (assumption)
      
      $V \mid \chtype{a}{\sigma_{a,1}} ; \chtype{b}{\sigma_b} \vdash m_i \evalp{w} v$ for some $w > 0$ \hfill (I.H.)
      
      $V \vdash e \evalto v_e$ \hfill (assumption)
      
      $[v_e = v_e] = 1$ thus $w \cdot [v_e = v_e] = w > 0$
      
      $V \mid \chtype{a}{[\dirobjr{v_e}] \concat \sigma_{a,1}} ; \chtype{b}{\sigma_b} \vdash \mbrancho{e}{m_1}{m_2}{a} \evalp{w} v$ \hfill (\textsc{EM:Cond:Send:L})
    \end{itemize}
  \end{description}
\end{proof}

\begin{corollary}\label{Cor:SoundnessIII}
  Suppose that
  $\cdot \mid \chtype{a}{A} ; \chtype{b}{B} \vdash m \dotsim \tau \mid \chtype{a}{\one} ; \chtype{b}{\one}$,
  Then for any $\sigma_a,\sigma_b,v$,
  we have
  $\emptyset \mid \chtype{a}{\sigma_a} ; \chtype{b}{\sigma_b} \red m \evalto v$
  if and only if
  $\emptyset \mid \chtype{a}{\sigma_a} ; \chtype{b}{\sigma_b} \vdash m \evalp{w} v$ for some $w > 0$.
\end{corollary}
\begin{proof}
  Appeal to \cref{Lem:SoundnessIII}.
\end{proof}


\section{Full Development of Sound Bayesian Inference}
\label{Se:FullInference}

\paragraph{Semantic domains}
For each scalar type $\tau$, we equip it with a standard Borel space $\interp{\tau}$ on inhabitants of $\tau$:
\begin{align*}
  \interp{\tunit} & \defeq ( \{ \etriv \}, \wp(\{ \etriv \}) ) \\
  \interp{\tbool} & \defeq ( \{ \etrue, \efalse \}, \wp(\{ \etrue, \efalse \}) ) \\
  \interp{\tureal} & \defeq ( (0,1), \calB(0,1) )  \\
  \interp{\tpreal} & \defeq ( (0,\infty), \calB(0,\infty) ) \\
  \interp{\treal} & \defeq ( ({-\infty}, \infty), \calB({-\infty}, \infty) ) \\
  \interp{\tnat_n} & \defeq ( \{0,1,\cdots,n-1\}, \wp( \{ 0,1,\cdots, n-1\} ) ) \\
  \interp{\tnat} & \defeq ( \{0,1,\cdots \}, \wp( \{0, 1, \cdots \}) )
\end{align*}

For each guide type $A$, we construct a standard Borel space $\interp{A}$ on guidance traces of type $A$.
We first present a parameterized construction $\interp{A}_{(C,\calC)}$, where $(C,\calC)$ is a standard Borel
space, for guide types \emph{excluding} type-level applications:
\begin{align*}
  \interp{\one}_{(C,\calC)} & \defeq (C, \calC) \ \\
  \interp{\tau \wedge A}_{(C,\calC)} & \defeq \msgobjl{ \interp{\tau} } \otimes \interp{A}_{(C,\calC)} \\
  \interp{\tau \supset A}_{(C,\calC)} & \defeq \msgobjr{ \interp{\tau} } \otimes \interp{A}_{(C,\calC)} \\
  \interp{A \ichoice B}_{(C,\calC)} & \defeq (\dirobjl{\etrue} \otimes \interp{A}_{(C,\calC)} ) \amalg (\dirobjl{\efalse} \otimes \interp{B}_{(C,\calC)}) \\
  \interp{A \echoice B}_{(C,\calC)} & \defeq (\dirobjr{\etrue} \otimes \interp{A}_{(C,\calC)} ) \amalg (\dirobjr{\efalse} \otimes \interp{B}_{(C,\calC)})
\end{align*}
We use product measurable spaces to construct trace spaces by treating a trace as a pair of its head and tail,
and coproduct measurable spaces to join trace spaces from different branches.
Then, for every type definition $\m{typedef}(T.X.A)$, we define a function $f_T$ that maps a standard Borel space to another one:
\[
f_T \defeq \lambda(C, \calC). \foldobj \otimes \interp{A}_{(C,\calC)}.
\]
We can now construct a map $F_T$ that computes fixed points:
\[
F_T \defeq \lambda(C,\calC). \coprod_{n=0}^\infty f_T^n(C,\calC).
\]
Because standard Borel spaces are closed under countable coproducts, we know that $F_T$ is
well-defined.
We can then add the construction for type-level applications:
\begin{align*}
  \interp{T[B]}_{(C,\calC)} & \defeq F_T(\interp{B}_{(C,\calC)}).
\end{align*}

For closed programs, we usually set the continuation space $(C,\calC)$ to $(\{ [] \}, \wp(\{ [] \}))$.
Thus, we obtain the following definitions:
\begin{align*}
  \interp{\one} & \defeq (\{ [] \}, \wp(\{[]\}) ) \\
  \interp{\tau \wedge A} & \defeq \msgobjl{ \interp{\tau} } \otimes \interp{A} \\
  \interp{\tau \supset A} & \defeq \msgobjr{ \interp{\tau} } \otimes \interp{A} \\
  \interp{A \ichoice B} & \defeq (\dirobjl{\etrue} \otimes \interp{A} ) \amalg (\dirobjl{\efalse} \otimes \interp{B}) \\
  \interp{A \echoice B} & \defeq (\dirobjr{\etrue} \otimes \interp{A} ) \amalg (\dirobjr{\efalse} \otimes \interp{B}) \\
  \interp{T[B]} & \defeq F_T(\interp{B})
\end{align*}

\paragraph{Stock measures}
For each scalar type $\tau$, we equip it with a stock measure $\lambda_{\interp{\tau}}$ on its semantic domain
$\interp{\tau}$.
For nullary products $\tunit$, Booleans $\tbool$, integer rings $\tnat_n$, and natural numbers $\tnat$,
we define $\lambda_{\interp{\tau}}$ to be the counting measure, i.e., $\lambda S. |S|$.
For unit interval $\tureal$, positive real line $\tpreal$, and real line $\treal$,
we define $\lambda_{\interp{\tau}}$ to be the Lebesgue measure $\mathbf{Leb}$, i.e., the
unique measure that satisfies $\mathbf{Leb}([a,b]) = b-a$ for any interval $[a,b]$.
All these measures are $\sigma$-finite.

For each guide type $A$, we construct a stock measure $\lambda_{\interp{A}}$ on its semantic domain
$\interp{A}$.
Similar to the construction of semantic domains, we first present a parameterized construction
$\lambda_{\interp{A}; {\mu_C}}$, where $\mu_C$ is a measure on a standard Borel space $(C,\calC)$,
for guide types \emph{excluding} type-level applications:
\begin{align*}
  \lambda_{\interp{1}; {\mu_C}} & \defeq \mu_C \\
  \lambda_{\interp{\tau \wedge A};{\mu_C}} & \defeq \msgobjl{\lambda_{\interp{\tau}}} \otimes \lambda_{\interp{A};{\mu_C}} \\
  \lambda_{\interp{\tau \supset A}; {\mu_C}} & \defeq \msgobjr{\lambda_{\interp{\tau}}} \otimes \lambda_{\interp{A};{\mu_C}} \\
  \lambda_{\interp{A \ichoice B}; \mu_C} & \defeq (\dirobjl{\etrue} \otimes \lambda_{\interp{A};\mu_C}) \amalg (\dirobjl{\efalse} \otimes \lambda_{\interp{B};\mu_C}) \\
  \lambda_{\interp{A \echoice B}; \mu_C} & \defeq (\dirobjr{\etrue} \otimes \lambda_{\interp{A};\mu_C}) \amalg (\dirobjr{\efalse} \otimes \lambda_{\interp{B};\mu_C})
\end{align*}
We use product measures to construct sequencing trace spaces, and coproduct measures to join measures for
trace spaces from different branches.
Then, for every type definition $\m{typedef}(T.X.A)$, we define a function $\mathscr{f}_T$ that maps a measure to another one:
\[
\mathscr{f}_T \defeq \lambda(\mu_C). \foldobj \otimes \lambda_{\interp{A}; \mu_C}.
\]
We can now construct a map $\scrF_T$ that computes fixed points:
\[
\scrF_T \defeq \lambda(\mu_C). \coprod_{n=0}^\infty \mathscr{f}_T^n(\mu_C).
\]
Because arbitrary coproduct of measures is well-defined, we know that $\scrF_T$ is also well-defined.
We can then add the construction for type-level applications:
\begin{align*}
  \lambda_{\interp{T[B]}; {\mu_C}} & \defeq \scrF_T(\lambda_{\interp{B};{\mu_C}}).
\end{align*}

Finally, we drop the $\mu_C$ parameter to obtain the following constructions:
\begin{align*}
  \lambda_{\interp{\one}} & \defeq \text{counting measure on $\interp{\one}$} \\
  \lambda_{\interp{\tau \wedge A}} & \defeq \msgobjl{ \lambda_{\interp{\tau}} } \otimes \lambda_{\interp{A}} \\
  \lambda_{\interp{\tau \supset A}} & \defeq \msgobjr{ \lambda_{\interp{\tau}} } \otimes \lambda_{\interp{A}} \\
  \lambda_{\interp{A \ichoice B}} & \defeq (\dirobjl{\etrue} \otimes \lambda_{\interp{A}} ) \amalg (\dirobjl{\efalse} \otimes \lambda_{\interp{B}}) \\
  \lambda_{\interp{A \echoice B}} & \defeq (\dirobjr{\etrue} \otimes \lambda_{\interp{A}} ) \amalg (\dirobjr{\efalse} \otimes \lambda_{\interp{B}}) \\
  \lambda_{\interp{T[B]}} & \defeq \scrF_T(\lambda_{\interp{B}})
\end{align*}
In addition, $\lambda_{\interp{A}}$ is a $\sigma$-finite measure for any guide type $A$, because
(i) the stock measure $\lambda_{\interp{\tau}}$ for any scalar type $\tau$ is $\sigma$-finite,
(ii) binary product of $\sigma$-finite measures is still $\sigma$-finite, and
(ii) countable coproduct of $\sigma$-finite measures still $\sigma$-finite.

\paragraph{Denotation of commands}
For a well-typed closed command $m$, i.e.,
$\cdot \mid \chtype{a}{A} ; \chtype{b}{B} \vdash_{\Sg} m \dotsim \tau \mid \chtype{a}{\one} ; \chtype{b}{\one}$,
we define the \emph{density function} of $m$ as
\[
  \mathbf{P}_m(\sigma_a,\sigma_b) \defeq \begin{dcases*}
   w & if $\emptyset \mid \chtype{a}{\sigma_a} ; \chtype{b}{\sigma_b} \vdash m \evalp{w} v$ \\
   0 & otherwise
 \end{dcases*}.
\]

\begin{proposition}\label{Lem:Meaurable}
  $\mathbf{P}_m$ is measurable.
\end{proposition}
\begin{proof}
  We follow the proof strategy of~\citet{ICFP:BLG16}, where they proved measurability of density functions
  in an untyped probabilistic lambda calculus.
\end{proof}

Then, we construct a measure denotation $\interp{m}$ for $m$, by integrating $\mathbf{P}_m$ with
respect to the stock measure on the product space $\interp{A} \otimes \interp{B}$, i.e.,
\[
\interp{m}(S_{a,b}) \defeq \int_{S_{a,b}} \mathbf{P}_m(\sigma_a,\sigma_b) \lambda_{\interp{A} \otimes \interp{B}}(d (\sigma_a,\sigma_b)),
\]
where $S_{a,b}$ is a measurable set in $\interp{A} \otimes \interp{B}$.

\paragraph{Bayesian inference}
Let us fix a well-typed model program $m_{\m{m}}$ that consumes latent random variables on
a channel \id{latent} and provides observations on a channel \id{obs}, i.e.,
\[
\cdot \mid \chtype{latent}{A}; \chtype{obs}{B} \vdash_{\Sg} m_{\m{m}} \dotsim \tau_{\m{m}} \mid \chtype{latent}{\one}; \chtype{obs}{\one}.
\]
Usually, the program $m_{\m{m}}$ does \emph{not} receive any branch selections, i.e.,
$A$ is $\ichoice$-free and $B$ is $\echoice$-free.
Given a concrete observation $\sigma_o : B$,
Bayesian inference is the problem of approximating the \emph{posterior} $\interp{m_{\m{m}}}_{\sigma_o}$,
a measure \emph{conditioned} with respect to $\sigma_o$, defined by
\begin{equation*}
\interp{m_{\m{m}}}_{\sigma_o}(S_\ell) \defeq \frac{ \int_{S_\ell} \mathbf{P}_{m_{\m{m}}}(\sigma_\ell,\sigma_o) \lambda_{\interp{A}}(d \sigma_\ell) }{ \int \mathbf{P}_{m_{\m{m}}}(\sigma_\ell, \sigma_o) \lambda_{\interp{A}}(d \sigma_\ell) },
\end{equation*}
where $S_\ell$ is a measurable set in $\interp{A}$, i.e., a set of guidance traces of type $A$.

\paragraph{Guide programs}
In our system, we implement a guide program $m_{\m{g}}$ as
a coroutine that works with the model program $m_{\m{m}}$ and provides the
\id{latent} channel with guide type $A$ that $m_{\m{m}}$ consumes, i.e.,
\begin{gather*}
\cdot \mid \varnothing ; \chtype{latent}{A} \vdash_{\Sg} m_{\m{g}} \dotsim \tau_{\m{g}} \mid \varnothing ; \chtype{latent}{\one}, \\
\cdot \mid \chtype{latent}{A}; \chtype{obs}{B} \vdash_{\Sg} m_{\m{m}} \dotsim \tau_{\m{m}} \mid \chtype{latent}{\one}; \chtype{obs}{\one}.
\end{gather*}

The coroutine-based paradigm folds the model and guide programs into a single entity;
thus, during the inference, both the model and guide coroutines execute.
However, to distinguish the two measures defined by the model and guide, respectively,
we define a denotation for the guide $m_{\m{g}}$, accompanied by the model $m_{\m{m}}$ and conditioned
on a concrete observation $\sigma_o : B$, as a measure defined by
\[
  \interp{m_{\m{g}} }^{m_{\m{m}}}_{\sigma_o}(S_\ell) \defeq \int_{S_\ell} [ \mathbf{P}_{m_{\m{m}}}(\sigma_\ell,\sigma_o) \neq 0 ] \cdot \mathbf{P}_{m_{\m{g}}}(\sigma_\ell) \lambda_{\interp{A}}(d\sigma_\ell),
\]
where $S_\ell$ is a measurable set in $\interp{A}$, i.e., a set of guidance traces of type $A$.

To justify the inclusion of $[\mathbf{P}_{m_\m{m}}(\sigma_\ell,\sigma_o) \neq 0]$ in the denotation of the guide program, we consider \emph{possible} traces for a model-guide system.
A combination of traces $(\sigma_\ell,\sigma_o)$ is said to be
\emph{possible} for the model program $m_\m{m}$ and the guide program $m_\m{g}$, if
$\emptyset \mid \chtype{latent}{\sigma_\ell} ; \chtype{obs}{\sigma_o} \red m_\m{m} \evalto v_\m{m}$ and $\emptyset \mid \varnothing; \chtype{latent}{\sigma_\ell} \red m_\m{g} \evalto v_\m{g}$ for some values $v_\m{m}$ and $v_\m{g}$.

\begin{lemma*}[\cref{Lem:JustificationOfIncludingModel}]
  Suppose that $A$ is $\ichoice$-free, $B$ is $\echoice$-free, and
  \begin{align*}
    \cdot \mid \varnothing; \chtype{latent}{A} & \vdash_\Sg m_{\m{g}} \dotsim \tau_{\m{g}} \mid \varnothing; \chtype{latent}{\one}, \\
    \cdot \mid \chtype{latent}{A} ; \chtype{obs}{B} & \vdash_\Sg m_{\m{m}} \dotsim \tau_{\m{m}} \mid \chtype{latent}{\one} ; \chtype{obs}{\one}.
  \end{align*}
  Then a combination of traces $(\sigma_\ell, \sigma_o)$ is possible for the model $m_\m{m}$ and the guide $m_\m{g}$
  if and only if
  $\mathbf{P}_{m_\m{m}}(\sigma_\ell,\sigma_o) \neq 0$.
\end{lemma*}
\begin{proof}\
  \begin{itemize}
    \item The ``if'' direction:
    
    $\mathbf{P}_{m_\m{m}}(\sigma_\ell,\sigma_o) \neq 0$ \hfill (assumption)
    
    $\emptyset \mid \chtype{latent}{\sigma_\ell} ; \chtype{obs}{\sigma_o} \vdash m_\m{m} \evalp{w} v$ for some $v,w$ such that $w>0$ \hfill (definition)
    
    $\emptyset \mid \chtype{latent}{\sigma_\ell} ; \chtype{obs}{\sigma_o} \red m_\m{m} \evalto v$ \hfill (\cref{Cor:SoundnessIII})
    
    $A$ is $\ichoice$-free \hfill (assumption)
    
    $\emptyset \mid \varnothing ; \chtype{latent}{\sigma_\ell} \vdash m_\m{g} \evalp{w'} v'$ for some $v',w'$ such that $w'>0$ \hfill (\cref{Cor:SoundnessII})
    
    $\emptyset \mid \varnothing ; \chtype{latent}{\sigma_\ell} \red m_\m{g} \evalto v'$ \hfill (\cref{Cor:SoundnessIII})
    
    \item The ``only if'' direction:
    
    $\emptyset \mid \chtype{latent}{\sigma_\ell} ; \chtype{obs}{\sigma_o} \red m_\m{m} \evalto v$ for some $v$ \hfill (assumption)
    
    $\emptyset \mid \chtype{latent}{\sigma_\ell} ; \chtype{obs}{\sigma_o} \vdash m_\m{m} \evalp{w} v$ for some $w > 0 $ \hfill (\cref{Cor:SoundnessIII})
    
    $\sigma_\ell : A$, $\sigma_o : B$ \hfill (\cref{Cor:SoundnessW})
    
    $\mathbf{P}_{m_\m{m}}(\sigma_\ell,\sigma_o) = w > 0$ \hfill (definition)
  \end{itemize}
\end{proof}

\paragraph{Absolute continuity}
Recall that a measure $\mu$ is said to be \emph{absolutely continuous} with respect to a measure $\nu$,
if $\mu$ and $\nu$ are defined on the same measurable space, and
$\nu(S) \neq 0$ for every measurable set $S$ for which $\mu(S) \neq 0$.

We prove that for a model-guide pair, guide types serve as certificates for
absolute continuity.

\begin{theorem*}[\cref{The:AbsoluteContinuity}]
  Suppose that
  \begin{align*}
    \cdot \mid \varnothing; \chtype{latent}{A} & \vdash_\Sg m_{\m{g}} \dotsim \tau_{\m{g}} \mid \varnothing; \chtype{latent}{\one}, \\
    \cdot \mid \chtype{latent}{A} ; \chtype{obs}{B} & \vdash_\Sg m_{\m{m}} \dotsim \tau_{\m{m}} \mid \chtype{latent}{\one} ; \chtype{obs}{\one},
  \end{align*}
  $A$ is $\ichoice$-free, $B$ is $\echoice$-free,
  and $\sigma_o : B$ such that $\int \mathbf{P}_{m_{\m{m}}}(\sigma_\ell, \sigma_o) \lambda_{\interp{A}}(d\sigma_\ell) > 0$.
  Then
  the measure $\interp{m_{\m{m}}}_{\sigma_o}$ is absolutely continuous with respect to
  the measure $\interp{m_{\m{g}}}^{m_{\m{m}}}_{\sigma_o}$, and vice versa.
\end{theorem*}
\begin{proof}  
  We first claim that for any $\sigma_\ell \in \interp{A}$, it holds that $\mathbf{P}_{m_\m{m}}(\sigma_\ell,\sigma_o) \neq 0$ if and only if $[ \mathbf{P}_{m_{\m{m}}}(\sigma_\ell,\sigma_o) \neq 0 ] \cdot \mathbf{P}_{m_{\m{g}}}(\sigma_\ell) \neq 0$.
  \begin{description}[labelindent=\parindent]
    \item[($\implies$)]
    Assume $\mathbf{P}_{m_\m{m}}(\sigma_\ell,\sigma_o) \neq 0$.
    Thus, there exists $v_1$ such that $\emptyset \mid \chtype{latent}{\sigma_\ell}; \chtype{obs}{\sigma_o} \vdash m_\m{m} \evalp{w_1} v_1$, where $w_1 \defeq \mathbf{P}_{m_\m{m}}(\sigma_\ell,\sigma_o)$.
    %
    %
    Because $A$ is $\ichoice$-free, we can apply \cref{Cor:SoundnessII} to program $m_\m{g}$.
    Thus, there exist $w_2,v_2$ such that $\emptyset \mid \varnothing; \chtype{latent}{\sigma_\ell} \vdash m_\m{g} \evalp{w_2} v_2$ and $w_2 > 0$.
    In other words, we have $\mathbf{P}_{m_\m{g}}(\sigma_\ell) > 0$.
    Then, we conclude that $[ \mathbf{P}_{m_{\m{m}}}(\sigma_\ell,\sigma_o) \neq 0 ] \cdot \mathbf{P}_{m_{\m{g}}}(\sigma_\ell) \neq 0$.
    
    \item[($\impliedby$)]
    Assume $[ \mathbf{P}_{m_{\m{m}}}(\sigma_\ell,\sigma_o) \neq 0 ] \cdot \mathbf{P}_{m_{\m{g}}}(\sigma_\ell) \neq 0$.
    Thus, we have both $\mathbf{P}_{m_{\m{m}}}(\sigma_\ell,\sigma_o) \neq 0$ and $\mathbf{P}_{m_{\m{g}}}(\sigma_\ell) \neq 0$.
    Then, we conclude that $\mathbf{P}_{m_{\m{m}}}(\sigma_\ell,\sigma_o) \neq 0$ directly.
  \end{description}
  
  Fix a measurable set $S_\ell$ in $\interp{A}$.
  It suffices to show that $\interp{m_\m{m}}_{\sigma_o}(S_\ell) = 0$ if and only if $\interp{m_\m{g}}^{m_\m{m}}_{\sigma_o}(S_\ell) = 0$.
  We conclude by the following reasoning:
  \begin{align*}
    \interp{m_\m{m}}_{\sigma_o}(S_\ell) = 0 & \iff \int_{S_\ell} \mathbf{P}_{m_{\m{m}}}(\sigma_\ell,\sigma_o) \lambda_{\interp{A}}(d \sigma_\ell) = 0 \\
    & \iff \lambda_{\interp{A}}( \mathbf{P}_{m_\m{m}}({\cdot},\sigma_o) \cdot \mathrm{I}_{S_\ell} ) = 0 \\
    & \iff \lambda_{\interp{A}}( \{ \mathbf{P}_{m_\m{m}}({\cdot},\sigma_o) \cdot \mathrm{I}_{S_\ell} > 0 \} ) = 0 \\
    & \iff \lambda_{\interp{A}}( \{ \mathbf{P}_{m_\m{m}}({\cdot},\sigma_o) > 0 \wedge \mathrm{I}_{S_\ell} > 0 \}) = 0 \\
    & \iff \lambda_{\interp{A}}( \{ ( [  \mathbf{P}_{m_{\m{m}}}({\cdot},\sigma_o) \neq 0 ] \cdot \mathbf{P}_{m_{\m{g}}}({\cdot}) ) > 0 \wedge \mathrm{I}_{S_\ell} > 0 \}) = 0 \\
    & \iff \lambda_{\interp{A}}( \{  [ \mathbf{P}_{m_{\m{m}}}({\cdot},\sigma_o) \neq 0 ] \cdot \mathbf{P}_{m_{\m{g}}}({\cdot})  \cdot \mathrm{I}_{S_\ell} > 0 \}) = 0 \\
    & \iff \lambda_{\interp{A}}( [ \mathbf{P}_{m_{\m{m}}}({\cdot},\sigma_o) \neq 0 ] \cdot \mathbf{P}_{m_{\m{g}}}({\cdot})  \cdot \mathrm{I}_{S_\ell} ) = 0 \\
    & \iff \int_{S_\ell} [ \mathbf{P}_{m_{\m{m}}}(\sigma_\ell,\sigma_o) \neq 0 ] \cdot \mathbf{P}_{m_{\m{g}}}(\sigma_\ell) \lambda_{\interp{A}}(d \sigma_\ell) = 0 \\
    & \iff \interp{m_{\m{g}}}^{m_{\m{m}}}_{\sigma_o}(S_\ell) = 0.
  \end{align*}
\end{proof}

\paragraph{Importance sampling (IS)}
Recall the operational rule below for a single step in the IS algorithm:
given a model program $m_{\m{m}}$, a guide program $m_{\m{g}}$, and a concrete observation $\sigma_o$,
IS performs joint execution of the two programs to draw a sample $\sigma_\ell$ with density $w_{\m{g}}$
and compute $\frac{w_{\m{m}}}{w_{\m{g}}}$ as the importance of $\sigma_\ell$.
\begin{mathpar}
  \inferrule
  { \emptyset \mid \varnothing ; \chtype{latent}{\sigma_\ell} \vdash m_{\m{g}} \evalp{w_{\m{g}}} \_ \\
    \emptyset \mid \chtype{latent}{\sigma_\ell} ; \chtype{obs}{\sigma_o} \vdash m_{\m{m}} \evalp{w_\m{m}} \_
  }
  { m_{\m{g}}; m_{\m{m}}; \sigma_o \vdash_{\textsc{is}}^{w_{\m{g}}}  \tuple{\sigma_\ell, \sfrac{w_{\m{m}}}{w_{\m{g}}} } }
\end{mathpar}
Define a density function $\textsc{is}(\sigma_\ell) \defeq w_\m{g} \cdot \frac{w_\m{m}}{w_\m{g}} = w_\m{m}$ on
the space $\interp{A}$ of guidance traces for the computation of IS.
Then, the measure defined by IS can be defined as $\mu_{\textsc{is}}(S_\ell) \defeq \int_{S_\ell} \textsc{is}(\sigma_\ell) \lambda_{\interp{A}}(d \sigma_\ell)$.

\begin{lemma}\label{Lem:SoundIS}
  Suppose that
  \begin{align*}
    \cdot \mid \varnothing; \chtype{latent}{A} & \vdash_\Sg m_{\m{g}} \dotsim \tau_{\m{g}} \mid \varnothing; \chtype{latent}{\one}, \\
    \cdot \mid \chtype{latent}{A} ; \chtype{obs}{B} & \vdash_\Sg m_{\m{m}} \dotsim \tau_{\m{m}} \mid \chtype{latent}{\one} ; \chtype{obs}{\one},
  \end{align*}
  $A$ is $\ichoice$-free, $B$ is $\echoice$-free,
  and $\sigma_o : B$ such that $\int \mathbf{P}_{m_{\m{m}}}(\sigma_\ell, \sigma_o) \lambda_{\interp{A}}(d\sigma_\ell) > 0$.
  Then $\mu_\textnormal{\textsc{is}} \propto \interp{m_{\m{m}}}_{\sigma_o}$.
\end{lemma}
\begin{proof}
  By \cref{The:AbsoluteContinuity}, we know that the posterior $\interp{m_{\m{m}}}_{\sigma_o}$ is absolutely continuous with respect to
  $\interp{m_{\m{g}}}^{m_{\m{m}}}_{\sigma_o}$.
  Thus, $\interp{m_{\m{m}}}_{\sigma_o}$ is also absolutely continuous with respect to $\mu_{\textsc{is}}$.
  In other words, the density function $\textsc{is}({\cdot})$ is positive at all possible latent variables
  $\sigma_\ell$ in the posterior.
  Then, for any $S_\ell$ such that $\interp{m_{\m{m}}}_{\sigma_o}(S_\ell) \neq 0$,
  we have
  \[
  \interp{m_{\m{m}}}_{\sigma_o}(S_\ell) = \frac{ \int_{S_\ell} \mathbf{P}_{m_{\m{m}}}(\sigma_\ell,\sigma_o) \lambda_{\interp{A}}(d \sigma_\ell) }{ \int \mathbf{P}_{m_{\m{m}}}(\sigma_\ell, \sigma_o) \lambda_{\interp{A}}(d \sigma_\ell) } = \frac{ \int_{S_\ell} \textsc{is}(\sigma_\ell) \lambda_{\interp{A}}(d \sigma_\ell) }{ \int \mathbf{P}_{m_{\m{m}}}(\sigma_\ell, \sigma_o) \lambda_{\interp{A}}(d \sigma_\ell) } = \frac{\mu_{\textsc{is}}(S_\ell)}{ \int \mathbf{P}_{m_{\m{m}}}(\sigma_\ell, \sigma_o) \lambda_{\interp{A}}(d \sigma_\ell) },
  \]
  and we conclude by the fact the the denominator is a constant.
\end{proof}

\paragraph{Variational inference (VI)}
Recall that we parameterize the guide program $m_{\m{g}, \theta}$ by a vector $\theta \in \Theta$ of parameters,
and use KL divergence as the distance metric, which is defined by
\[
\mathrm{KL}(\mu \parallel \nu) \defeq \int p_{\mu}(\sigma_\ell) (\log p_{\mu}(\sigma_\ell) - \log p_{\nu}(\sigma_\ell)) \lambda_{\interp{A}}(d\sigma_\ell),
\] 
where $\mu$ and $\nu$ are measures on a space $\interp{A}$ of guidance traces of type $A$ with
densities $p_{\mu}$ and $p_{\nu}$, respectively,
and $\mu$ is absolutely continuous with respect to $\nu$.
The rule below formulates the computation of KL divergence for a specific $\theta$,
via joint execution of the two programs.
\begin{mathpar}
  \inferrule
  { \emptyset \mid \varnothing ; \chtype{latent}{\sigma_\ell} \vdash m_{\m{g},\theta} \evalp{w_{\m{g}}} \_ \\
    \emptyset \mid \chtype{latent}{\sigma_\ell} ; \chtype{obs}{\sigma_o} \vdash m_{\m{m}} \evalp{w_\m{m}} \_
  }
  { m_{\m{g},\theta}; m_{\m{m}}; \sigma_o \vdash_{\textsc{vi}}^{w_{\m{g}}}  \tuple{\sigma_\ell,  \log w_{\m{m}} - \log w_{\m{g}} } }
\end{mathpar}
The rule can be seen as defining a map $\sigma_\ell \mapsto w_{\m{g}} \cdot (\log w_{\m{m}} - \log w_{\m{g}})$,
which is the integrand of the divergence $\mathrm{KL}(\interp{m_{\m{g},\theta}}^{m_{\m{m}}}_{\sigma_o} \parallel \interp{m_{\m{m}}}_{\sigma_o})$.

\begin{lemma}\label{Lem:SoundVI}
  Suppose that
  \begin{align*}
    \cdot \mid \varnothing; \chtype{latent}{A} & \vdash_\Sg m_{\m{g},\theta} \dotsim \tau_{\m{g}} \mid \varnothing; \chtype{latent}{\one}, \\
    \cdot \mid \chtype{latent}{A} ; \chtype{obs}{B} & \vdash_\Sg m_{\m{m}} \dotsim \tau_{\m{m}} \mid \chtype{latent}{\one} ; \chtype{obs}{\one},
  \end{align*}
  $A$ is $\ichoice$-free, $B$ is $\echoice$-free,
  and $\sigma_o : B$ such that $\int \mathbf{P}_{m_{\m{m}}}(\sigma_\ell, \sigma_o) \lambda_{\interp{A}}(d\sigma_\ell) > 0$.
  Then, $\mathrm{KL}(\interp{m_{\m{g},\theta}}^{m_{\m{m}}}_{\sigma_o} \parallel \interp{m_{\m{m}}}_{\sigma_o})$
  is well-defined.
\end{lemma}
\begin{proof}
  By \cref{The:AbsoluteContinuity}, we know that $\interp{m_{\m{g},\theta}}^{m_{\m{m}}}_{\sigma_o}$ is absolutely continuous with respect to $\interp{m_{\m{m}}}_{\sigma_o}$.
  Thus, the KL divergence used in VI is well-defined.
\end{proof}

\paragraph{Markov-Chain Monte Carlo (MCMC)}
We focus on \emph{Metropolis-Hastings} (MH), which constructs the transition
kernel from a \emph{proposal} subroutine, and a probabilistic decision subroutine that either
accepts the proposed random sample, or rejects it and keeps the old one.
To implement proposal subroutines in our system, we extend the core calculus such that guidance
traces can be used as first-class data.
Then we implement the proposal subroutine as a procedure $g$ whose argument is a guidance trace
on the channel for latent random variables.
The operational rule below formulates a single step in the MH algorithm;
given a proposal procedure $g$, a model program $m_\m{m}$, a concrete observation $\sigma_o$,
and the current latent trace $\sigma_\ell$, MH first performs joint execution of $\mcall{g}{\sigma_\ell}$
and $m_\m{m}$ to generate a new latent trace $\sigma_\ell'$ with density $w_{\m{fwd}}$, 
and then uses the new $\sigma_\ell'$ and the old $\sigma_\ell$ to calculate a \emph{backward} density $w_{\m{bwd}}$.
MH then uses these densities to compute an \emph{acceptance ratio} $\alpha$, and then accepts
the new sample $\sigma_\ell'$ with probability $\alpha$.
\begin{mathpar}
  \inferrule
  { \emptyset \mid \varnothing ; \chtype{latent}{\sigma_\ell'} \vdash \mcall{g}{\sigma_\ell} \evalp{w_{\m{fwd}}} \_ \\
    \emptyset \mid \chtype{latent}{\sigma_\ell'} ; \chtype{obs}{\sigma_o} \vdash m_\m{m} \evalp{w'_\m{m}} \_ \\
    \emptyset \mid \varnothing ; \chtype{latent}{\sigma_\ell} \vdash \mcall{g}{\sigma_\ell'} \evalp{w_{\m{bwd}}} \_ \\
    \emptyset \mid \chtype{latent}{\sigma_\ell} ; \chtype{obs}{\sigma_o} \vdash m_\m{m} \evalp{w_\m{m}} \_ \\
    \alpha = \min(1, \frac{w'_\m{m} \cdot w_{\m{bwd}}}{w_\m{m} \cdot w_{\m{fwd}}})
  }
  { g ;m_\m{m}; \sigma_o \vdash_{\textsc{mh}} \sigma_\ell \xRightarrow{w_{\m{fwd}} \cdot \alpha}  \sigma_\ell' }
\end{mathpar}
The MH algorithm specifies a transition density $f_1(\sigma_\ell, \sigma_\ell') \defeq w_{\m{fwd}} \cdot \alpha$ and
an unchanged density $f_2(\sigma_\ell, \sigma_\ell') \defeq w_{\m{fwd}} \cdot (1-\alpha)$.
Then we can use the two density functions to construct the MH kernel
$\kappa_{\textsc{mh}}(\sigma,S) \defeq \int_S (f_1(\sigma,\sigma') + f_2(\sigma,\sigma') \cdot [\sigma \in S]) \lambda_{\interp{A}}(d \sigma')$.

\begin{lemma}\label{Lem:SoundMH}
  Suppose that
  \begin{align*}
    \id{old} : |A| \mid \varnothing; \chtype{latent}{A} & \vdash_\Sg m_\m{g} \dotsim \tau_{\m{g}} \mid \varnothing; \chtype{latent}{\one}, \\
    \cdot \mid \chtype{latent}{A} ; \chtype{obs}{B} & \vdash_\Sg m_{\m{m}} \dotsim \tau_{\m{m}} \mid \chtype{latent}{\one} ; \chtype{obs}{\one},
  \end{align*}
  $m_\m{g}$ is the procedure body of the proposal program $g$, which has a single parameter $\id{old}$ with
  type $|A|$ that describes first-class guidance traces,
  $A$ is $\ichoice$-free, $B$ is $\echoice$-free,
  and $\sigma_o : B$ such that $\int \mathbf{P}_{m_{\m{m}}}(\sigma_\ell, \sigma_o) \lambda_{\interp{A}}(d\sigma_\ell) > 0$.
  Then the posterior $\interp{m_{\m{m}}}_{\sigma_o}$ is stationary for the kernel $\kappa_{\textnormal{\textsc{mh}}}$.
\end{lemma}
\begin{proof}
  We can extend \cref{The:AbsoluteContinuity} by allowing the environments to contain values and
  then using \cref{The:SoundnessII} instead of \cref{Cor:SoundnessII} in
  the proof.
  Then, we derive that for any latent variables $\sigma$,
  the posterior $\interp{m_{\m{m}}}_{\sigma_o}$ is absolutely continuous with respect to
  the measure $\kappa_{\textsc{mh}}(\sigma,{\cdot})$;
  that is, the guide program is able to sample any latent variables in the posterior,
  no matter what the current sample is.
  We then conclude by applying the Metropolis-Hastings-Green theorem~\cite{BIOMETRIKA:Green95}.
\end{proof}

\fi

\end{document}